\documentclass[11pt]{article}

\usepackage{ifthen}
\usepackage{mdwlist}
\usepackage{fullpage}
\usepackage{amsmath,amsfonts,amsthm}
\usepackage{bm}
\usepackage{bbm}
\usepackage{enumitem}
\usepackage{graphicx}
\usepackage{xspace}
\usepackage{verbatim}
\usepackage{algorithm}
\usepackage{algorithmic}
\usepackage{color}
\newcommand{\ind}{\mathbbm{1}}

\newtheorem{theorem}{Theorem}{\bfseries}{\itshape}
\newtheorem{inftheorem}[theorem]{Informal Theorem}{\bfseries}{\itshape}
\newtheorem{corollary}[theorem]{Corollary}{\bfseries}{\itshape}
{\bfseries}{\itshape}
{\bfseries}{\itshape}
\newtheorem{definition}[theorem]{Definition}{\bfseries}{\itshape}
\newtheorem{proposition}[theorem]{Proposition}{\bfseries}{\itshape}
{\bfseries}{\itshape}
\newtheorem{lemma}[theorem]{Lemma}{\bfseries}{\itshape}
{\bfseries}{\em}
\newtheorem{claim}[theorem]{Claim}{\bfseries}{\em}

\newenvironment{varthm}[1][Theorem]{\begin{trivlist}
\item[\hskip \labelsep {\bfseries #1}]}{\end{trivlist}}

\newenvironment{prevproof}[2]{\noindent {\em {Proof of {#1}~\ref{#2}:}}}{$\Box$\vskip \belowdisplayskip}

\newcommand{\poly}{\text{poly}}

\newcommand{\costasnote}[1]{{\color{red}{#1}}}
\newcommand{\costasfootnote}[1]{{\footnote{\color{red}{{\bf Costas says:}~#1}}}}
\newcommand{\yangnote}[1]{{\color{blue}{#1}}}
\newcommand{\notshow}[1]{{}}

\newcommand{\angler}[1]{\langle #1 \rangle}

\begin{document}
\title {Extreme-Value Theorems for Optimal Multidimensional Pricing}
\author {Yang Cai\footnote{Work done while the author was a student at MIT, supported by NSF Awards CCF-0953960 (CAREER) and CCF-1101491.}\\
Computer Science, McGill University \\
cai@cs.mcgill.ca\\
\and
Constantinos Daskalakis\footnote{Supported by a Sloan Foundation Fellowship, a Microsoft Research Faculty Fellowship, and NSF Awards CCF-0953960 (CAREER) and CCF-1101491.}\\
EECS, MIT\\
costis@csail.mit.edu\\
}
\addtocounter{page}{-1}
\maketitle

\begin{abstract}
We provide a near-optimal, computationally efficient algorithm for the {unit-demand pricing problem}, where a seller wants to price $n$ items to optimize revenue against a unit-demand buyer whose values for the items are independently drawn from known distributions. For any chosen accuracy $\epsilon>0$ and item values bounded in $[0,1]$, our algorithm achieves revenue that is optimal up to an additive error of at most~$\epsilon$, in polynomial time. For values sampled from Monotone Hazard Rate (MHR) distributions, we achieve a $(1-\epsilon)$-fraction of the optimal revenue in polynomial time, while for values sampled from regular distributions the same revenue guarantees are achieved in quasi-polynomial time.

 
Our algorithm for bounded distributions applies probabilistic techniques to understand the statistical properties of revenue distributions, obtaining a reduction in the search space of the algorithm through dynamic programming. Adapting this approach to MHR and regular distributions requires the proof of novel extreme value theorems for such distributions. 

As a byproduct, our techniques establish structural properties of {approximately-optimal and near-optimal solutions. We show that, when the buyer's values are independently distributed according to MHR distributions, pricing all items at the same price achieves a constant fraction of the optimal revenue.} Moreover, for all $\epsilon >0$, at most $g(1/\epsilon)$ distinct prices suffice to obtain a $(1-\epsilon)$-fraction of the optimal revenue, where $g(1/\epsilon)$ is a quadratic function of $1/\epsilon$ that does not depend on the number of items. Similarly, for all $\epsilon>0$ and $n>0$, at most $g(1/\epsilon \cdot \log n)$ distinct prices suffice if the values are independently distributed according to regular distributions, where $g(\cdot)$ is a polynomial function. Finally, if the values are i.i.d. from some MHR distribution, we show that, as long as the number of items is a sufficiently large function of $1/\epsilon$, a single price suffices to achieve a $(1-\epsilon)$-fraction of the optimal revenue.
\end{abstract}

\thispagestyle{empty}
\newpage
\section{Introduction} \label{sec:intro}

We study the following {pricing problem}. A seller has $n$ items to sell  to a buyer who is looking to buy a single item. The seller wants to maximize profit from the sale, leveraging stochastic knowledge she has about the buyer to achieve this goal. In particular, we assume that the seller has access to a distribution $\cal F$ from which the values $(v_1,\ldots,v_n)$ of the buyer for the items are drawn. Given this information, the seller wants to compute prices $p_1,\ldots,p_n$ for the items  to maximize her revenue, assuming that the buyer is {\em quasi-linear}---i.e. will buy the item~$i$ maximizing $v_i - p_i$, as long as this difference is positive. That is, the seller's expected revenue from a price vector $P=(p_1,\ldots,p_n)$ is
\begin{align}
{\cal R}_P=\sum_{i=1}^n p_i \cdot \Pr \big[ (i = \arg \max\{v_j - p_j\})~\wedge~(v_i-p_i\ge0)\big], \label{eq:objective}
\end{align}
where we assume that the $\arg \max$ breaks ties in favor of a single item, when there are multiple maximizers. A more sophisticated seller could try to improve her revenue by pricing lotteries over items, that is also price randomized allocations of items~\cite{BriestCKW10}, albeit this may be less natural than item pricing, and we will not study it extensively in this paper.

While our problem  has a simple statement, it exhibits rich behavior depending on the nature of $\cal F$. For example, if $\cal F$ assigns the same value to all the items with probability $1$, i.e. when the buyer always values all items equally, the problem becomes {\em single-dimensional}. In this setting, it is clear that lotteries  do not improve  the revenue and that the optimal price vector can assign the same price to all the items. This observation is a special case of the more general, celebrated result of Myerson~\cite{Myerson81} on optimal mechanism design, i.e. the multi-buyer version of our problem, and generalizations thereof. Myerson's result provides a closed-form solution to the multi-buyer problem in a single sweep that covers many settings, but only works under the same limiting assumption that every buyer is single-dimensional, i.e. receives the same value from all the items. (More generally, every buyer receives the same value from all outcomes of the mechanism that provide her service.)

Following Myerson, a large body of research in both Economics and Engineering has been devoted to extending his result to the {\em multi-dimensional setting}, where the buyers' values come from general distributions.  And, while there has been sporadic progress (see survey~\cite{ManelliV07} and its references), an optimal multi-dimensional mechanism, generalizing Myerson's result, does not seem to be in sight. Indeed, there is not even an optimal solution known for the single-buyer item pricing problem.  Even the ostensibly easier version of that problem, where the values of the buyer for the items are independent and supported on a set of cardinality $2$ is unresolved.\footnote{Incidentally, the problem is trickier than it originally seems, and various intuitive properties that one would expect from the optimal solution fail to hold. See Appendix~\ref{sec:counter} for an interesting example.} Our main contribution in this paper is to develop {near-optimal polynomial-time algorithms for this problem,} when the buyer's values for the items are independent.


\subsection{Main Results}

We partition our results into algorithmic and structural. The former provide efficient algorithmic procedures for computing near-optimal price vectors. The latter shed light into the structure of optimal solutions.

\paragraph{Algorithmic Results.} Previous work on the item pricing problem has provided constant factor approximation algorithms. The best known polynomial-time algorithm obtains revenue that is at least $1/2$ of the revenue of the optimal price vector~\cite{ChawlaHK07,ChawlaHMS10}. We discuss these approaches in Section~\ref{sec:related}, also noting that they are limited to constant factor approximations. We are aiming instead for item pricing mechanisms that come arbitrarily close to the optimal revenue, obtaining the following results. Their proofs are overviewed in Sections~\ref{sec: cover view of the problem} through~\ref{sec:regular}, while complete details are provided in the appendix.

\begin{theorem}[{\bf Main Algorithmic Result:} Additive PTAS for Bounded Distributions]\label{thm:additive ptas}
Suppose that the values of the buyer for $n$ items are independent and normalized to lie in $[0,1]$. Then, for all $\epsilon >0$, {there exists an algorithm that computes a price vector whose revenue is within an additive $\epsilon$ of optimal, and whose running time is polynomial in $n^{{\log^{3} 1/\epsilon\over \epsilon^{4}}}$.}
\end{theorem}

\begin{theorem}[General Algorithm]\label{thm:general algorithm}
Suppose that the values of the buyer for $n$ items are independent and supported on some interval $[u_{min},r \cdot u_{min}]$ for some $u_{min} >0$ and $r \ge 1$. Then, for all $\epsilon >0$, there is an algorithm that computes a price vector whose revenue is at least a $(1-\epsilon)$-fraction of the optimal revenue, and whose running time is polynomial in {$\max\left\{n^{\log^{11} r \cdot  \log \log r}, n^{\log^{3} r \cdot\log {1 \over \epsilon}\over \epsilon^{8}}\right\}$}\notshow{$n^{{\rm poly}({1 \over \epsilon}, \log r)}$}.\footnote{We note that a natural approach for computing approximately optimal price vectors is to discretize the domain of price vectors and show that searching over the discretized domain suffices for approximating the optimal revenue. However, a straightforward application of the discretizations  proposed by Nisan~\cite{ChawlaHK07} and Hartline and Koltun~\cite{HartlineK05} to our problem would result in running time of  $\left({1 \over \epsilon} \log r\right)^{O(n)}$. The purpose of our theorem is to remove the exponential dependence of the running time on the number of items $n$.}
\end{theorem}

\begin{theorem}[Multiplicative PTAS for MHR Distributions]\label{thm:ptas mhr}
There is a Polynomial-Time Approximation Scheme\footnote{A Polynomial-Time Approximation Scheme (PTAS) is a family of algorithms $\{\mathcal{A}_{\epsilon}\}_{\epsilon}$, indexed by the accuracy parameter $\epsilon >0$, such that for every fixed $\epsilon>0$, $\mathcal{A}_{\epsilon}$ runs in time polynomial in the size of its input. See Section~\ref{sec:prelim} for a formal definition.} for computing an optimal price vector, when the values of the buyer are independently drawn from Monotone Hazard Rate distributions.\footnote{Monotone Hazard Rate (MHR)  distributions are a commonly studied class of distributions that contain such familiar distributions as the Uniform, Gaussian and Exponential distributions. See Section~\ref{sec:prelim} for a formal definition.} 

For any accuracy $\epsilon >0$, the algorithm runs in time polynomial {in $n^{{1 \over \epsilon^{7}}}$},  and outputs a price vector whose revenue is at least a $(1-\epsilon)$-fraction of the optimal revenue, where $n$ is the number of items.
\end{theorem}

\begin{theorem}[Multiplicative Quasi-PTAS for Regular Distributions]\label{thm:quasi ptas regular}
There is a Quasi-Polynomial-Time Approximation Scheme\footnote{A Quasi-Polynomial-Time Approximation Scheme (Quasi-PTAS) is a family of algorithms $\{\mathcal{A}_{\epsilon}\}_{\epsilon}$, indexed by the accuracy parameter $\epsilon >0$, such that for every fixed $\epsilon>0$, $\mathcal{A}_{\epsilon}$ runs in time quasi-polynomial in the size of its input. See Section~\ref{sec:prelim} for formal definition.} for computing an optimal price vector, when the values of the buyer are independent and drawn from regular distributions.\footnote{Regular distributions are another widely studied class of distributions that contain MHR distributions. See Section~\ref{sec:prelim} for a formal defintion.} 

For any accuracy $\epsilon >0$, the algorithm runs in time {polynomial in $\max\left\{ n^{\log^{11} {n\over \epsilon} \cdot  \log \log {n\over\epsilon}}, n^{ {\log^3 {n\over\epsilon} \cdot  \log {1\over\epsilon} \over {\epsilon}^{8}}} \right\}$},  and outputs a price vector whose revenue is at least a $(1-\epsilon)$-fraction of the optimal revenue, where $n$ is the number of items.
\end{theorem}

\notshow{\yangnote{In fact, Theorem~\ref{thm:general algorithm} easily implies
Theorem~\ref{thm:additive ptas}. When the values are distributed in $[0,1]$, by ignoring the values that are smaller than $\epsilon$, we face a new pricing problem whose values are distributed in $[\epsilon, 1]$. It is not hard to see that for any price vector, the revenue achieved under these two distributions will differ for at most $\epsilon$. If we use the general algorithm to solve the new pricing problem, the solution will be an additive $O(\epsilon)$ approximation for the original problem.(Yang: Reviewer 1 wants this paragraph. Shall we add a graph to show that we prove our theorems by reducing the problems to the general problem.)} }

\paragraph{Discussion of Algorithmic Results.} {Prior to our work, there were no (near-)optimal algorithms known for multi-dimensional auction problems without special structure.} In particular, only constant factor approximation algorithms were known for the item pricing problems addressed by Theorems~\ref{thm:additive ptas} through~\ref{thm:quasi ptas regular}. (For an extensive discussion of related work, we refer the reader to Section~\ref{sec:related}.) Our results are the first to obtain near-optimal solutions for these problems in polynomial time. We view the main contribution of our results not to be the practicality of our algorithms, but establishing that there is no lingering constant inapproximability results for item pricing. In particular, our results show that, for any desired accuracy $\epsilon>0$, there are polynomial-time algorithms that compute $\epsilon$-optimal solutions. Complemented with the {\tt NP}-hardness result for the item pricing problem discussed in Section~\ref{sec:related}, what is left open by our work is obtaining faster near-optimal algorithms.

 \paragraph{Structural Results.} Our algorithms are obtained by studying the distribution of the optimal revenue, as a function of the buyer's values (which are random) and the optimal price vector (which is unknown), as overviewed in Section~\ref{sec:techniques}. As a byproduct of our techniques, we deduce the following structural properties of optimal solutions, whose proofs are given in Appendix~\ref{sec:proofs of structural}. {Theorem~\ref{thm:single price constant factor} states that, when the values are independently distributed according to monotone hazard rate distributions,  then pricing all items at the same price guarantees a constant fraction of the optimal revenue. Theorem~\ref{thm:constant prices suffice} generalizes this to showing that only the desired approximation~$\epsilon$ dictates the number of distinct prices that are necessary to achieve a $(1-\epsilon)$-fraction of the optimal revenue, and {not the number of items or the size of the support of the distributions}, as long as they are monotone hazard rate. Theorem~\ref{thm:logn prices suffice} generalizes this result to a mild dependence on $n$ for regular distributions. 

\begin{theorem}[{\bf Structural 1 (MHR):} Constant Factor Approximation from a Single Price] \label{thm:single price constant factor} If the buyer's values for the items are independently distributed according to MHR distributions, there exists a price $p$ such that pricing all items at $p$ guarantees a constant fraction of the optimal revenue. Price $p$ can be computed efficiently from the value distributions.
\end{theorem}}

\begin{theorem}[{\bf Structural 2 (MHR):}  A Constant Number of Distinct Prices Suffice for Near-Optimal Revenue] \label{thm:constant prices suffice} There exists a quasi-quadratic\footnote{A function $g: \mathbb{R}_+ \longrightarrow \mathbb{R}_+$ is {\em quasi-quadratic} iff it satisfies $g(x)=O(x^2 \log^c x)$, for some absolute constant $c>0$. For the meaning of the $O(\cdot)$ notation please refer to Section~\ref{sec:prelim}.} function $g(\cdot)$ such that, for all $\epsilon>0$ and all $n>0$, $g(1/\epsilon)$ distinct prices suffice to achieve a $(1-\epsilon)$-fraction of the optimal revenue, when the buyer's values for the $n$ items are independently distributed according to MHR distributions. These distinct prices can be computed efficiently from the value distributions.
\end{theorem}
\begin{theorem}[{\bf Structural 3 (Regular):}  A Polylogarithmic Number of Distinct Prices Suffice for Near-Optimal Revenue] \label{thm:logn prices suffice} There exists a polynomial function $g(\cdot)$ such that, for all $\epsilon>0$ and $n>0$, $g(1/\epsilon \cdot \log n)$ distinct prices suffice to achieve a $(1-\epsilon)$-fraction of the optimal revenue, when the buyer's values for the $n$ items are independently distributed according to regular distributions. These  prices can be computed efficiently from the value distributions.
\end{theorem}

Finally, it seems intuitive that, when the value distributions are not widely different, a single price might suffice for extracting a $(1-\epsilon)$-fraction of the optimal revenue, as long as there is a sufficient number of items for sale. We show such a result for the case where the buyer's values are i.i.d. according to a MHR distribution. See Appendix~\ref{sec:MHR iid} for the proof of this theorem.

\begin{theorem}[{\bf Structural 4 (i.i.d. MHR):} A Single Price Suffices for Near-Optimal Revenue] \label{thm:structural theorem 3}There is a function $g(\cdot)$ such that, for all $\epsilon>0$, if the number of items is larger than $g(1/\epsilon)$ then pricing all the items at the same price obtains a $(1-\epsilon)$-fraction of the optimal revenue, if the buyer's values are i.i.d. according to a MHR distribution.
\end{theorem} 

 \paragraph{Extreme Value Theorems.} Establishing the above structural properties relies on understanding the tails of MHR and regular distributions. For this purpose, we develop {\em extreme value theorems} for these classes of distributions. {We state our extreme value theorems informally below, referring the reader to Theorems~\ref{thm:extreme MHR} and~\ref{thm:regextremevalue} (in Sections~\ref{sec:truncate} and~\ref{sec:regular} respectively) for formal statements.}
{ 
 \begin{inftheorem}\label{infthm:extMHR}[Extreme Values of MHR Distributions]
Let $X_1,\ldots,X_n$ be a collection of independent random variables whose distributions are MHR, and let $Z=\max_i X_i$. Then, for all $\epsilon$ sufficiently small, at least a $(1-\epsilon)$-fraction of $\mathbb{E}[Z]$ is contributed to by the event $Z \le O(\log_2{1 \over \epsilon}) \cdot \mathbb{E}[Z].$
\end{inftheorem}

 \begin{inftheorem}\label{infthm:extREG}[Extreme Values of Regular Distributions]
Let $X_1,\ldots,X_n$ be a collection of independent random variables whose distributions are regular, and let $Z=\max_i X_i$. Then the {tail} of $Z$ is eventually not fatter than the {tail} of the equal revenue distribution.\footnote{The {\em equal revenue distribution} is supported on $[1,+\infty]$ and has cumulative density function $F(x)=1-{1 \over x}$. Notice that, if a buyer's value for a single item is distributed according to this distribution, the buyer's expected value for the item is $+\infty$. However, if the item is priced at any price $x$, the expected revenue is $1$, hence the name ``equal revenue.'' {The equal revenue distribution is itself a regular distribution. So our theorem says that the fattest the tail of the maximum of $n$ regular distributions can eventually be is that of a regular distribution.}}
\end{inftheorem}}
 
{Bounding the size of the tail of the {\em maximum} of $n$ independent random variables, which are MHR or regular respectively, is instrumental} in establishing the following truncation property: {restricting all item prices into an interval of the form $[\alpha, {\rm poly}(1/\epsilon) \alpha]$ in the MHR case, and $[\alpha, {\rm poly}(n,1/\epsilon)\alpha]$ in the regular case, for some $\alpha$ that depends on the value distributions, only loses an $\epsilon$-fraction of the optimal revenue.}~\notshow{\footnote{\yangnote{Notice that a straightforward concentration bound will only show the probability for the maximum lying outside the interval is small. However, that is not sufficient to show the contribution to the revenue from these cases is small.}}} This is quite remarkable, especially when the value distributions are non-identical or have large tails. How is it possible to restrict the prices into a bounded interval, when the underlying value distributions may concentrate on different supports, or even worse when they do not exhibit good concentration at all as when they are power law distributions?~{\footnote{{A {\em power law} distribution is a distribution whose probability density function $f(x)\propto L(x)x^{-\alpha}$ where $\alpha>1$ and $L(\cdot)$ is a slowly varying function, that is, for any $t>0$, $\lim_{x\rightarrow \infty} \frac{L(tx)}{L(x)}=1$. It usually has large or even unbounded variance. Many power law distributions are also regular, for example when $L(x)$ equals some constant $c$.}}} To establish the  truncation properties claimed above, we follow a different approach depending on whether the underlying distributions are MHR or regular. In the MHR case, we argue (using Theorem~\ref{infthm:extMHR}) that even if we could extract full surplus in the event that $Z\geq \Omega(\log_2{1 \over \epsilon}) \cdot \mathbb{E}[Z]$, the revenue would only increase by a tiny factor. Thus, to obtain nearly-optimal revenue, it suffices to only consider item prices in a bounded range of the form $[\alpha, {\rm poly}(1/\epsilon) \alpha]$.  When the distributions are regular, this approach fails, simply because the expectation $\mathbb{E}[Z]$ could be infinite. We bypass this issue by arguing (using Theorem~\ref{infthm:extREG}) that the tail of $Z$ eventually becomes no heavier than the tail of the equal revenue distribution. Intuitively, this means that varying the extremely high item prices barely affects the revenue. Formally, we prove that, whenever some item price is set higher than some large enough threshold, then bringing it down to the threshold has little effect on revenue.

{Besides enabling the aforementioned structural results for our problem, we expect that our extreme value theorems will find applications in future work, and indeed they have already been used in followup research. In~\cite{DaskalakisW12,CaiDW12}, our theorem is used to convert nearly-optimal multi-item multi-bidder mechanisms for distributions with bounded support to nearly-optimal mechanisms for MHR distributions. For the same setting, \cite{CaiH13} use our theorem to show that relatively simple auctions can extract near-optimal revenue when the bidders are identical, by showing that the welfare is highly-concentrated.} We also note that extreme value theorems {have} been obtained in Statistics for large classes of distributions~\cite{HaanF06}, and indeed such theorems have been applied to optimal mechanism design prior to our work~\cite{BlumrosenH08}. Nevertheless, the known extreme value theorems are typically asymptotic, only hold for maxima of i.i.d. random variables, and are not known to hold for all MHR or regular distributions. We can instead handle the non-i.i.d. case, maxima of a finite number of random variables, and the full spectrum of MHR and regular distributions.



\subsection{Algorithmic Ideas: Covers of Revenue Distributions} \label{sec:techniques}

We overview our approach for Theorem~\ref{thm:additive ptas}. A natural strategy for reducing the search space for an approximately optimal price vector is to discretize the set of possible prices into a finite set, whose size scales mildly with the number of items, $n$, and the approximation accuracy,~$1/\epsilon$. Of course, even with discretization the number of possible price vectors is  exponential in the number of items, and it is not clear how to search this set efficiently. A natural idea to shortcut the search further is to cluster the value distributions into a small number of buckets, containing distributions with similar statistical properties, and proceed to treat all items in a bucket as essentially identical. However, the expected revenue is not sufficiently smooth for us to perform such bucketing. We do obtain a delicate discretization of the supports of the value distributions (Corollary~\ref{cor:additive to discrete}), but cannot discretize the probabilities used by these distributions into a coarse-enough accuracy to allow for polynomial-time solvability of the problem.

Our main algorithmic idea is to shift the focus of attention from the space of {\em value distributions}, which is inherently exponential in the number of items, to the space of all possible {\em revenue distributions}, which are {\em single-dimensional} distributions. The revenue from a given price vector can be viewed as a random variable that depends on the (random) values of the items. So, there is still an exponential number of possible revenue distributions, corresponding to all possible price vectors. Nevertheless, we can exploit the single-dimensional nature of these distributions to construct a polynomial-size $\delta$-cover of the set of all possible revenue distributions under the total variation distance between distributions. That is, for every possible revenue distribution, there exists a distribution in our cover that is within $\delta$ in total variation distance from it.

Our cover is implicit, i.e. we do not provide a closed-form description for it. We show instead that it can be constructed efficiently using dynamic programming. Our algorithm iteratively considers prefixes of the items and, for each prefix $1\ldots i$, constructs a cover of all possible revenue distributions from only pricing items $1, \ldots, i$. For the next iteration, we show that the cover for items $1,\ldots,i+1$ can be easily computed from the cover for items $1,\ldots,i$ and the distribution of~$v_{i+1}$. In the end of our iterations we obtain a polynomial-size $\delta$-cover of all possible revenue distributions, and we argue that only a $\delta$-fraction of revenue is lost if we replace the optimal revenue distribution with its closest one in our cover. And, because the cover has polynomial size, we can exhaustively try every distribution in the cover and its associated price vector to pick the one with the highest expected revenue. A more detailed description of our algorithm is given in Section~\ref{sec: cover view of the problem}, and complete details are provided in Section~\ref{sec:ptas for discretized problem}. Theorem~\ref{thm:additive ptas} follows then easily in Section~\ref{sec:additive}.

Theorem~\ref{thm:general algorithm} follows similarly, except we employ a stronger discretization (Theorem~\ref{thm:discretization}) before using dynamic programming to obtain a cover of all possible revenue distributions. Finally, our algorithms for MHR and regular distributions (Theorems~\ref{thm:ptas mhr} and \ref{thm:quasi ptas regular} respectively) are corollaries of Theorem~\ref{thm:general algorithm}, except they require some extra work for restricting the value distributions into a bounded range. This is achieved in Sections~\ref{sec:truncate} (for MHR distributions) and~\ref{sec:regular} (for regular distributions) using our extreme value theorems (Theorems~\ref{thm:extreme MHR} and~\ref{thm:regextremevalue}). The detailed proofs of Theorems~\ref{thm:ptas mhr} and \ref{thm:quasi ptas regular} are provided in Appendix~\ref{sec:overall ptas}.


\notshow{\yangnote{Yang: Remove this theorem?

An interesting byproduct of our  techniques is that any constant-factor approximation to the optimal revenue can be bootstrapped to achieve an arbitrarily close approximation to it in the case of MHR and regular value distributions. This result (whose proof is given in Section~\ref{sec:overall ptas}) follows from the structure of the proof of our extreme value theorems, which can be initialized with any constant factor approximation to the optimal revenue. 

\yangnote{Yang: Remove this theorem?}
\begin{theorem}[Constant Factor to Near-Optimal Transformation] \label{thm:constant factor to PTAS}
Given a constant-factor approximation to the optimal revenue of an instance of the pricing problem where the values are either MHR or regular, the algorithms of Theorems~\ref{thm:ptas mhr} and~\ref{thm:quasi ptas regular} can be sped up.
\end{theorem}}}

\subsection{Related and Future Work} \label{sec:related}

The focus of this paper is the multidimensional item pricing problem for a unit-demand buyer whose values for the items are independent. This problem is related to the celebrated multidimensional mechanism design problem, but it is restricted in two ways. First, there is a single bidder who is unit-demand. Second, we are interested in coming close to the revenue of the optimal {\em deterministic}---i.e. item pricing---mechanism and not the optimal unrestricted  mechanism, which may also price lotteries over items. While it is unclear whether the restriction to deterministic mechanisms should make the problem easier or harder computationally, the restriction to a single unit-demand bidder should make the problem easier compared to having many bidders with arbitrary valuations.

Despite the apparent simplicity of the item pricing problem,  (near-)optimal polynomial-time algorithms for it were not known prior to our work. Chawla et al.~\cite{ChawlaHK07} provide a  $3$-approximation algorithm, computing a price vector whose revenue is at least a third of the optimal revenue. Their technique is quite elegant, connecting the item pricing problem to a related, single-dimensional mechanism design problem, which can be analyzed using Myerson's result~\cite{Myerson81}. Using the same approach, the approximation factor is improved to $2$ in~\cite{ChawlaHMS10}, and the result is generalized to the multi-bidder setting, albeit with a worse approximation factor. Different work~\cite{BhattacharyaGGM10,Alaei11} obtains polynomial-time constant factor approximations for additive bidders, using convex programming relaxations of the problem. 

However, all these approaches are limited to constant factor approximations, as ultimately the attained revenue is compared to the optimal revenue in a related single-dimensional setting~\cite{ChawlaHK07,ChawlaHMS10}, or a convex programming relaxation of the problem~\cite{BhattacharyaGGM10,Alaei11}. In particular, the limitation of these approaches comes from avoiding a direct comparison of the attained revenue to the optimal revenue in the actual problem, comparing it instead to the optimal revenue in a related problem. Our work provides instead near-optimal algorithms, using a direct comparison to the real optimum via covers of revenue distributions. 

Our work leaves several directions open for exploration and some have already been studied following the announcement of our results~\cite{CaiD11}. We classify them into three categories discussed below.
\begin{itemize}
\item {\em Unit-demand Bidders:} Can our near-optimal algorithms be improved to be exactly optimal? Recent work has shown that the answer is no, namely that there are no exactly optimal polynomial-time  algorithms for product value distributions, unless ${\tt P} = {\tt NP}$~\cite{ChenDPSM14}. Still there is room for improving the dependence of our running times on the approximation parameter $\epsilon$. E.g., is there an algorithm that runs in time polynomial in $n$ and $1/\epsilon$ when the item values are bounded in~$[0,1]$?

And how about correlated distributions over item values? Here, it had already been known that computing an optimal price vector is highly inapproximable in polynomial-time~\cite{BriestK07}. So there cannot even be a polynomial-time constant factor approximation in this case.

Beyond item pricing, it is important to understand the complexity of optimal randomized mechanisms, which may increase revenue by also pricing lotteries over items~\cite{Thanassoulis04,BriestCKW10}. For product distributions, Chawla et al.~\cite{ChawlaMS10} {show that randomization does not increase revenue by more than a factor of 4, thus extending} the constant-factor approximation algorithms of~\cite{ChawlaHK07,ChawlaHMS10} to the randomized multi-bidder setting, except with worse approximation guarantees. Is there a polynomial-time optimal mechanism for this setting? No computational lower bound is known at the time of writing of this paper. 

For correlated distributions over item values, Cai et al.~\cite{CaiDW12b} obtain near-optimal, randomized mechanisms for multi-bidder multi-item settings with unit-demand bidders. For any desired accuracy $\epsilon>0$, they compute a mechanism whose revenue is within an additive error of $\epsilon$ from optimal in time polynomial in $1/\epsilon$ and the size of the support of each bidder's distribution over valuations, {when these distributions are discrete. (When they are continuous, they are handled via fine enough discretization.)} This algorithm is clearly also applicable  when every bidder's values for the items are independent (i.e. the setting discussed in the previous paragraph). However, the dependence of the running time on the support of the product distribution may be unreasonable computationally. Indeed, a {discrete} product distribution can be described by specifying all of its marginals, with description complexity logarithmic in the size of its support. 
%

\item {\em Additive Bidders:} Can our algorithms be extended to additive bidders? Here, an optimal mechanism may increase revenue by pricing bundles of items~\cite{ManelliV06}, or (if randomization is allowed) lotteries over bundles of items. Exploiting our extreme value theorems for MHR distributions, Cai and Huang~\cite{CaiH13} provide near-optimal polynomial-time mechanisms for multiple i.i.d. bidders, whose values for the items are independently distributed according to MHR distributions. Moreover, Daskalakis et al.~\cite{DaskalakisDT14} show that this result cannot be made exact for general product distributions. They show that, subject to widely held complexity theoretic beliefs---in particular that ${\tt ZPP} \not\supseteq {\tt P}^{\#{\tt P}}$,\footnote{${\tt ZPP} \supseteq {\tt P}^{\#{\tt P}}$ would imply that there are randomized polynomial-time algorithms for ${\tt NP}$-complete problems, which is widely believed to be unlikely.} computing and implementing an exactly optimal mechanism cannot be done computationally efficiently. Indeed, this is true even in the ostensibly simple setting where there is a single additive bidder whose values for the items are independently distributed on two rational numbers with rational probabilities. 

For correlated distributions, Cai et al.~\cite{CaiDW12} obtain (exactly) optimal mechanisms for multi-bidder multi-item settings with additive bidders, in time polynomial in the size of the support of each bidder's distribution over valuations.

On a different vein, Daskalakis et al.~\cite{DaskalakisDT13} study the structure (rather than the computational complexity) of optimal mechanisms, following earlier work on the topic by Economists, e.g.~\cite{Rochet85,Armstrong00,ManelliV06,ManelliV07,Pavlov11}. They provide a duality framework  based on Monge-Kantorovich duality for characterizing the structure of the optimal mechanism of selling multiple items to a single additive bidder.


\item {\em General Settings:}\notshow{\yangnote{\footnote{\yangnote{Yang says: I think going directly to general objective might be too big a jump for the readers. Can we talk about revenue with general valuations first, then go to general objectives? Also, I think we should specifically point out that for revenue the algorithmic problem is only maximizing welfare. I am worried that the readers will think revenue plus welfare is not a well defined objective (e.g. under what constraints?)...}}}} It is important to understand the computational complexity of mechanism design in general settings: multiple bidders, general valuations {(beyond unit-demand and additive), general constraints on what allocations of items to bidders are feasible,} and general objectives, potentially going beyond the familiar objectives of revenue and welfare.\footnote{{A general objective $O$ takes as input the valuations of the bidders $\vec{t}$ and a randomized allocation and price vector $(A,p)$ and outputs a real number $O(\vec{t},(A,p))$.}} In recent work, Cai et al.~\cite{CaiDW13b} provide a {computational black-box reduction from {\em mechanism design} for maximizing an arbitrary concave objective $O$\footnote{{An objective function $O(\vec{t},(A,p))$ is called {\em concave} iff, for all bidder valuations $\vec{t}$, and all $(A_1,p_1)$ and $(A_2,p_2)$, it holds that $O(\vec{t}, {1 \over 2} (A_1,p_1) + {1 \over 2}  (A_2, p_2)) \ge {1 \over 2}  O(\vec{t}, (A_1, P_1)) + {1 \over 2}  O(\vec{t}, (A_2, P_2))$, where ${1 \over 2}  (A_1,p_1) + {1 \over 2}  (A_2, p_2)$ denotes uniformly randomizing between $(A_1,p_1)$ and $(A_2,p_2)$. Clearly, revenue and welfare satisfy this condition with equality, but several other objectives are concave, such as the max-min fairness objective considered in~\cite{CaiDW13b}.}} under arbitrary allocation constraints and an arbitrary family of bidder valuations (e.g. submodular, supermodular, etc.) to {\em algorithm design} for that same objective $O$, {modified by an additive virtual welfare and virtual revenue term}, and under the same allocation constraints and family of allowed valuations.} Roughly speaking, they show that, whenever the algorithmic problem is polynomial-time solvable (exactly or approximately), the mechanism design problem also becomes solvable (exactly or approximately) in time polynomial in the size of the support of each bidder's distribution over valuations. It is important to find applications of this reduction to settings of interest beyond optimizing fractional max-min fairness for additive bidders, which was done in~\cite{CaiDW13b}.

\end{itemize}

\section{Preliminaries} \label{sec:prelim}

\paragraph{Computational Problems.} We define three variants of the item pricing problem. {\sc AdditivePrice} and {\sc Price} are the main computational problems that we aim to solve, but {\sc RestrictedPrice} is an auxiliary one that is helpful in the analysis. For the value distributions that we consider, it can be shown that all three problems have finite optimal solutions.

\medskip \noindent\framebox{\begin{minipage}[h]{16.5cm}{\sc AdditivePrice}: {\bf Input:} A collection of mutually independent random variables $\{v_i\}_{i=1}^n$, and some $\epsilon>0$. {\bf Output:} A vector of prices $P=(p_1,\ldots,p_n)$ such that the expected revenue ${\cal R}_P$ from using $P$, defined as in Eq.~\eqref{eq:objective}, is within an additive $\epsilon$ of the optimal revenue achieved by any price vector.
\end{minipage}}

\medskip \noindent\framebox{\begin{minipage}[h]{16.5cm}{\sc Price}: {\bf Input:} A collection of mutually independent random variables $\{v_i\}_{i=1}^n$, and some $\epsilon>0$. {\bf Output:} A vector of prices $P=(p_1,\ldots,p_n)$ such that the expected revenue ${\cal R}_P$ from using $P$, defined as in Eq.~\eqref{eq:objective}, is within a $(1+\epsilon)$-factor of the optimal revenue achieved by any price vector.
\end{minipage}}

\medskip \noindent\framebox{\begin{minipage}[h]{16.5cm} {\sc RestrictedPrice}: {\bf Input:} A collection of mutually independent random variables $\{v_i\}_{i=1}^n$ supported on a common discrete set $\mathcal{S}$, and a discrete set ${\cal P} \subset \mathbb{R}_{\ge 0}$ of possible prices.\\ {\bf Output:} A vector of prices $P=(p_1,\ldots,p_n) \in {\cal P}^n$ such that the expected revenue ${\cal R}_P$ from using $P$ is optimal among all vectors in ${\cal P}^n$.
\end{minipage}}

\bigskip \noindent In Section~\ref{sec:proof outline} we describe how these computational problems are interconnected through other results in this paper to establish Theorems~\ref{thm:additive ptas} through~\ref{thm:quasi ptas regular}.

\paragraph{Computational Efficiency.} Throughout the paper we use the standard convention of identifying ``computational efficiency'' with polynomial-time computation. Namely, we will say that an algorithm is ``computationally efficient'' iff its running time is polynomial in the number of bits required to describe the input to the algorithm.

\paragraph{Reductions Between Computational Problems.} We provide several reductions between different flavors of the item pricing problem. Formally, a (polynomial-time) reduction from a computational problem $P_1$ (e.g. {\sc Price}) to a computational problem $P_2$ (e.g. {\sc RestrictedPrice}) is a pair of (polynomial-time) algorithms ${\cal A}$ and ${\cal B}$ satisfying the following properties. For all inputs $\Pi_1$ to $P_1$:
\begin{itemize}
\item ${\cal A}(\Pi_1)$ is a valid input to $P_2$;
\item if $S$ is a solution to ${\cal A}(\Pi_1)$ then ${\cal B}(S)$ is a solution to $\Pi_1$.
\end{itemize}
For example, a polynomial-time reduction from {\sc Price} to {\sc RestrictedPrice} would allow us to convert (in polynomial time) any input to {\sc Price} to a valid input to {\sc RestrictedPrice} so that, if we found a solution to the latter, we would also be able to compute (in polynomial time) a solution to the former.

\paragraph{Approximation Algorithms.} Our algorithmic results use the concept of a {\em Polynomial-Time Approximation Scheme}, or PTAS. A PTAS for a computational problem such as {\sc Price} is a collection of algorithms $({\cal A}_{\epsilon})_{\epsilon}$, indexed by the approximation parameter $\epsilon$, such that, for all $\epsilon>0$ and for any given input ${ \Pi}$ to the algorithm, Algorithm ${\cal A}_{\epsilon}$ computes an $\epsilon$-optimal solution to $\Pi$ in time $d(\Pi)^{g(1/\epsilon)}$, where $d(\Pi)$ is the number of bits required to describe problem $\Pi$ and $g$ is some increasing function of $1/\epsilon$, which does not depend on $\Pi$ or its description complexity. The algorithms in the collection are called polynomial-time because for all fixed $\epsilon$, e.g. $\epsilon=1/3$, the running time of ${\cal A}_{\epsilon}$ is polynomial in the description of the problem. A {\em quasi-polynomial-time approximation scheme}, or Quasi-PTAS is a similar concept, except that running time is $2^{g(1/\epsilon)\cdot\rm{poly}(\log d(\Pi))}$ for some function $g$ instead of $2^{g(1/\epsilon)\cdot \log {d(\Pi)}}$ as in a PTAS.

\paragraph{Distributions.} For a random variable $X$ we denote by $F_X(x)$ the cumulative distribution function of $X$, and by $f_X(x)$ its probability density function. We also let $u^X_{min} = \sup\{x|F_X(x)=0\}$ and $u^X_{max}= \inf\{x|F_X(x)=1\}$. $u^X_{max}$ may be $+\infty$, but we assume that $u^X_{min}\geq 0$, since our distributions represent value distributions. We drop the subscript/superscript of $X$, if $X$ is clear from context.  

We proceed with a precise definition of {\em Monotone Hazard Rate (MHR)} and {\em Regular} distributions, which are two commonly studied families of distributions. 
    
\begin{definition}[Monotone Hazard Rate Distribution]
   We say that a one-dimensional differentiable distribution $F$ {has} {\em Monotone Hazard Rate} if $\frac{f(x)}{1-F(x)}$ is  non-decreasing in $[u_{min},u_{max}]$. We call such $F$ a {\em Monotone Hazard Rate, or MHR, distribution.}
\end{definition}
\begin{definition}[Regular Distribution]
    A one-dimensional differentiable distribution $F$ is called {\em regular} if $x-{1-F(x)\over f(x)}$ is  non-decreasing  in $[u_{min},u_{max}]$. 
\end{definition}
\noindent It is worth noting that all MHR distributions are also regular distributions, but there are regular distributions that are not MHR. The family of MHR distributions includes such familiar distributions as the Normal, Exponential, and Uniform distributions. The family of regular distributions contains a broader range of distributions, including fat-tail distributions $f_X(x)\sim x^{-(1+\alpha)}$ for $\alpha\ge1$ (which are not MHR). 
In Appendices~\ref{appendix:MHR} and~\ref{sec:regconcavity} we establish  important properties of MHR and regular distributions. These properties are instrumental in establishing our extreme value theorems for these distributions (Theorems~\ref{thm:extreme MHR} and~\ref{thm:regextremevalue} in Sections~\ref{sec:truncate} and~\ref{sec:regular}).

To formally study the complexity of the item pricing problem, we need to pay attention to how value distributions are described as part of the input to the problem. We discuss this technical issue in Appendix~\ref{sec:model}, entertaining three types of access to a distribution. Maybe we are given an {\em explicit description} of the distribution, specifying its support and the probabilities assigned to every value in the support. Such explicit description is appropriate if the distribution is discrete and has finite support. Alternatively, we may have more limited access to the distribution. E.g., we may only have {\em sample access} to the distribution via a procedure that generates independent samples from it. Or, we may have {\em oracle access} to the cumulative distribution function via a procedure that returns its value at any queried point. We formally discuss these types of access to a distribution in Appendix~\ref{sec:model}, also defining precisely what it means for an algorithm who takes as input a distribution or outputs a distribution to be ``computationally efficient'' in each case.

\paragraph{Order Notation.} Throughout the text we use the $O(\cdot)$, $\Omega(\cdot)$ and $\Theta(\cdot)$ notation. Let $f(x)$, $g(x)$ be two positive functions defined on some infinite subset of $\mathbb{R}_+$. Then:
\begin{itemize}
\item we write $f(x) = O(g(x))$ iff there exist some positive reals $M$ and $x_0$ such that $f(x) \le M g(x)$, for all $x>x_0$;
\item we write $f(x) = \Omega(g(x))$ iff there exist some positive reals $m$ and $x_0$ such that $f(x) \ge m g(x)$, for all $x>x_0$; and
\item we write $f(x) = \Theta(g(x))$ iff $f(x) = O(g(x))$ and $f(x) = \Omega(g(x))$.
\end{itemize}

\paragraph{Other Notation.} Whenever we write ${\rm poly}(x)$ in an expression providing a bound to some quantity, we mean that there exists some positive polynomial $p(x)$ which can replace ``${\rm poly}(x)$'' so that the bound is true. Whenever we write $\log$ in some expression without specifying the base of the logarithm, any constant base that is larger than $1$ works. For some positive rational number $x$, we write $\langle x \rangle$ to denote the bit complexity of $x$, i.e. the number of bits required to specify the numerator and denominator of an irreducible fraction representing $x$.
\section{Paper Organization}\label{sec:proof outline}

We provide a roadmap to the paper and a high-level description of our approach. We first study {\scshape RestrictedPrice}. Despite its input/output restrictions, it addresses the major algorithmic challenges of the item pricing problem. Our approach to {\scshape RestrictedPrice} is to first design a dynamic programming algorithm that produces a cover of all possible revenue distributions arising from price vectors in ${\cal P}^n$, where ${\cal P}$ is the set of possible prices specified in the input to the problem. Using the cover it is then easy to obtain a near-optimal algorithm for {\scshape RestrictedPrice}, which is not necessarily polynomial-time. Section~\ref{sec: cover view of the problem} provides an intuitive description of the dynamic programming approach for producing the cover, and Section~\ref{sec:ptas for discretized problem} describes our algorithm for {\scshape RestrictedPrice} formally. This algorithm serves as the main algorithmic tool of this paper, and is at the root of the tree structure of Figure~\ref{fig:roadmap}, summarizing our results and proofs.
\begin{figure}[h!]
 
  \includegraphics[height = 5.7in]{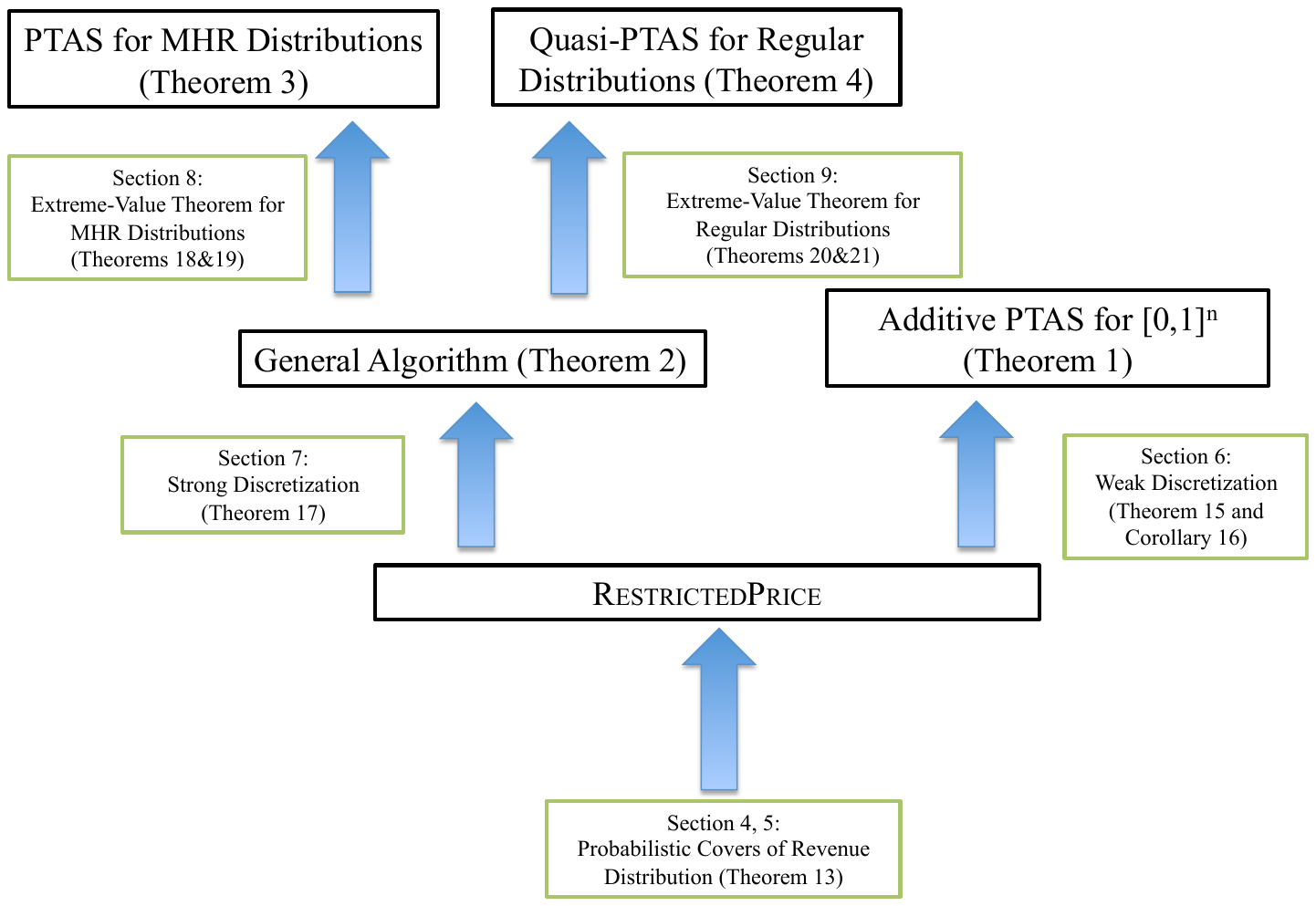}\\
   \caption{Overview of our results and the proof structure. Arrows are implications.}\label{fig:roadmap}
\end{figure}

In Section~\ref{sec:additive}, we obtain Theorem~\ref{thm:additive ptas} by reducing {\scshape AdditivePrice} for value distributions supported on $[0,1]$ to {\scshape RestrictedPrice}. The reduction is obtained by showing a discretization result, establishing that the supports of the value distributions as well as the candidate prices can be discretized without too much loss in revenue. The reduction is summarized by Corollary~\ref{cor:additive to discrete}, which together with our algorithm for {\sc RestrictedPrice} immediately shows Theorem~\ref{thm:additive ptas}. 

In Section~\ref{sec:balanced range}, we move on to multiplicative approximations, establishing Theorem~\ref{thm:general algorithm}. The approach is similar, reducing {\scshape Price} to {\scshape RestrictedPrice} by discretizing the supports of the value distributions as well as the set of available prices. However, Theorem~\ref{thm:additive discretization}, the discretization result at the heart of Corollary~\ref{cor:additive to discrete} (our reduction from Section~\ref{sec:additive}), is not strong enough for our purposes here. We establish instead a stronger discretization (Theorem~\ref{thm:discretization}) that is sufficiently powerful for our reduction.

In Sections~\ref{sec:truncate} and~\ref{sec:regular}, we establish our algorithms for MHR and regular distributions. In Section~\ref{sec:truncate}, we present an extreme value theorem for MHR distributions (Theorem~\ref{thm:extreme MHR}). This theorem enables us to obtain a polynomial-time reduction from {\sc Price} where the value distributions are MHR to {\sc Price} where the value distributions are supported on a common range of the form $[u_{\min}, r \cdot u_{\min}]$, where the multiplier $r$ is independent of the number of items $n$. Our reduction is summarized by Theorem~\ref{thm:reduction MHR to balanced}. Theorem~\ref{thm:ptas mhr} follows then as a corollary of Theorem~\ref{thm:reduction MHR to balanced} and Theorem~\ref{thm:general algorithm}. Our algorithm for regular distributions follows similarly in Section~\ref{sec:regular}. We show an extreme value theorem for regular distributions (Theorem~\ref{thm:regextremevalue}), enabling a reduction from {\sc Price} with regular distributions to {\sc Price} with distributions supported on a common range of the form $[u_{\min}, r \cdot u_{\min}]$, except that now the multiplier $r$ depends polynomially on the number  of items $n$. Our reduction is summarized in Theorem~\ref{thm:regreduction}. Theorem~\ref{thm:quasi ptas regular} follows then as a corollary of Theorem~\ref{thm:regreduction} and Theorem~\ref{thm:general algorithm}.

\paragraph{Reading the paper.} Sections~\ref{sec: cover view of the problem} through~\ref{sec:regular} are meant to provide a self-contained overview of the proofs of our algorithmic results, with the appendices containing the complete proof details. Appendix~\ref{app:roadmap} provides a roadmap to the appendices.

\section{Probabilistic Covers of Revenue Distributions} \label{sec: cover view of the problem}

\notshow{\yangnote{Let ${\cal V}:=\{v_i\}_{i}$ be an instance of {\sc AdditivePrice}, where the $v_i$'s are mutually independent random variables distributed  on $[0,1]$ according to the distributions $\{{F}_i\}_{i}$, and let ${\cal R}_{OPT}$ be the expected revenue from the optimal price vector for $\cal V$. Our goal in this section is to compute a price vector with expected revenue at least ${\cal R}_{OPT}-\epsilon$. 

Section~\ref{sec:additive} provides a {computationally} efficient reduction of this problem to the problem of approximating within additive $O(\epsilon)$ of the optimal revenue of a discretized problem, where both the values and the prices come from discrete sets whose cardinality is $O({{\log 1/\epsilon}\over{\epsilon^2}})$. 

To accommodate scenarios that will be considered in future sections, we provide a more general algorithm. We assume the values lie in an interval $[u_{min},u_{max}]$ where $r = {u_{max} \over u_{min}}$. For convenience, we denote by $\{\hat{F}_i\}_i$ the discretized distributions resulting from the reduction, by $\{\hat{v}_i\}_{i}$ a collection of mutually independent random variables distributed according to the $\hat{F}_i$'s, by $\{v^{(1)},v^{(2)},\ldots,v^{(k_1)}\}$ the (common) support of all the $\hat{F}_i$'s, and by $\{p^{(1)}, p^{(2)},\ldots, p^{(k_2)}\}$ the set of available price levels, where for {\sc RestrictedPrice}$(\mathcal{\hat{V}},\mathcal{P},O(\epsilon))$ in Corollary~\ref{cor:additive to discrete} both $k_1$ and $k_2$ are $O({{\log 1/\epsilon}\over{\epsilon^2}})$. 
}}

In this section, we discuss our algorithmic approach to {\scshape RestrictedPrice}, postponing the description of our algorithm for it to Section~\ref{sec:ptas for discretized problem}. As we have already discussed in Section~\ref{sec:proof outline}, although seemingly restricted this problem captures the main algorithmic challenges underlying  problems {\scshape Price} and {\scshape AdditivePrice}. In particular, our algorithm for {\scshape RestrictedPrice} will become a central building block in all our algorithmic results (Theorems~\ref{thm:additive ptas} through~\ref{thm:quasi ptas regular}). 

For convenience, throughout this section we will take ${F}_1,\ldots,{F}_n$ to be a collection of distributions supported on a discrete set ${\cal S}=\left\{v^{(1)},v^{(2)},\ldots,v^{(k_1)}\right\}$, and ${v}_1,\ldots,{v}_n$ to be a collection of mutually independent random variables distributed according to the ${F}_i$'s. We will then assume that the  input to {\scshape RestrictedPrice} comprises the ${v}_i$'s and a finite set of prices ${\cal P}=\{p^{(1)}, p^{(2)},\ldots, p^{(k_2)}\}$.

\smallskip The obvious algorithmic challenge in {\sc RestrictedPrice} is that, even though the set of possible prices is finite, there are still exponentially many (namely $k_2^n$) possible price vectors that we need to choose from for an optimal one. If $k_2$ were a constant and the items were i.i.d., then we could decrease the possible vectors to a polynomial number by exploiting the symmetry of the items.\footnote{A broader exposition of the role of symmetries in mechanism design can be found in~\cite{DaskalakisW12}.} Similarly, we can obtain polynomial-time algorithms for the case where there is only a constant number of possible value distributions and a constant number of possible prices. However, when all the $F_i$'s may be different, the problem looks inherently exponential, even if both $k_1$ and $k_2$ are absolute constants, e.g., even when the value distributions are supported on $2$ possible values and there are $2$ possible prices available.

Our algorithmic approach is enabled by a shift in perspective, which may be applicable to other problems with a similar structure. To illustrate the approach, let us view our problem in the graphical representation of Figure~\ref{fig:d1}. $C$ is a function that takes as input a price vector $P=(p_1,\ldots,p_n)$ and outputs the distribution ${F}_{{ R}_P}$ of the revenue of the seller under this price vector. Indeed, the revenue of the seller is a random variable ${R}_P$ that depends on the random variables $\{{v}_i\}_{i \in [n]}$. So in order to compute the distribution of the revenue, $C$ also takes as input the distributions $F_1,\ldots,F_n$. What we are aiming at maximizing is the expectation ${\cal R}_P$ of ${R}_P$.
\begin{figure}[h!]
  \centering
  \includegraphics[height = 1.4in]{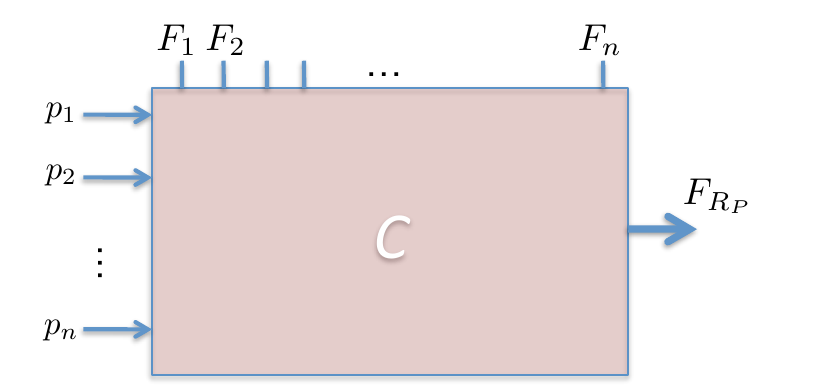}\\
   \caption{The Revenue Distribution as the output of a function. The inputs to the function are the prices and the value distributions.}\label{fig:d1}
\end{figure}
 
Given our restriction of the prices to a finite set $\{p^{(1)}, p^{(2)},\ldots, p^{(k_2)}\}$, there are $k_2^n$ possible inputs to the {function}, and at most $k_2^n$ possible revenue distributions that the {function} can output. Our main conceptual idea is the following:\\ 

\begin{minipage}[h]{15cm}{\em Instead of searching the space of possible price vectors that can be input to $C$, we search the space of possible outputs of $C$, i.e. the space of all possible revenue distributions resulting from different price vectors, for one with maximum expectation.

\smallskip Moreover, to efficiently search the space of all possible revenue distributions, we construct an appropriately small subset of it and only search the distributions in that subset.}\end{minipage}

\bigskip \noindent The subset we construct is a probabilistic cover (under some appropriate metric) of the space of all possible revenue distributions.\footnote{A {\em $\delta$-cover} of a set of distributions ${\cal F}$ with metric $d$ is a subset ${\cal F}' \subseteq {\cal F}$ such that for all $F \in {\cal F}$ there exists some $F' \in {\cal F}'$ such that $d(F,F') \le \delta$.} The properties of our cover that are crucial for our algorithmic applications are the following: (a) the cover has small cardinality, and (b) for any possible revenue distribution that the {function} may output, there exists a revenue distribution in our cover whose expectation is close. 

\paragraph{Constructing the Cover.} At a high level, we construct our cover using dynamic programming (henceforth DP for short), whose steps are interleaved with coupling arguments that prune the size of the DP table before proceeding to the next step. 

\smallskip Intuitively, our DP algorithm sweeps the items from $1$ through $n$, maintaining a cover of the revenue distributions produced by all possible price vectors on every prefix of the items. More precisely, for each prefix $1\ldots j$ of the items, our DP table keeps track of all possible feasible collections of $k_1\times k_2$ probability values, where $\Pr_{i_1,i_2}$, $i_1 \in [k_1], i_2 \in [k_2]$, denotes the probability that the item with the largest value-minus-price gap (i.e. the item that would have been sold in a sale that only sells items $1$ through $j$) has value $v^{(i_1)}$ for the buyer and is assigned price $p^{(i_2)}$ by the seller. I.e. we store in our DP table all possible (winning-value, winning-price) distributions that can arise from a price vector on every prefix of the items. The reasons we store these distributions are the following:
\begin{itemize}

\item First, if we have all possible (winning-value, winning-price) distributions for the full set of items, we can search for the one with the highest expected revenue. For every distribution we will also maintain in our DP table a price vector resulting in that distribution. So, once we have found the distribution with the optimal expected revenue, we will also find the price vector with that optimal revenue.

\item Second, we can construct the set of all possible (winning-value, winning-price) distributions for the full set of items, by considering one prefix at a time. In particular, suppose that we have all possible (winning-value, winning-price) distributions for the prefix of items $1\ldots j$. By combining every such distribution with all $k_2$ possible prices for item $j+1$, we can compute all possible (winning-value, winning-price) distributions for the prefix of items $1\ldots j+1$. That is, if we have these distributions for a prefix of items, we do not need any other information to extend the prefix by one item. For this scheme to work, observe that it is crucial to maintain the joint distribution of both the winning-value and the winning-price, rather than just the distribution of the winning-price.

%
%
%
%
%
\end{itemize}

Clearly, the dynamic programming approach that we just outlined for computing a cover of all possible revenue distributions achieves nothing in terms of reducing the number of distributions. Indeed, there could be one (winning-value, winning-price) distribution for every price vector, so that the total number of distributions that we need to store in our DP table is exponential. To control the size of our cover from exploding, we show that we can be coarse in our bookkeeping of the (winning-value, winning-price) distributions, without sacrificing much revenue. Indeed, it is here where viewing our problem in the ``upside-down'' manner illustrated in Figure~\ref{fig:d1} (i.e. targeting a cover of the output of $C$) is important. We show that we can discretize the probabilities used by the distributions stored in the DP table into multiples of some fraction $1 \over m$ without losing much revenue. In particular, after a prefix of items is processed by the algorithm, we show that we can discretize the probabilities in all distributions in the table before considering the next item. That the loss due to coarsening the probabilities is not significant follows from coupling arguments interleaved with the steps of  dynamic programming.

%
%

\medskip In the next section we make our ideas precise, obtaining our algorithm for {\sc RestrictedPrice}.

\section{The Algorithm for the Discrete Problem} \label{sec:ptas for discretized problem}

In this section, we formalize our ideas from the previous section, describing our main algorithmic result for {\sc RestrictedPrice}. We use the same notation as in Section~\ref{sec: cover view of the problem}, namely we assume that the input distributions are supported on a common set ${\cal S}$ of cardinality $k_1$ and the prices are restricted to a set ${\cal P}$ of cardinality $k_2$. We also denote by $OPT$ the optimal expected revenue for the input value distributions when the prices are restricted to ${\cal P}$.


\paragraph{The Algorithm.} As a first step we discretize the probabilities used by the input distributions. We prove a discretization lemma that provides a polynomial-time reduction from our problem into a new one, where additionally the probabilities that the value distributions assign to each point in their support~${\cal S}$  are integer multiples of $1/m$, for some integer $m$ that is a free parameter in our algorithm. We show that the loss in revenue resulting from our reduction is at most an additive {$\frac{4k_1n}{m} \max_i\{p^{(i)}\}$} in the following sense: for any price vector $P$, the expected revenue from the original value distributions $\{F_i\}_i$ and the expected revenue from the discretized distributions $\{\hat{F}_i\}_i$ are within an additive $\frac{4k_1n}{m} \max_i\{p^{(i)}\}$. Moreover, the construction of Lemma~\ref{lem:verticaldiscretization} is explicit, so from now on we can assume that we know the $\{\hat{F}_i\}_i$ explicitly, regardless of what type of access we have to the $\{F_i\}_i$ (see Appendix~\ref{sec:model}).

The second phase of our algorithm is the dynamic programming algorithm outlined in Section~\ref{sec: cover view of the problem}. We provide some further details on this now. Our algorithm computes a Boolean function $g(i,{{\Pr}})$, whose arguments lie in the following range:  $i\in[n]$ and ${{\Pr}}=({\Pr}_{1,1},{\Pr}_{1,2},\ldots,{\Pr}_{k_1,k_2})$, where each ${\Pr}_{i_1,i_2} \in [0,1]$ is an integer multiple of {$1\over m$}. The function $g$ is stored in a table that has one cell for every setting of $i$ and ${{\Pr}}$, and the cell contains a $0$ or a $1$ depending on the value of $g$ at the corresponding input. In the terminology of the previous section, argument $i$ indexes the last item in a prefix of items and ${{\Pr}}$ defines a (winning-value, winning-price) distribution whose probabilities are integer multiples of ${1 \over m}$. If ${{\Pr}}$ can arise from some pricing of the items $1 \ldots i$ (up to discretization of probabilities into multiples of ${1 \over m}$), we intend to store $g(i,{{\Pr}})=1$; otherwise we store $g(i,{{\Pr}})=0$. {For each cell of the table such that $g(i,{{\Pr}})=1$, we also store a price vector on the corresponding prefix of items $1\ldots i$ consistent with $\Pr$.}

For conciseness, we give next a high-level description of the dynamic programming algorithm, postponing its full details to Appendix~\ref{sec:DP step}. The table is filled in a bottom-up fashion from $i=1$ through $n$. At the end of the $i$-th iteration, we have computed all feasible ``discretized'' (winning-value,winning-price) distributions for the prefix of items $1\ldots i$, where ``discretized'' means that all probabilities have been rounded into multiples of $1/m$. For the next iteration, we try all possible prices $p^{(j)}$ for item $i+1$ and compute how each of the feasible discretized (winning-value,winning-price) distributions for the prefix $1\ldots i$ evolves into a discretized distribution for the prefix $1\ldots i+1$, setting the corresponding cell of layer $g(i+1,\cdot)$ of the DP table to $1$. Notice, in particular, that {\em we lose accuracy in every step of the dynamic programming algorithm}, as each step involves computing how a discretized distribution for items $1\ldots i$ evolves into a distribution for items $1\ldots i+1$ and then rounding the latter back into multiples of $1/m$. We show in the analysis of our algorithm that the error accumulating from these roundings can be controlled via coupling arguments.

After computing the truth-table of function $g$, we look at all cells such that $g(n,{\Pr})=1$ and evaluate the expected revenue resulting from the distribution ${\Pr}$, i.e.
$${\cal R}_{{\Pr}} = \sum_{i_1\in[k_1],i_2\in[k_2]} p^{(i_2)}\cdot {\Pr}_{i_1,i_2} \cdot \ind_{v^{(i_1)}\geq p^{(i_2)}}.$$
Having located the cell whose ${\cal R}_{{\Pr}}$ is the largest, {we output the price vector stored in that cell.} \notshow{we follow back-pointers to obtain a price-vector consistent with ${\Pr}$. At some steps of the back-tracing, there may be multiple choices; we pick an arbitrary one to proceed.}



\paragraph{Running Time and Correctness.} Next we bound the algorithm's running time and revenue.

\begin{theorem} \label{lem:dynamic programming revenue}     \label{lem:running time}
    Given an instance of {\sc RestrictedPrice}, where the value distributions are supported on a discrete set ${\cal S}$ of cardinality $k_1$ and the prices are restricted to a discrete set ${\cal P}$ of cardinality $k_2$, and for any choice of discretization accuracy $m \ge 2k_1$, the algorithm described in this section produces a price vector with expected revenue at least $$OPT -{(2nk_{1}k_{2} + 16k_{1}n)\over  m}\cdot \max\{{\cal P}\},$$
where $\max\{{\cal P}\}$ is the maximum element in ${\cal P}$ and $OPT$ the optimal expected revenue. The running time of the algorithm is  polynomial in the size of the input and {$m^{k_1 k_2}$}. 
\end{theorem}

%
The proof of the theorem is given in Appendix~\ref{sec:algorithm analysis}. Intuitively, if we did not perform any rounding of distributions, our algorithm would have been  {\em exact}, outputting an optimal price vector in $\{p^{(1)},\ldots,p^{(k_2)}\}^n$. What we show is that the roundings performed at the steps of the dynamic programming algorithm are fine enough that do not become detrimental to the revenue. To show this, we use coupling arguments, invoking the coupling lemma and the optimal coupling theorem after each step of the algorithm.~(See Lemma~\ref{lem:closeness} in Appendix~\ref{sec:correctness of DP}.) This way, we show that the rounded (winning-value,winning-price) distributions maintained by the algorithm for each price vector are close in total variation distance to the corresponding exact distributions arising from these price vectors, culminating in Theorem~\ref{lem:dynamic programming revenue}.

\notshow{Using Lemmas~\ref{lem:dynamic programming revenue} and~\ref{lem:running time} and our work in the later sections, we obtain our main algorithmic results in this paper (Theorems~\ref{thm:ptas mhr},~\ref{thm:quasi ptas regular},~\ref{thm:general algorithm}, and~\ref{thm:additive ptas}). See Section~\ref{sec:app overall ptas} for the proof of these theorems.}


\section{Additive PTAS for values distributed in $[0,1]^{n}$}\label{sec:additive}
In this section, we provide a polynomial-time reduction from {\sc AdditivePrice}$(\mathcal{{V}},\epsilon)$, for value distributions $\mathcal{{V}}=\{v_{i}\}_i$ supported on~$[0,1]$, to $O(\epsilon)$-approximating {\sc RestrictedPrice}$(\mathcal{\hat{V}},\mathcal{P})$, where $\mathcal{\hat{V}}$ is a collection of mutually independent random variables supported on a common set of cardinality $\poly(1/\epsilon)$ and $|\mathcal{P}|=\poly(1/\epsilon)$. A PTAS for {\sc AdditivePrice} then follows from Theorem~\ref{lem:dynamic programming revenue} with an appropriate choice of the discretization $m$.

\medskip As a first step, we reduce {\sc AdditivePrice}$(\mathcal{{V}},\epsilon)$ to {\sc AdditivePrice}$(\mathcal{\tilde{V}},O(\epsilon))$, where the random variables $\mathcal{\tilde{V}}=\{\tilde{v}_{i}\}_i$ are independently distributed in $[O(\epsilon),1]$. The reduction is quite straightforward, replacing all sampled values that are smaller than some $O(\epsilon)$ with $O(\epsilon)$ and keeping the rest unchanged. We argue that a nearly optimal price vector for the new value distributions is also nearly optimal for the original value distributions. Formally, 
\begin{lemma}\label{lem:additive to balanced}
Let $\mathcal{V}=\{v_{i}\}_{i\in[n]}$ be a collection of mutually independent random variables supported on $[0,1]$. For any $\epsilon>0$, there is a polynomial-time reduction from {\sc AdditivePrice}$(\mathcal{V},\epsilon)$ to {\sc AdditivePrice}$(\mathcal{\tilde{V}},\epsilon/3)$, where $\mathcal{\tilde{V}}=\{\tilde{v}_{i}\}_{i\in[n]}$ is a collection of mutually independent random variables supported on $[\epsilon/6,1]$.\end{lemma}

The proof can be found in Appendix~\ref{appendix:additive}.


\smallskip Next we want to discretize the problem {\sc AdditivePrice}$(\mathcal{\tilde{V}},\epsilon/3)$. As alluded to in Section~\ref{sec:intro}, the expected revenue can be sensitive even to small perturbations of the prices and the probability distributions. So our discretization, summarized in the next theorem, must be done delicately.

%

\begin{theorem}[Price/Value Discretization for Additive Approximation] \label{thm:additive discretization}
Let ${\cal V}=\{v_i\}_{i\in[n]}$ be a collection of mutually independent random variables supported on a {bounded} set $[u_{min},u_{max}] \subset \mathbb{R}_+$, and let $r={u_{max}\over u_{min} } \ge 1.$
For any $\epsilon  >0$, there is a reduction from {\sc AdditivePrice}$({\cal V}, \epsilon)$ to approximating {\sc RestrictedPrice}$(\hat{{\cal V}}, {\cal P})$ to within an additive error of {$\epsilon \over 6$}, where
\begin{itemize}
\item $\hat{{\cal V}}=\{\hat{v}_i\}_{i\in[n]}$ is a collection of mutually independent random variables that are supported on a common set of cardinality $O\left(\frac{u_{max}^{2}\log r}{\epsilon^{2}}\right)$; 

\item $|{\cal P}| = O\left( {{u_{max}^{2}} \log{r}\over \epsilon^{2}}\right)$ and $\max_{x \in {\cal P}}{x} \le {7 \over 6} u_{max}$.
\end{itemize}

Moreover, if $u_{min}$ and $u_{max}$ are given explicitly as input to the reduction,\footnote{This requirement is only relevant if we have oracle access to the distributions of the $v_i$'s, as if we are given the distributions explicitly we immediately also know $u_{min}$ and $u_{max}$.} the running time of the reduction is polynomial in the description of ${\cal V}$, $\angler{u_{min}}$, $\angler{u_{max}}$, and $1/\epsilon$.
%
\end{theorem}

\noindent That the prices can be restricted to a discrete set {without hurting the revenue too much} follows immediately from a discretization lemma attributed to Nisan~\cite{ChawlaHK07}. (See also~\cite{HartlineK05} for a related discretization.) Our price discretization result is summarized in Lemma~\ref{cor:discreteprice} of Appendix~\ref{sec:discretizing prices}.  The discretization of the support of the value distributions is inspired by Nisan's lemma, and our corresponding discretization result is summarized in Lemma~\ref{lem:horizontal-discretization additive}. 

\smallskip Combining Lemma~\ref{lem:additive to balanced} and Theorem~\ref{thm:additive discretization}, we complete our reduction from {\sc AdditivePrice} to {\sc RestrictedPrice}.
\begin{corollary}\label{cor:additive to discrete}
Let $\mathcal{V}=\{v_{i}\}_{i\in[n]}$ be a collection of mutually independent random variables supported on $[0,1]$. For any $\epsilon>0$, there is a polynomial-time reduction from {\sc AdditivePrice}$(\mathcal{V},\epsilon)$ to approximating {\sc RestrictedPrice}$(\mathcal{\hat{V}},\mathcal{P})$ to within an additive error of {$\epsilon \over 18$}, where $\mathcal{\hat{V}}$ is a collection of mutually independent random variables supported on a common set of cardinality $O({\log 1/\epsilon\over \epsilon^{2}})$, $|{\cal P}|=O({\log 1/\epsilon\over \epsilon^{2}})$ and $\max_{x \in {\cal P}}x \le 7/6$.
\end{corollary}

We are now ready to prove Theorem~\ref{thm:additive ptas}, using the reduction of Corollary~\ref{cor:additive to discrete} and our algorithm from Section~\ref{sec:ptas for discretized problem}. 

\medskip \begin{prevproof}{Theorem}{thm:additive ptas}
We first perform the reduction of Corollary~\ref{cor:additive to discrete}. In the resulting instance of {\sc RestrictedPrice} both the cardinality of the support of the value distribution and the number of available prices are $O({\log 1/\epsilon\over \epsilon^{2}})$. Using $m = O( {n\cdot\log^2 1/\epsilon\over \epsilon^{5}})$ we can solve the resulting instance of {\sc RestrictedPrice} to within additive error $O(\epsilon)$ using the algorithm of Theorem~\ref{lem:running time}. The running time of the algorithm is polynomial in the input and $n^{{\log^{3} 1/\epsilon\over \epsilon^{4}}}$. 
\end{prevproof}

\section{Multiplicative PTAS} \label{sec:balanced range} \label{sec:smoothness properties}
For the remainder of our main exposition, we move on to multiplicative approximations to the item pricing problem, obtaining algorithms for {\sc Price}. In this section, we study the general problem where the values are independently distributed on a bounded range $[u_{min}, u_{max} \equiv r \cdot u_{min}]$ according to arbitrary distributions, proving Theorem~\ref{thm:general algorithm}. 

Notice that, using our results from the previous sections, we can already get an algorithm for {\sc Price}. We can first apply our reduction from Theorem~\ref{thm:additive discretization} to discretize the prices and the supports of the value distributions. Then we can use our algorithm from Theorem~\ref{lem:dynamic programming revenue} to solve the discretized problem. However to convert the additive approximation of this algorithm to a multiplicative one, we need to choose the approximation to be no worse than $\epsilon \cdot u_{min}$. This requirement forces the support of the discretized value distributions to be $\Omega(r^{2}\log r/\epsilon^2)$ and the discrete set of prices to also have cardinality $\Omega(r^{2}\log r/\epsilon^2)$. Hence, the algorithm has running time polynomial in $n^{r^{4}\log^{2} r/\epsilon^4}$. 

In this section, we present a stronger discretization result, reducing the size of the support of the value distributions and the cardinality of the price set to linear in $\log r$. With this new discretization, we can speed up the running time of our algorithm to  $n^{ {\rm poly}(\log r, 1/\epsilon)}$.
Our improved discretization reduction is presented below, and proved in Appendix~\ref{app:discretizationreduction for multiplicative}.
\begin{theorem}[Price/Value Distribution Discretization] \label{thm:discretization}
Let ${\cal V}=\{v_i\}_{i\in[n]}$ be a collection of mutually independent random variables supported on a bounded range $[u_{min},u_{max}] \subset \mathbb{R}_+$, and let $r={u_{max}\over u_{min} } \ge 1.$
For any $\epsilon  \in \left(0,\frac{1}{(4 \lceil \log_2 r\rceil) ^{1/6}}\right)$, there is a reduction from {\sc Price}$({\cal V}, \epsilon)$ to the problem of approximating {\sc RestrictedPrice}$(\hat{{\cal V}}, {\cal P})$ to within a factor of {$(1-{\epsilon \over 16})$}, where
\begin{itemize}
\item $\hat{{\cal V}}=\{\hat{v}_i\}_{i\in[n]}$ is a collection of mutually independent random variables that are supported on a common set of cardinality $O\left(\frac{\log r}{\epsilon^{16}}\right)$; 

\item $|{\cal P}_{}| = O\left( {\log{r}\over \epsilon^{2}}\right)$.
\end{itemize}
Moreover, if $u_{min}$ and $u_{max}$ are given explicitly as input to the reduction,\footnote{This requirement is only relevant if we have oracle access to the distributions of the $v_i$'s, as if we are given the distributions explicitly we immediately also know $u_{min}$ and $u_{max}$.} the running time of the reduction is polynomial in the description of ${\cal V}$, $\angler{u_{min}}$, $\angler{u_{max}}$, and $1/\epsilon$.
%
%
%
\end{theorem}

Combining our discretization from Theorem~\ref{thm:discretization} with our algorithm from Theorem~\ref{lem:dynamic programming revenue}, it is easy to show Theorem~\ref{thm:general algorithm}. We only sketch the proof here, providing a formal proof in Appendix~\ref{app:proof of general multiplicative PTAS}.\\

\begin{prevproof}{Theorem}{thm:general algorithm} 
(sketch) We first perform the reduction of Theorem~\ref{thm:discretization} to get an instance of {\sc RestrictedPrice} where both the values and the prices come from discrete sets of cardinality $O({{\log r}\over \poly(\epsilon)})$. Using the algorithm of Theorem~\ref{lem:dynamic programming revenue}, we can then approximately solve this instance to within a factor of $1-O(\epsilon)$ in time polynomial in the input and $n^{{{\log^{2} r}\over \poly(\epsilon)}}$.
%
%
%
%
\end{prevproof}

\section{Extreme Values of MHR Distributions}\label{sec:truncate}

We reduce the problem of finding a near-optimal price vector for value distributions that are MHR to finding a near-optimal price vector for value distributions that are supported on a bounded range $[u_{min},u_{max}]$, where $u_{max}/u_{min}$ is only a function of the desired approximation $\epsilon>0$.  More precisely, we establish the following reduction.
\begin{theorem}[From MHR to Bounded Distributions]\label{thm:reduction MHR to balanced} 
Let $\mathcal{V}=\{v_i\}_{i\in[n]}$ be a collection of mutually independent  \textsl{MHR} random variables. Then there exists some $\beta=\beta({\cal V})>0$ such that for all $\epsilon\in(0,1/4)$, there is a reduction from {\scshape Price}$(\mathcal{V},c \epsilon \log_2({1 \over \epsilon}))$ to {\scshape Price}$(\tilde{\mathcal{V}},\epsilon)$, where $\tilde{\cal V}:=\{\tilde{v}_i\}_i$ is a collection of mutually independent random variables supported on the set $[{\epsilon \over 2} \beta, 2 \log_2{1\over \epsilon} \beta]$, and $c$ is some absolute constant.\footnote{{Clearly, by plugging $\epsilon = O({\hat{\epsilon} \over \log_2 {1 / \hat{\epsilon}}})$ into our reduction, we obtain a reduction from {\sc Price}$({\cal V},\hat{\epsilon})$ to {\sc Price}$(\tilde{\cal V},O({\hat{\epsilon} \over \log_2 {1 / \hat{\epsilon}}}))$, for any desired $\hat{\epsilon}$. We phrased our theorem as a reduction from {\scshape Price}$(\mathcal{V},c \epsilon \log_2({1 \over \epsilon}))$ to {\scshape Price}$(\tilde{\mathcal{V}},\epsilon)$ only to have better expressions in the supports of the $\tilde{v}_i$'s.}}

Moreover, $\beta$ is efficiently computable from the distributions of the $v_i$'s,\notshow{\yangnote{Remove:(whether we are given the distributions explicitly, or  we have oracle access to them,)}} and, for all $\epsilon$, the running time of the reduction is polynomial in the size of the input and ${1 \over \epsilon}$. 
\end{theorem}
We discuss the essential elements of our reduction below.  Most crucially, the reduction is enabled by the following theorem, characterizing the extreme values of a collection of independent MHR distributions. 
\begin{theorem}[Extreme Values of MHR distributions]\label{thm:extreme MHR}
Let $X_1,\ldots,X_n$ be a collection of independent random variables whose distributions are MHR. Then there exists some anchoring point $\beta$ such that $\Pr[ \max_i\{X_i\} \ge \beta/2] \ge 1-{1\over \sqrt{e}}$ and
\begin{align}\int_{2\beta \log_2{1/\epsilon}}^{+\infty} t \cdot f_{\max_i\{X_i\}}(t) dt \le 36 \beta \epsilon \log_2{1/\epsilon},~\text{for all }\epsilon \in (0,1/4), \label{eq:extreme MHR}
\end{align}
where $f_{\max_i\{X_i\}}(t)$ is the probability density function of $\max_i\{X_i\}$.
Moreover, $\beta$ is efficiently computable from the distributions of the $X_i$'s. 
\end{theorem}
 
{Theorem~\ref{thm:extreme MHR}, whose proof is given in Appendix~\ref{sec:proof of extreme MHR}, implies that, for $\epsilon$ sufficiently small, at least a $(1-\epsilon)$-fraction of $\mathbb{E}[\max_i{X_i}]$ is contributed to by values that are no larger than $\mathbb{E}[\max_i{X_i}]  \cdot O(\log_2{1 \over \epsilon}).$} Our result is quite surprising, especially for the case of non-identically distributed MHR random variables. Why should most of the contribution to $\mathbb{E}[\max_i{X_i}]$ come from values that are close ({\em within a function of $\epsilon$ only}) to the expectation, when the underlying random variables $X_i$ may concentrate on widely different supports? To obtain the theorem one needs to understand how the tails of the distributions of a collection of independent MHR random variables contribute to the expectation of their maximum. Our proof technique is intricate, defining a tournament between the tails of the distributions. Each round of the tournament ranks the remaining distributions according to the size of their tails, and eliminates the lightest half of the distributions. The threshold $\beta$ is then obtained by some side-information that the algorithm records in every round.  

Given our understanding of the extreme values of MHR distributions, our reduction of Theorem~\ref{thm:reduction MHR to balanced} from MHR to bounded distributions proceeds in the following steps:
\begin{itemize}
\item We start with the computation of the threshold $\beta$ specified by Theorem~\ref{thm:extreme MHR}. This computation can be done efficiently, as stated in the statement of the theorem. Given that $\Pr[ \max_i\{X_i\} \ge \beta/2]$ is bounded away from $0$, {the revenue from pricing every item at $\beta/2$ is $\Omega(\beta)$, hence the optimal revenue is also $\Omega(\beta)$.} See Appendix~\ref{sec:OPT vs Beta} for the precise lower bound we obtain. Such lower bound is useful as it implies that, if our transformation loses revenue that is a small fraction of $\beta$, this corresponds to a small fraction of optimal revenue lost. 


\item Next, using~\eqref{eq:extreme MHR} we show that, for all $\epsilon >0$, if we restrict the prices to lie in the range $[\epsilon\cdot \beta,2\log_2(\frac{1}{\epsilon})\cdot \beta]$, we only lose a $O(\epsilon \log_21/\epsilon)$ fraction of the optimal revenue; this step is detailed in Appendix~\ref{sec:bounding the prices}.

\item Finally, we show that we can efficiently transform the given MHR random variables  $\{v_{i}\}_{i\in[n]}$ into a new collection of random variables $\{\tilde{v}_{i}\}_{i\in[n]}$ that take values in  $[{\epsilon \over 2}\cdot \beta, 2\log_2(\frac{1}{\epsilon})\cdot\beta]$ and satisfy the following: a near-optimal price vector for the setting where the buyer's values are distributed as $\{\tilde{v}_{i}\}_{i\in[n]}$ can be  efficiently transformed into a near-optimal price vector for the original setting, i.e. where the buyer's values are distributed as $\{{v}_{i}\}_{i\in[n]}$. This step is detailed in Appendix~\ref{sec:truncating}. 
\end{itemize}

Theorem~\ref{thm:ptas mhr} is established by combining the reduction of Theorem~\ref{thm:reduction MHR to balanced} with our algorithm for bounded distributions of Theorem~\ref{thm:general algorithm}. See Appendix~\ref{sec:overall ptas}.

\section{Extreme Values of Regular Distributions}\label{sec:regular}

We reduce the problem of finding a near-optimal price vector for value distributions that are regular to finding a near-optimal price vector for value distributions that are supported on a bounded range $[u_{min},u_{max}]$ satisfying $u_{max}/u_{min} \le 16n^8/\epsilon^{4}$, where $n$ is the number of distributions and {$\epsilon$ is the desired approximation}. It is important to notice that our bound on the ratio $u_{max}/u_{min}$ does not depend on the distributions at hand,  just their number and the required approximation. We also emphasize that the given regular distributions may be supported on $[0,+\infty)$, so it is a priori not clear if we can truncate these distributions to any finite set without losing substantial revenue. Our reduction is the following.

\begin{theorem}[Reduction from Regular to ${\rm Poly}(n)$-Bounded Distributions]\label{thm:regreduction}
Let $\mathcal{V}=\{v_i\}_{i\in[n]}$ be a collection of mutually independent  \textsl{regular} random variables. Then there exists some $\alpha=\alpha({\cal V})>0$ such that, for any {$\epsilon\in(0,1)$}, there is a reduction from {\scshape Price}$(\mathcal{V},\epsilon)$ to {\scshape Price}$(\tilde{\mathcal{V}},\epsilon-\Theta(\epsilon/n))$, where $\tilde{\mathcal{V}}=\{\tilde{v}_i\}_{i\in [n]}$ is a collection of mutually independent random variables that are supported on $[{\epsilon\alpha\over 4n^{4}}, {4n^{4}\alpha\over \epsilon^{3}}]$.

Moreover, $\alpha$ is efficiently computable from the distributions of the $v_i$'s,\notshow{\yangnote{Remove: (whether we have the distributions of the $v_i$'s explicitly, or have oracle access to them.)}} and, for all $\epsilon$, the  running time of the reduction is polynomial in the size of the input and $1/\epsilon$. 
%
\end{theorem}

Our reduction is based on the following extreme value theorem for regular distributions, whose proof is provided in Appendix~\ref{sec:regextremevalue}.  Immediately following the statement of the theorem we sketch how it is used to establish our reduction, whose detailed proof is in Appendix~\ref{app: reduction for regular}. Section~\ref{app:discussion of xtreme regular} gives other example applications of the theorem to illustrate its usefulness in bounding extreme values of regular distributions.

\begin{theorem}[Homogenization of the Extreme Values of Regular Distributions]\label{thm:regextremevalue}
Let $\{X_{i}\}_{i\in[n]}$ be a collection of mutually independent regular random variables, where $n\ge2$. Then there exists some $\alpha=\alpha(\{X_i\}_{i\in[n]})$ such that:
\begin{enumerate}
\item $\alpha$ has the following ``anchoring'' properties:
\begin{itemize}
\item for all $\ell \ge 1$, $\Pr[X_{i}\geq \ell \alpha ]\leq 2/(\ell n^{3})$, for all $i\in[n]$;
 \item $\alpha/n^{3}\leq c\cdot \max_{z} (z\cdot \Pr[\max_{i}\{X_{i}\}\geq z])$, where $c$ is an absolute constant.
 \end{itemize}
 \item For all {$\epsilon\in (0,1)$}, the tails beyond $ {2n^{2}\alpha\over \epsilon^{2}}$ can be ``homogenized'', i.e.
 \begin{itemize}
 \item for any integer $m\leq n$, thresholds $t_{1},\ldots, t_{m}\geq t \geq {2n^{2}\alpha\over \epsilon^{2}}$, and index set  $\{a_1,\ldots,a_m\}\subseteq [n]$:
\begin{align*}
\sum_{i=1}^{m} t_{i} \Pr[X_{a_{i}}\geq t_{i}]\leq \left(t-{2\alpha\over \epsilon}\right)\cdot \Pr\left[\max_{i \in [m]}\{X_{a_{i}}\} \geq t\right]+{7\epsilon \over n}\cdot\left({2\alpha \over \epsilon}\cdot\Pr\left[\max_{i \in [m]}\{X_{a_{i}}\}\geq {2\alpha \over \epsilon}\right] \right).
\end{align*}
\end{itemize}
\end{enumerate}
Furthermore, $\alpha$ is efficiently computable from the distributions of the $X_i$'s. \notshow{\yangnote{Remove: (whether we are given the distributions explicitly, or have oracle access to them.)}}
\end{theorem}
\noindent Given our homogenization theorem, our reduction of Theorem~\ref{thm:regreduction} is obtained as follows.
\begin{itemize}
\item First, we compute the threshold $\alpha$ specified in Theorem~\ref{thm:regextremevalue}.
This  can be done  efficiently as stated in Theorem~\ref{thm:regextremevalue}. Now given the second anchoring property of $\alpha$, we obtain an $\Omega(\alpha/n^{3})$ lower bound to the optimal revenue. Such a lower bound is useful as it implies that we can ignore prices below some $O(\epsilon\alpha/n^{3})$, without losing more than an $\epsilon$-fraction of revenue. 
\item Next, using the homogenization part of Theorem~\ref{thm:regextremevalue}, we show that, if we restrict a price vector to lie in {$[\epsilon\alpha/n^{4},2n^{2}\alpha/\epsilon^{2}]^{n}$}, we only lose a $O({\epsilon\over n})$ fraction of the optimal revenue. This step is detailed in Appendix~\ref{sec:regrestrictprice}.
\item Finally, we show that we can efficiently transform the input regular random variables $\{v_{i}\}_{i\in[n]}$ into a new collection of random variables $\{\tilde{v}_{i}\}_{i\in[n]}$ that are supported on $[{\epsilon\alpha\over 4n^{4}}, {4n^{4}\alpha\over \epsilon^{3}}]$ and satisfy the following: a near-optimal price vector for when the buyer's values are distributed as $\{\tilde{v}_{i}\}_{i\in[n]}$ can be efficiently transformed into a near-optimal price vector for when the buyer's values are distributed as $\{v_{i}\}_{i\in[n]}$. This step is detailed in Appendix~\ref{sec:regboundingdistribution}, and Appendix~\ref{app: finish regular reduction} concludes the proof of Theorem~\ref{thm:regreduction}.
\end{itemize}

Theorem~\ref{thm:quasi ptas regular} is established by combining the reduction of Theorem~\ref{thm:regreduction} with our algorithm for bounded distributions of Theorem~\ref{thm:general algorithm}. See Appendix~\ref{sec:overall ptas}.

\subsection{Discussion of Theorem~\ref{thm:regextremevalue}} \label{app:discussion of xtreme regular}

We give a couple of applications of Theorem~\ref{thm:regextremevalue} to gain some intuition about its content:

\begin{itemize}

\item Suppose that we set all the $t_{i}$'s equal to $t \ge 2n^2 \alpha/\epsilon^2$. In this case, the homogenization property of Theorem~\ref{thm:regextremevalue} implies that the union bound is essentially tight for $t$ large enough, as ${14 \alpha\over tn}\cdot\Pr\left[\max_{i \in [m]}\{X_{a_{i}}\}\geq {2\alpha \over \epsilon}\right]$ in the following calculation gets arbitrary close to $0$:
\begin{align*}\Pr\left[\max_{i \in [m]}\{X_{a_{i}}\} \geq t\right]&\leq \left(\sum_{i=1}^{m}  \Pr[X_{a_{i}}\geq t]\right)\\&\leq \left({t-{2\alpha\over \epsilon}\over t}\right)\cdot \Pr\left[\max_{i \in [m]}\{X_{a_{i}}\} \geq t\right]+{7\epsilon \over tn}\cdot\left({2\alpha \over \epsilon}\cdot\Pr\left[\max_{i \in [m]}\{X_{a_{i}}\}\geq {2\alpha \over \epsilon}\right] \right)\\
&\le  \Pr\left[\max_{i \in [m]}\{X_{a_{i}}\} \geq t\right]+{14 \alpha\over tn}\cdot\Pr\left[\max_{i \in [m]}\{X_{a_{i}}\}\geq {2\alpha \over \epsilon}\right].
\end{align*}
This is not surprising, since for all $i$, the event $X_{a_{i}}\geq t$ only happens with tiny probability, by the anchoring property of $\alpha$.

\item Now let's try to set all the $t_{i}$'s to the same value $t' >t\ge 2n^2 \alpha/\epsilon^2$. The homogenization property can be used to show that the probability of the event $\max_{i \in [m]}\{X_{a_{i}}\}\geq t'$ scales inverse proportionally with $t'$. Essentially this says that the tails of $\max_{i \in [m]}\{X_{a_{i}}\}$ are not fatter than those of the equal revenue distribution.\footnote{Recall that the {\em equal revenue distribution} is supported on $[1,+\infty]$ and has cumulative density function $F(x)=1-{1 \over x}$.}
\begin{align*}
\Pr\left[\max_{i \in [m]}\{X_{a_{i}}\} \geq t'\right]\leq& \sum_{i=1}^{m} \Pr[X_{a_{i}}\geq t']\\
\leq& \left({t-{2\alpha\over \epsilon}\over t'}\right)\cdot \Pr\left[\max_{i \in [m]}\{X_{a_{i}}\} \geq t\right]+{7\epsilon \over t'n}\cdot\left({2\alpha \over \epsilon}\cdot\Pr\left[\max_{i \in [m]}\{X_{a_{i}}\}\geq {2\alpha \over \epsilon}\right] \right)\\
\le & {1 \over t'} \cdot \left[t\cdot \Pr\left[\max_{i \in [m]}\{X_{a_{i}}\} \geq t\right]+{7\epsilon \over n}\cdot\left({2\alpha \over \epsilon}\cdot\Pr\left[\max_{i \in [m]}\{X_{a_{i}}\}\geq {2\alpha \over \epsilon}\right] \right)\right].
\end{align*}

A similar bound would follow from Markov's inequality, if the expression inside the brackets were within a constant factor of $\mathbb{E}[\max_{i \in [m]}\{X_{a_{i}}\}]$. The result is interesting as it is possible for that expression to be much smaller than $\mathbb{E}[\max_{i \in [m]}\{X_{a_{i}}\}]$. For example, if $m=1$ and $X_{a_1}$ is distributed according to the equal revenue distribution, the expectation of $X_{a_1}$ is $+\infty$, while the expression inside the brackets is $1+{7 \epsilon \over n}$.
\end{itemize}

\appendix
\newpage 
\section*{\huge Appendix}
\section{Roadmap to the Appendix}\label{app:roadmap}
Appendix~\ref{sec:model}  describes several computational models of accessing a value distribution, explaining what it means for an algorithm with each type of access to be ``computationally efficient'' or ``take time polynomial in the input.'' 

Appendix~\ref{sec:algorithm analysis} contains a formal description and analysis of our dynamic programming approach for {\sc RestrictedPrice}, culminating in the proof of Theorem~\ref{lem:dynamic programming revenue}. 

Appendix~\ref{appendix:balanced range} provides several reductions among item pricing problems, whose goal is to discretize some aspect of the problem such as the support of the value distributions, the probabilities they assign to their support, or the set of available prices. The appendix culminates in the reductions of Theorems~\ref{thm:additive discretization} and~\ref{thm:discretization}. 

Appendix~\ref{app:proof of general multiplicative PTAS} provides a proof of Theorem~\ref{thm:general algorithm}, our algorithm for bounded distributions. 

\smallskip The rest of the appendix is dedicated to our treatment of MHR and regular distributions. Appendix~\ref{sec:MHR to BD} provides the proof of our extreme value theorem for MHR distributions (Theorem~\ref{thm:extreme MHR}), as well as our reduction from item pricing problems with MHR distributions to item pricing problems with bounded distributions (Theorem~\ref{thm:reduction MHR to balanced}). Similarly, Appendix~\ref{appendix:regularbounding} provides the proof of our extreme value theorem for regular distributions (Theorem~\ref{thm:regextremevalue}), as well as our reduction from item pricing problems with regular distributions to item pricing problems with bounded distributions (Theorem~\ref{thm:regreduction}). The proofs of our algorithmic results for MHR and regular distributions (Theorems~\ref{thm:ptas mhr} and~\ref{thm:quasi ptas regular}) are provided in Appendix~\ref{sec:overall ptas}. The proofs of our structural results for independent MHR and regular distributions (Theorems~\ref{thm:single price constant factor}, \ref{thm:constant prices suffice} and~\ref{thm:logn prices suffice}) are provided in Appendix~\ref{sec:proofs of structural}.  Finally, Appendix~\ref{sec:MHR iid} contains the proof of our structural result for i.i.d. MHR distributions (Theorem~\ref{thm:structural theorem 3}).
\section{Access to Value Distributions, and Computational Complexity} \label{sec:model}

We consider three ways in which a distribution may be input to an algorithm, as well as what it means for the algorithm to run in time ``polynomial in the description of the distribution'' in each case.

\begin{itemize}
\item {\bf Explicitly:} In this case, the distribution has to be discrete, and we are given its support (as a list of numbers), and the probabilities that the distribution places on every element in its support. If a distribution is explicitly input to an algorithm, the algorithm is computationally efficient if it runs in time polynomial in its other inputs and the bit-complexity of the numbers required to specify the distribution, i.e. the numbers in the support of the distribution and the probabilities assigned to them.

\item {\bf As an Oracle:} In this case, we are given (potentially black-box) access to a subroutine, called an {\em oracle}, that answers queries about the value of the cumulative distribution function on a queried point. In particular, a query to the oracle consists of a point $x$ and a precision $\epsilon$, and the oracle outputs a value of bit-complexity polynomial in the bit-complexity of $x$ and $\epsilon$, which is within $\epsilon$ from the value of the cumulative distribution function at point $x$. Moreover, we assume that we are given an {\em anchoring point} $x^*$ such that the value of the cumulative distribution at that point is between two a priori known absolute constants $c_1$ and $c_2$, such that $0<c_1<c_2<1$. Having such a point is necessary, as otherwise it would be computationally impossible to find any interesting point  in the support of the distribution (i.e. any point where the cumulative is different than~$0$~or~$1$). 

If a distribution is provided to an algorithm as an oracle, the algorithm is {computationally} efficient if it runs in time polynomial in its other inputs and the bit complexity of $x^*$, ignoring the time spent by the oracle to answer queries (since this is not under the algorithm's control). 

If, as it so happens in practice, we have a closed-form description of our input distribution, e.g. if our distribution is ${\cal N}(\mu, \sigma^2)$, we think of it as given to us as an oracle, answering queries of the form $(x,\epsilon)$ as specified above. In most common cases, such an oracle can be implemented so that it also runs efficiently in the bit-complexity  of the query to the oracle.

\item {\bf Sample Access:} In this case, our only access to the distribution is our ability to take samples from it. It is easy to see that sample access to a distribution can be reduced to  oracle access as follows. {Suppose we have an algorithm ${\cal A}$ designed to work with oracle access to a distribution, and let $B$ be a bound on the total number of queries that the algorithm may make to the oracle. ($B$ is always upper bounded by the running time of the algorithm.) Suppose now that instead of oracle access we have sample access to the distribution. Here is how we can fix this: For any query $(x,\epsilon)$ that ${\cal A}$ needs to make to the oracle, we can simply take ${1 \over 2\epsilon^2} \ln ({ 2B \over \delta})$ samples from the distribution to estimate the cumulative distribution function at $x$. By Chernoff bounds, our estimate will have error greater than $\epsilon$ with probability at most $\delta \over B$. So a union bound shows that all (at most $B$) queries of the algorithm will have error smaller than $\epsilon$ with probability at least $1-\delta$. (We can tune this probability to be as close to one as we want at a cost of a factor of $\log {1\over \delta}$ in the running time.)} It is also easy to find an anchoring point. If we take many samples from the distribution and pick the median as the anchoring point, with very high probability the value of the cumulative distribution at this point is between $1/3$ and $2/3$.

Given the above, whenever we have sample access to a distribution we will pretend to have instead oracle access to it, and we will say that an algorithm is computationally efficient using the same criterion we used for oracle access.
\end{itemize}

{\paragraph{Polynomial-Time Reductions Involving Value Distributions.} This paper provides several polynomial-time reductions among item pricing problems. Recall from Section~\ref{sec:prelim} that a reduction contains an algorithm ${\cal A}$ that takes as input an instance of the item pricing problem, comprising distributions (and sometimes a restricted set of prices), and outputs another instance of the item pricing problem, comprising potentially different distributions (and prices). But what do we mean when we say that ``an algorithm ${\cal A}$ outputs a distribution $F$?'' The algorithm may either output an explicit description of the distribution or an oracle for it.\footnote{Our reductions never output a distribution by providing sample access to it.} In the former case, ${\cal A}$ must enumerate the support of the distribution and specify the probabilities assigned to every point in the support, as required by the first bullet  above. In the latter case, ${\cal A}$ outputs an oracle for $F$, i.e. the description of an algorithm that satisfies the requirements of the second bullet above. This oracle may use as subroutines the oracles of the distributions provided in the input to ${\cal A}$, if any. We will then say that ``${\cal A}$ runs in polynomial time'' if two properties are satisfied: 1. ${\cal A}$'s running time is polynomial in its input; and 2. if ${\cal A}$ outputs an oracle for some distribution ${\cal F}$, this oracle must run in time polynomial in the description of the oracle and the input $(x,\epsilon)$ to the oracle, excluding the time spent in oracles (from the input to ${\cal A}$) that the oracle may use as subroutines.} \notshow{\yangnote{\footnote{\yangnote{Yang says: I think we can also point out that our reductions almost always preserve the way to access the distributions, except when we solve {\scshape RestrictedPrice}, our reduction outputs an explicit distribution by rounding the probabilities. }}}}

\section{The Algorithm for Discrete Distributions} \label{sec:algorithm analysis}


\subsection{The Generic DP Step: Add an Item and Discretize Probabilities} \label{sec:DP step}

In Section~\ref{sec:ptas for discretized problem}, we described our intended meaning for the Boolean function $g(i,{\Pr})$. Here we explain how to compute $g$ using dynamic programming. Our algorithm works bottom-up (i.e. from smaller to larger $i$'s), filling in $g$'s table so that the following recursive conditions are met. 

\smallskip\noindent$\bullet$ If $i>1$, we set $g(i,{\Pr})=1$ iff there is a price $p^{(j)}$ and a distribution ${\Pr}'$ so that the following hold:
\begin{enumerate}
\item $g\left(i-1,{\Pr}'\right)=1$.
    \item Suppose that {$P_{i-1}$} is the price vector stored at cell $(i-1,\Pr')$ of the table, namely that under price vector $P_{i-1}$ the (winning-value, winning-price) distribution for the prefix $1\ldots i-1$ of the items is ${\Pr}'$. What would happen if we assigned price $p^{(j)}$ to the $i$-th item? {If the gap between the winning-value and winning-price among the first $i-1$ items is larger than the gap between the value and price for the i-th item, the winning-value and winning-price would remain the same. Otherwise, they will become the value and price for the i-th item. Based on this}, we can compute the resulting (winning-value, winning-price) distribution $\{{\Pr}''_{i_1,i_2}\}_{i_1\in[k_1],\ i_2\in [k_2]}$ for the prefix $1\ldots i$ from just ${\Pr}'_{i_1,i_2}$ and the distribution $\hat{F}_i$ of item $i$. Indeed:
    \begin{align}&{\Pr}''_{i_1,i_2} = {\Pr}'_{i_1,i_2}\cdot \Pr_{v_i \sim \hat{F}_i}[v_i-p^{(j)}<v^{(i_1)}-p^{(i_2)}]\notag\\&~~~~~~~~~~~~~~~~~~+\left(\sum_{
    \begin{subarray}{c}
         j_1,j_2\\
          s.t.\ v^{(j_1)}-p^{(j_2)}\\
          \quad \leq v^{(i_1)}-p^{(i_2)}
    \end{subarray}} {\Pr}'_{j_1,j_2}\right) \cdot \Pr_{v_i \sim \hat{F}_i}[v_i = v^{(i_1)}] \cdot \ind_{p^{(j)}=p^{(i_2)}}. \label{eq: DP step}
    \end{align}   
    We require that ${\Pr}$ is a rounded version of ${\Pr}''$ computed as above, where all the probabilities are integer multiples of $1\over m$. The rounding should be of the following canonical form. Setting $\delta_{i_1,i_2} = {\Pr}''_{i_1,i_2}-\left\lfloor\frac{{\Pr}''_{i_1,i_2}}{1/ m}\right\rfloor\cdot \frac{1}{m}$, and $l = \Big{(}\sum_{i_1\in[k_1],\ i_2\in[k_2]}\delta_{i_1,i_2}\Big{)}\Big{/}({1\over m}),$ we will round the first $ l$ probabilities in $\{{\Pr}''_{i_1,i_2}\}_{i_1\in[k_1],[i_2]\in k_2}$ in some fixed lexicographic order up to the closest multiple of $1\over m$, and round the rest down to the closest multiple of $1\over m$.{\footnote{{Any rounding would work. We use this one just to make the description of our algorithm explicit.}}}  
\end{enumerate}
If Conditions 1 and 2 are met, we also store price vector $(P_{i-1},p^{(j)})$ in cell $g(i,{\Pr})$ of the table.\notshow{we also keep a pointer from cell $g(i,{\Pr})$ to cell $g(i-1,{\Pr}')$ of the DP table recording the price $p^{(j_2)}$ on that pointer. We use these pointers  to recover price vectors consistent with a certain distribution.}

\medskip \noindent $\bullet$ To fill in the first slice of the table corresponding to $i=1$, we use the same recursive definition given above, imagining that there is a slice $i=0$, whose cells are all $0$ except for those corresponding to the distributions ${\Pr}$ that satisfy: ${\Pr}_{i_1,i_2}=0$, for all $i_1$, $i_2$, except for the lexicographically smallest $(i_1^*,i_2^*) \in \arg\min_{(k_1,k_2)}\  v^{(k_1)}-p^{(k_2)}$, where ${\Pr}_{i_1^*,i_2^*}=1$. 

\medskip \noindent While we decribed the function $g$ recursively above, we compute it iteratively from $i=1$ through~$n$.

\subsection{Proof of Theorem~\ref{lem:dynamic programming revenue}}\label{sec:correctness of DP}
In this appendix, we prove the correctness and running time of the algorithm presented in Section~\ref{sec:ptas for discretized problem}, providing a proof of Theorem~\ref{lem:dynamic programming revenue}. Intuitively, if we did not perform any rounding of distributions, our algorithm would have been  {\em exact}, outputting an optimal price vector in $\{p^{(1)},\ldots,p^{(k_2)}\}^n$. We show next that the rounding is fine enough that it does not become detrimental to our revenue. To show this, we use the probabilistic concepts of {\em total variation distance} and {\em coupling of random variables}. Recall that the {total variation distance} between two distributions $\mathbb{P}$ and $\mathbb{Q}$ over a finite set $\mathcal{A}$ is defined as follows 
$$||\mathbb{P}-\mathbb{Q}||_{TV}=\frac{1}{2}\sum_{\alpha\in\mathcal{A}}|\mathbb{P}(\alpha)-\mathbb{Q(\alpha)}|.$$ Similarly, if $X$ and $Y$ are two random variables ranging over a finite set, their total variation distance, denoted $||X-Y||_{TV}$ is defined as the total variation distance between their distributions.

Proceeding to the correctness of our algorithm, let $P = (p_1,p_2,\cdots, p_n) \in \{p^{(1)},\ldots,p^{(k_2)}\}^n$ be an arbitrary price vector. We can use this price vector to select $n$ cells of our dynamic programming table, picking one cell per layer. The cells are those that the algorithm would have traversed if it made the decision of assigning price $p_i$ to item $i$, for all $i$. Let us call the resulting cells $cell_1,cell_2, \ldots, cell_n$.

For all $i$, we intend to compare the distributions $\left\{\widehat{\Pr}^{(i)}_{i_1,i_2}\right\}_{i_{1}\in[k_{1}],\ i_{2}\in[k_{2}]}$ and $\left\{\Pr^{(i)}_{i_1,i_2}\right\}_{i_{1}\in[k_{1}],\ i_{2}\in[k_{2}]}$, which are respectively the (winning-value,winning-price) distribution:
\begin{itemize}

\item arising when the prefix $1\ldots i$ of items with distributions $\{\hat{F}_j\}_{j=1,\ldots,i}$ is priced according to price vector $(p_1,\ldots,p_i)$; 
\item stored in $cell_i$ of the DP table.
\end{itemize}
The following lemma shows that these distributions have small total variation distance. 

{
\begin{lemma}\label{lem:closeness}
  For all $i\in[n]$, $||\Pr^{(i)}-\widehat{\Pr}^{(i)}||_{TV}\leq ik_1k_2/m.$
\end{lemma}
\begin{proof}
At a high level, our argument shows two properties for every $i$: (1) if rounding was not performed at step $i$ of the DP algorithm, the distance between $\widehat{\Pr}^{(i)}$ and $\Pr^{(i)}$ would not increase  compared to the distance between $\widehat{\Pr}^{(i-1)}$ and $\Pr^{(i-1)}$; (2) after the rounding is performed the distance increases by at most  $k_{1}k_{2}/m$. Combining the two properties, we can prove the lemma.

\smallskip Formally, we prove the lemma by induction. The base case is trivially true as ${\Pr}^{(1)}$ is just a rounding of $\widehat{\Pr}^{(1)}$  into probabilities that are multiples of $1\over m$, whereby the probability of every point in the support is modified by no more than an additive $1\over m$. 

We proceed to show the inductive step. For convenience, for all $i$, let $X_{i}$ be a random variable distributed according to $\Pr^{(i)}$, i.e. $\Pr[X_{i}=(v^{(i_{1})},p^{(i_{2})})]=\Pr^{(i)}_{i_{1},i_{2}}$ for all $i_1,i_2$, and let $\hat{X}_{i}$ be a random variable distributed according to $\widehat{\Pr}^{(i)}$.

Now suppose that the claim is true for $i$. We want to show that it holds for $i+1$. For this purpose we define an auxiliary random variable $Z_{i+1}$. $Z_{i+1}$ is a function of the random variable ${X}_{i}$ and an independent random variable $\hat{v}_{i+1}$ distributed according to $\hat{F}_{i+1}$. If $\hat{v}_{i+1}-p_{i+1} \ge X_i(1)-X_i(2)$, we set $Z_{i+1}=(\hat{v}_{i+1},p_{i+1})$, otherwise we set $Z_{i+1}={X}_{i}$. Clearly, if  we replaced $X_{i}$ by $\hat{X}_{i}$ in this definition, we would get a random variable with the same distribution as $\hat{X}_{i+1}$. 

Now consider the following coupling of $\hat{X}_{i+1}$ and $Z_{i+1}$. Use the optimal coupling of $\hat{X}_{i}$ and $X_{i}$. Then generate both $\hat{X}_{i+1}$ and $Z_{i+1}$ using the above procedure with the same sample for $\hat{v}_{i+1}$. It is clear then that, conditioning on $X_{i}=\hat{X}_{i}$, $\hat{X}_{i+1}=Z_{i+1}$ with probability $1$. So 
\begin{align}
||\hat{X}_{i+1}-Z_{i+1}||_{TV}\leq \Pr[\hat{X}_{i+1}\neq Z_{i+1}]\leq \Pr[X_{i}\neq \hat{X}_{i}]=||X_{i}-\hat{X}_{i}||_{TV}, \label{eq:coupling inequality 1}
\end{align}
where the first inequality is true under any coupling, the second inequality is true for our particular coupling, and the last equality is true because we assumed an optimal coupling of $X_i$ and $\hat{X}_i$.

On the other hand, we know that, if we round the distribution of $Z_{i+1}$ into integer multiples of $1/m$, we will get the distribution of ${X}_{i+1}$. Therefore, 
\begin{align}
||Z_{i+1}-{X}_{i+1}||_{TV}\leq k_{1}k_{2}/m   \label{eq:coupling inequality 2}
\end{align}

Combining~\eqref{eq:coupling inequality 1} and~\eqref{eq:coupling inequality 2}, the triangle inequality implies that $||X_{i+1}-\hat{X}_{i+1}||_{TV}\leq ||X_{i}-\hat{X}_{i}||_{TV}+k_{1}k_{2}/m$, which completes the inductive step.
\end{proof}

\notshow{\begin{lemma}\label{lem:closeness}
   Let $\{X_i\}_{i\in[n]}$ and $\{\hat{X}_i\}_{i\in[n]}$ be two collections of $(k_1k_2)$-dimensional random unit vectors defined in terms of $\left\{\Pr^{(i)}\right\}_i$ and $\left\{\widehat{\Pr}^{(i)}\right\}_i$ as follows: for all $i\in[n]$, $i_1 \in [k_1]$, $i_2 \in [k_2]$, and $\ell = (i_1-1)\cdot k_2 + i_2\in[k_1k_2]$, we set $\Pr[X_i = e_{\ell}] = \Pr^{(i)}_{i_1,i_2}$ and $\Pr[\hat{X}_i = e_{\ell}] = \widehat{\Pr}^{(i)}_{i_1,i_2}$, where $e_{\ell}$ is the unit vector along dimension $\ell$. 
   
   Then, for all $i$, $$||X_i-\hat{X}_i||_{TV}\leq nk_1k_2/m.$$
\end{lemma}

\begin{proof}
We prove this by induction. For the base case ($i=1$), observe that $||X_1-\hat{X}_1||_{TV}\leq k_1k_2/m$, because $\widehat{\Pr}^{(1)}$ is just a rounding of ${\Pr}^{(1)}$ into probabilities that are multiples of $1\over m$, whereby the probability of every point in the support  is not modified by more than an additive $1\over m$. 

For the inductive step, it suffices to argue that for all $i\in[n-1]$, $$||X_{i+1}-\hat{X}_{i+1}||_{TV}-||X_i-\hat{X}_i||_{TV}\leq k_1k_2/m.$$ To show this, we are going to consider two auxiliary random variables $Y_{i+1}$ and $Z_{i+1}$:
\begin{itemize} 
\item $Y_{i+1}$ is a $(k_1k_2)$-dimensional random unit vector defined as follows: for $i_2^*\in[k_{2}]$ such that $p^{(i_2^*)}=p_{i+1}$, if $\ell = (i_1-1)\cdot k_2 + i_2^*$, then \yangnote{$\Pr[Y_{i+1} = e_{\ell}] = \hat{F}_{i+1}(v^{(i_1)})-\hat{F}_{i+1}(v^{(i_1-1)})$}, otherwise $\Pr[Y_{i+1} = e_{\ell}]=0$.~\footnote{In other words, $Y_{i+1}$ is the random unit vector analog of $\hat{F}_{i+1}$.}

\item $Z_{i+1}$ is a $(k_1k_2)$-dimensional random unit vector defined as follows
\begin{align}\Pr[Z_{i+1} = e_{(i_1-1)\cdot k_2 +i_2}] = \sum_{\begin{subarray}{c}
    j_1,j_2\\ 
    s.t.\ v^{(j_1)}-p^{(j_2)}\\
    \quad\leq v^{(i_1)}-p^{(i_2)}
    \end{subarray}}\Pr[\hat{X}_i+Y_{i+1} = e_{(i_1-1)\cdot k_2 +i_2}+e_{(j_1-1)\cdot k_2 +j_2}], \label{eq:lala1}
    \end{align}
    where for the purposes of the above definition $\hat{X}_i$ and $Y_{i+1}$ are taken to be independent.~\footnote{$Z_{i+1}$ is the random unit vector analog of the (winning-value,winning-price) distribution for the prefix $1\ldots i+1$, if item $i+1$ is assigned price $p_{i+1}$ and the (winning-value,winning-price) distribution for the prefix $1\ldots i$ is $\widehat{\Pr}_i$.}
  
  \end{itemize}
  
  \smallskip We claim that 
  \begin{align}
  ||X_{i+1}-Z_{i+1}||_{TV}\leq ||(X_i+Y_{i+1})-(\hat{X}_i+Y_{i+1})||_{TV}. \label{eq:kostasss}
  \end{align} Indeed, we can define $X_{i+1}$ as follows: for all $i_1$, $i_2$, \begin{align}\Pr[X_{i+1} = e_{(i_1-1)\cdot k_2 +i_2}] = \sum_{\begin{subarray}{c}
  j_1,j_2\\ 
  s.t.\ v^{(j_1)}-p^{(j_2)}\\
  \quad\leq v^{(i_1)}-p^{(i_2)}
  \end{subarray}}\Pr[X_i+Y_{i+1} = e_{(i_1-1)\cdot k_2 +i_2}+e_{(j_1-1)\cdot k_2 +j_2}],\label{eq:lala2}
  \end{align}
    where for the purposes of the above definition ${X}_i$ and $Y_{i+1}$ are taken to be independent. Now \eqref{eq:lala1}, \eqref{eq:lala2} and the definition of the total variation distance imply~\eqref{eq:kostasss}.
    
    To finish our proof, suppose further that we couple $X_i$ and $\hat{X}_i$ optimally. By the {\em optimal coupling theorem} our joint distribution satisfies  $\Pr[X_i\neq\hat{X}_i]= ||X_i-\hat{X}_i||_{TV}$. Defining $Y_{i+1}$ in the same space as $X_i$ and $\hat{X}_i$ so that $Y_{i+1}$ is independent from both $X_i$ and $\hat{X}_i$, the {\em coupling lemma} implies:
\begin{align*}
     ||(X_i+Y_{i+1})-(\hat{X}_i+Y_{i+1})||_{TV} &\leq \Pr[(X_i+Y_{i+1})\neq(\hat{X}_i+Y_{i+1})]\\
     &\leq \Pr[X_i\neq\hat{X}_i]\\
     &= ||X_i-\hat{X}_i||_{TV} \qquad \text{(since\ $X_i$ and $\hat{X}_i$ are optimally coupled)}
\end{align*}
\noindent On the other hand, it is easy to see that $||Z_{i+1}-\hat{X}_{i+1}||_{TV}\leq k_1k_2/m$. Indeed:
\begin{itemize}
\item $Z_{i+1}$ is the random unit vector analog of the (winning-value,winning-price) distribution for the prefix $1\ldots i+1$ of items, if item $i+1$ is assigned price $p_{i+1}$ and the (winning-value,winning-price) distribution for the prefix $1\ldots i$ is $\widehat{\Pr}_i$;

\item $\hat{X}_{i+1}$ is the random unit vector analog of the same (winning-value,winning-price) distribution as above, except after rounding that distribution according to the rounding rule used in our dynamic program;

\item  the rounding changes every probability in the support by at most an additive $1/m$. 
\end{itemize}
\medskip Combining the above and using the triangle inequality, we obtain $$||X_{i+1}-\hat{X}_{i+1}||_{TV}\leq ||Z_{i+1}-\hat{X}_{i+1}||_{TV}+||X_{i+1}-Z_{i+1}||_{TV}\leq k_1k_2/m+||X_i-\hat{X}_i||_{TV}.$$
\end{proof}}

\begin{prevproof}{Theorem}{lem:dynamic programming revenue}
{\bf Correctness:} Let $P^*$ be an optimal price vector for the instance of {\sc RestrictedPrice} resulting after the reduction of Lemma~\ref{lem:verticaldiscretization} is applied to discretize the $F_i$'s into $\hat{F}_i$'s. Let $cell^*$ be the cell at layer $n$ of the DP table corresponding to the price vector $P^*$. 
Lemma~\ref{lem:closeness} implies that
$$\sum_{i_1\in[k_1],\ i_2\in[k_2]} |{\Pr}^{(n)}_{i_1,i_2}-\widehat{\Pr}^{(n)}_{i_1,i_2}|\leq nk_1k_2/m,$$
where $\widehat{\Pr}^{(n)}$ is the true (winning-value,winning-price) distribution corresponding to price vector $P^*$ and ${\Pr}^{(n)}$ is the distribution stored in cell $cell^*$. Clearly, the expected revenues $\mathcal{R}_{P^*}$ and ${\cal R}_{cell^*}$ from these two distributions are related, as follows
$$|\mathcal{R}_{P^*}-{\cal R}_{cell^*}| \leq \sum_{i_1\in[k_1],\ i_2\in[k_2]} |{\Pr}^{(n)}_{i_1,i_2}-\widehat{\Pr}^{(n)}_{i_1,i_2}|\cdot p^{(i_2)}\leq {nk_1k_2\over m} \cdot \max_i\{p^{(i)}\}.
$$

Now let $cell'$ be the cell at layer $n$ of the DP table that has the highest expected revenue, and let{ $P'$ be the price vector stored in $cell'$.} Using the same notation as above, call ${\cal R}_{cell'}$ the revenue from the distribution stored at $cell'$ and $\mathcal{R}_{P'}$ the revenue from price vector $P'$. Then we have the following:
\begin{align}
{\cal R}_{cell'} &\ge {\cal R}_{cell^*};~~~~~~~~~~~~~~~~~~~~~~\text{(by the optimality of $cell'$)}\\
|\mathcal{R}_{P'}-{\cal R}_{cell'}|&\leq {n k_{1}k_{2}\over  m}\cdot \max_i\{p^{(i)}\}.~~~~~~~~~~\text{(using Lemma ~\ref{lem:closeness}, as above)}
\end{align}

Putting all the above together, we obtain that  
\begin{align}
\mathcal{R}_{P'} \ge \mathcal{R}_{P^*}-{2nk_{1}k_{2}\over  m}\cdot \max_i\{p^{(i)}\}. \label{eq:ah costa costa}
\end{align}

Hence, the price vector $P'$ output by our algorithm achieves revenue $\mathcal{R}_{P'}$ that is close to the optimal revenue $\mathcal{R}_{P^*}$ for the discretized distributions $\{\hat{F}_i\}_i$. We now have to relate this revenue to the optimal revenue for the distributions $\{F_i\}_i$.
So let us define the following quantities:
\begin{itemize}
\item ${\cal R}(P^*)$: the revenue achieved by price vector $P^*$ in the original instance $\{F_i\}_i$;
\item ${\cal R}(P')$: the revenue achieved by price vector $P'$ in the original instance $\{F_i\}_i$.
\end{itemize}
Using Lemma~\ref{lem:verticaldiscretization} we easily see the following: 
\begin{itemize}
\item ${\cal R}(P^*) \ge OPT-\frac{8k_{1}n}{m}\cdot \max_i\{p^{(i)}\}$;
\item ${\cal R}(P') \ge {\cal R}_{P'}-\frac{4k_{1}n}{m}\cdot \max_i\{p^{(i)}\}; \text{and}$
\item $ {\cal R}_{P^*} \ge {\cal R}(P^*) -\frac{4k_{1}n}{m}\cdot \max_i\{p^{(i)}\}.$
\end{itemize}
Combining these with~\eqref{eq:ah costa costa}, we get
\begin{align*}
{\cal R}(P') &\ge \mathcal{R}_{P^*}-{2nk_{1}k_{2}\over  m}\cdot \max_i\{p^{(i)}\}-\frac{4k_{1}n}{m}\cdot \max_i\{p^{(i)}\}\\
&\ge {\cal R}(P^*)-{2nk_{1}k_{2}\over  m}\cdot \max_i\{p^{(i)}\}-\frac{8k_{1}n}{m}\cdot \max_i\{p^{(i)}\}\\
&\ge OPT-{(2nk_{1}k_{2} + 16k_{1}n)\over  m}\cdot \max_i\{p^{(i)}\}.
\end{align*}
%
%
%

\noindent {\bf Running Time:} Recall that both the support ${\cal S}=\{v^{(1)},v^{(2)},\ldots,v^{(k_1)}\}$  of the value distributions and the set ${\cal P}:=\left\{p^{(1)},\ldots,p^{(k_2)}\right\}$ of prices  are explicitly part of the input to our algorithm.

Given this, the reduction of Lemma~\ref{lem:verticaldiscretization} (used as the first step of our algorithm) takes time polynomial in the size of the input and $\log m$. After this reduction is carried out, the value distributions $\{\hat{F}_i\}_i$ that are provided as input to the dynamic programming algorithm are known explicitly and the probabilities they assign to every point in $\cal S$ are integer multiples of $1 \over m$.

We proceed to bound the running time of the dynamic programming algorithm. First, it is easy to see that its table has  at most { $n \times (m+1)^{k_1k_2}$} cells, since there are $n$ possible choices for $i$ and {$m+1$} possible values for each $\Pr_{i_1,i_2}$. Our DP computation proceeds iteratively from layer $i=1$ to layer $i=n$ of the table. For every cell of layer $i$, there are at most $k_{2}$ different prices we can assign to the next item $i+1$. {For every such price we need to compute a distribution using Eq.~\eqref{eq: DP step} and then round that distribution.} Hence, the total work we need to do per cell of layer $i$ is polynomial in the input size and $\log m$, since our computation involves probabilities that are integer multiples of~$1\over m$. {Indeed the probability distributions maintained in the DP table use probabilities that are integer multiples of { $1\over m$}, and recall that the distributions $\hat{F}_i$ also use probabilities in multiples of~{$1\over m$}.}  Hence, the total time we need to spend to fill up the whole table is polynomial in the size of the input and  {$m^{k_1k_2}$}. In the last phase of the algorithm, we exhaustively search for the cell of layer $n$ with the highest expected revenue. This costs time polynomial in the size of the input and $m^{k_1k_2}$, 
since there are {$m^{O(k_1k_2)}$} cells at layer $n$, and the expected revenue computation for each cell can be done in time polynomial in the input size and $\log m$.
%
Overall, the running time of the algorithm is polynomial in the size of the input and $m^{k_1k_2}$.
\end{prevproof}

\section{Our Discretization Results} \label{appendix:balanced range}

This appendix provides various reductions among item pricing problems. All reductions discretize some aspect of a given item pricing problem. This could be the set of allowable prices, the support of the value distributions, or the probabilities assigned by these distributions to the points in their support. We will bound the loss in approximation resulting from each reduction. This bound is useful in telling us how much revenue we are losing if we solve the discretized problem instead of the given problem.

\subsection{Discretization of Probabilities} \label{sec:vertical discretizing values}

The following lemma allows us to discretize the probabilities assigned by value distributions supported on a discrete set to points in their support. 

%

\begin{lemma}[Probability Discretization]\label{lem:verticaldiscretization}
    Suppose we are given a collection of mutually independent random variables $\{v_i\}_{i\in[n]}$ supported on a discrete set $S=\{s_1,\ldots,s_k \} \subset \mathbb{R}_{\ge 0}$, an interval $[p_{min},p_{max}] \subset \mathbb{R}_{\ge 0}$ of possible prices, and an integer $m \ge 2 k$.\footnote{It is assumed that $S, p_{min}, p_{max}$ and $m$ are given explicitly. We may have any access to the value distributions (as discussed in Appendix~\ref{sec:model}).} In polynomial-time we can construct another collection of mutually independent random variables $\{v'_i\}_{i\in[n]}$ whose distributions are supported on the same set $S$ but only use probabilities that are integer multiples of $1/m$. The distributions of the $v_i'$'s are computed explicitly. Moreover, for any price vector $P\in[p_{min},p_{max}]^n$, the difference in expected revenue from the two collections of random variables is upper bounded $\frac{4kn}{m}\cdot p_{max}$.\end{lemma}

\begin{prevproof}{Lemma}{lem:verticaldiscretization}
If we know the distributions of the $v_i$'s explicitly, then, for all $i$, we construct the distribution of $v'_i$ as follows. Let $\pi_{s_j}= \Pr[v_i = s_j]$ and $\pi'_{s_{j}}=\Pr[v'_{i}=s_{j}]$. For all $j\geq 2$, round all $\pi_{s_{j}}$ down to the nearest integer multiple of $1/m$ to get $\pi'_{s_{j}}$. We then round $\pi_{s_{1}}$ up to get $\pi'_{s_{1}}$ to guarantee that $\pi'$ is still a distribution. We use $\delta_{s_j}$ to denote the rounding error at $s_j$.
    
    As the total variation distance between the distribution of $v_i$ and $v_i'$ is $\frac{1}{2}\sum_{j=1}^m \delta_{s_j}\leq \frac{k}{m}$, we can couple $v_{i}$ and $v'_{i}$ so that $\Pr[v_i\neq v'_i]\leq  \frac{k}{m}.$ Now taking a union bound over all $i$, the probability that the vector $v=(v_1,v_2,\ldots,v_n)$ is different from $v'=(v'_1,v'_2,\ldots,v'_n)$ is at most {$\frac{kn}{m}$}. In other words, with probability at least $1-\frac{kn}{m}$, $v=v'$. Clearly, for all draws from the distribution such that $v=v'$, the revenues are the same. When $v\neq v'$, the difference between the revenues is at most $p_{max}$, since $P\in[p_{min},p_{max}]^n$. And this only happens with probability at most {$\frac{kn}{m}$}. Therefore, the difference between the expected revenues under the two distributions should be no greater than $\frac{kn}{m}\cdot p_{max}$.  
    

    Clearly, we can compute the distributions of the $v'_i$'s in time polynomial in $n$, $k$, $\log m$ and the description complexity of the distributions of the variables $v_i$'s, if these distributions are given to us explicitly. If we have oracle access to the distributions of the $v_i$'s we can query our oracle with high enough precision, say {$1/m$}, to obtain a function $g_i:S \rightarrow [0,1]$ that satisfies {$\sum_{x \in S}g_i(x)=1\pm {k \over m}$}. Using the normalized $g_i$ as a proxy for the distribution of $v_i$ we can follow the algorithm outlined above to define the distribution of $v'_i$.
%
    It is not hard to argue that the total variation distance between $v_i$ and $v'_i$ can be bounded by {${4 k \over m}.$} Hence, we can couple $v_i$ and $v_i'$ so that {$\Pr[v_i\neq v'_i] \le {4 k \over m}$} and proceed as above.
\end{prevproof}



\subsection{Discretization of Prices} \label{sec:discretizing prices}

In this appendix, we present several easy lemmas that can be used to restrict the search space for a (nearly-)optimal price vector. First, it is staightforward to see that, when the value distributions are supported in some range $[u_{min}, u_{max}]$, it is sufficient to only consider prices that lie in the same range, without any sacrifice in revenue.

\begin{lemma}[Price Restriction]\label{lem:betweenminmax}
In an instance of the item pricing problem, suppose that the values are independently distributed in some range $[u_{min},u_{max}]$. Let also $P=(p_1,\ldots,p_n)$ be an arbitrary price vector, and suppose that we modify $P$ into a new price vector $P'$ as follows: for all $i$, set $p'_i = u_{max}$, if $p_i>u_{max}$; set $p'_i=u_{min}$, if $p_i<u_{min}$; otherwise set $p'_i=p_i$. The expected revenues $\mathcal{R}_{P}$ and $\mathcal{R}_{P'}$ achieved by the price vectors $P$ and $P'$ respectively satisfy $\mathcal{R}_{P'}\geq\mathcal{R}_{P}$.
\end{lemma}

\begin{prevproof}{Lemma}{lem:betweenminmax}
Let us do the modification in two steps. We first increase the prices that are below $u_{min}$ to $u_{min}$, and then decrease the prices that are above $u_{max}$ to $u_{max}$. We will show that each step will not decrease the expected revenue.

Let us increase the low prices first, and call $P''$ the resulting price vector. For every sample $(v_1,\ldots,v_n)$ from the value distributions, {if the buyer makes the same decision under $P$ and $P''$, his price under $P''$ is at least as high as under $P$. If the buyer makes a different decision under $P$ and $P''$, it must be that, under $P$, the buyer is buying an item priced less than $u_{min}$ and, because the price of that item was increased to $u_{min}$ in $P''$, the buyer prefers to buy a different item. In this case, the buyer was paying less than $u_{min}$ under $P$ and is paying at least $u_{min}$ under $P''$.} \notshow{let us compare the price paid by the buyer under the original vector $P$ and the vector $P''$, where the prices below $u_{min}$ have been increased to $u_{min}$. First, if the buyer is buying no item under $P''$, then he is also not buying any item under $P$, as boosting prices that are below $u_{min}$ to $u_{min}$ cannot turn any item's value-minus-price gap negative if it was not already negative. Now, if the buyer buys the same item under $P$ and $P''$, the price he pays under $P''$ is at least as high since $P''$ is coordinate-wise at least as large as $P$. If the buyer buys a different item under $P$ and $P''$, it must be that, under $P$, the buyer is buying an item priced less than $u_{min}$ and, because the price of that item was increased to $u_{min}$, the buyer prefers to buy a different item. Still the buyer was paying less than $u_{min}$ under $P$ and is paying at least $u_{min}$ under $P''$.}

Now let us decrease the prices in $P''$ that are above $u_{max}$ to $u_{max}$ to obtain $P'$, and let us compare the price that a buyer will pay under these two price vectors. Whenever the buyer does not buy anything under $P'$, he is also not buying anything under $P''$, as the items under $P'$ are at least as cheap. Also notice that the only items whose value-minus-price gap is different under $P''$ and $P'$ are those that are priced above $u_{max}$ in $P''$ and $u_{max}$ in $P'$, and their gap increases in $P'$. So if the buyer buys different items under $P''$ and $P'$, then the buyer is paying $u_{max}$ under $P'$ and at most $u_{max}$ under $P''$.\end{prevproof}

\notshow{\begin{prevproof}{Lemma}{lem:betweenminmax}
    The proof follows easily by combining Lemmas~\ref{lem:belowmax} and \ref{lem:abovemin} below.
    \begin{lemma}\label{lem:belowmax}
    Let $P$ be a price vector and let $S_{over}= \{i: p_i>u_{max}\}$. We define $P'$ to be the following price vector: $p'_i = u_{max}$, for all $i\in S_{over}$, and $p'_i=p_i$ otherwise. We have $\mathcal{R}_{P'}\geq\mathcal{R}_{P}$.
\end{lemma}

\begin{proof}
    Observe that, for any valuation vector $(v_1,v_2,\ldots,v_n)$ where no item has value larger than or equal to its price in $P$, the revenue is $0$. If that's the case, the revenue under $P'$ can be no worse than $0$. So we only consider valuation vectors where some item $i$ is sold under $P$. Then it must be that  $v_i-p_i\geq 0$ and, for all $j$, $v_i-p_i\geq v_j-p_j$. Under price vector $P'$, there are two possibilities: (1) $i$ is still the item sold; in this case the revenue is the same under $P$ and $P'$; and (2) $i$ is not the winner anymore; the new winner must be among those items whose price was decreased going from $P$ to $P'$, i.e. some item in $j \in S_{over}$. Observe that $p'_j = u_{max}\geq p_i$ (since for $i$ to be a winner under $P$, $p_i$ must be no greater than $u_{max}$, as otherwise $v_i-p_i<0$). Hence, under $P'$ an item $j$ that is at least as expensive as $i$ is sold. Given that the above is true point-wise, $\mathcal{R}_{P'}\geq\mathcal{R}_{P}$. 
\end{proof}

\begin{lemma}\label{lem:abovemin}
    Let $P$ be a price vector and let $S_{below}= \{i: p_i<u_{min}\}$. We define $P'$ to be the following price vector: $p'_i = u_{min}$, for all $i\in S_{below}$, and $p'_i=p_i$ otherwise. We have $\mathcal{R}_{P'}\geq\mathcal{R}_{P}$.
\end{lemma}

\begin{proof}
For any valuation vector $(v_1,v_2,\ldots,v_n)$ where no item has value larger than or equal to its price in $P$, the revenue is $0$. If that's the case, the revenue under $P'$ can be no worse than $0$. So we only consider valuation vectors where some item $i$ is sold under $P$. Suppose first that $i\notin S_{below}$. Then $i$ is still the winner under price vector $P'$, since we only increased the prices of items different than $i$ going from $P$ to $P'$. If the winner $i \in S_{below}$, there are two possibilities: (1) $i$ is still the winner; then the price paid is higher under $P'$, since $p_i'=u_{min}\geq p_i$. (2) $i$ is not the winner anymore; nevertheless, there still needs to be a winner since $v_i\ge u_{min}=p'_i$, and the price paid is at least $u_{min}\geq p_i$. Given that the above hold point-wise,  $\mathcal{R}_{P'}\geq\mathcal{R}_{P}$.\end{proof}\end{prevproof}
 }
Combining Lemma~\ref{lem:betweenminmax} with a price discretization lemma attributed to Nisan~\cite{ChawlaHK07}, we can restrict the set of prices to a set of cardinality $O({\log u_{max}/u_{min}\over \epsilon^{2}})$, if $[u_{min},u_{max}]$ is the support of the value distributions. 

\begin{lemma}[Price Discretization]\label{cor:discreteprice}
    Suppose that the value distributions in an instance of the item pricing problem are independent and supported on $[u_{min},u_{max}] \subset \mathbb{R}_+$. For any $\epsilon\in(0,1/2)$, consider the following finite set of prices: $${\cal P}_{\epsilon}=\left\{ p\ \vline\ p =\frac{1+\epsilon^2-\epsilon}{(1-\epsilon^2)^i}\cdot u_{min},\ i\in \left[ \left\lfloor \log_{\frac{1}{(1-\epsilon^2)}} (u_{max}/u_{min}) \right\rfloor \right] \right\}.$$ For any price vector $P\in[u_{min},u_{max}]^n$, there exists a price vector $P'$ such that $p'_i\in {\cal P}_{\epsilon}$ and $p'_i\in[1-\epsilon, 1+\epsilon^2-\epsilon]
    \cdot p_i$, for all $i$. The expected revenue achieved by the two price vectors satisfies $\mathcal{R}_{P'}\geq(1-2\epsilon)\mathcal{R}_P$.
\end{lemma}

\begin{prevproof}{Lemma}{cor:discreteprice}
Our proof exploits the following lemma, attributed in~\cite{ChawlaHK07} to Nisan.
\begin{lemma}\label{lem:nisan-price}
    Let $\epsilon\in(0,1)$, and let $P$, $P'$ be price vectors satisfying $p'_i\in[1-\epsilon,1+\epsilon^2-\epsilon]\cdot p_i$, for all~$i$. Then the expected revenues achieved by the two price vectors in an instance of the item pricing problem satisfy $\mathcal{R}_{P'}\geq(1-2\epsilon)\mathcal{R}_P$.
\end{lemma}
 
 To prove Lemma~\ref{cor:discreteprice}, define for every $p_i$:
    $$p'_i =\frac{1+\epsilon^2-\epsilon}{(1-\epsilon^2)^{\left\lfloor \log_{1/(1-\epsilon^2)} (p_i/u_{min})\right\rfloor }}\cdot u_{min}.$$
    Observe that $$\frac{1}{(1-\epsilon^2)^{{\lfloor{\log_{{1}/{(1-\epsilon^2)}} (p_i/u_{min})\rfloor}}}}\cdot u_{min}\in [1-\epsilon^2,1]\cdot p_i.$$
    On the other hand, $(1-\epsilon^2)(1+\epsilon^2-\epsilon)=1-\epsilon+\epsilon^3-\epsilon^4\geq 1-\epsilon$. Thus, $p'_i\in[1-\epsilon, 1+\epsilon^2-\epsilon]
    \cdot p_i$, for all $i$. Now Lemma~\ref{lem:nisan-price} implies that $\mathcal{R}_{P'}\geq(1-2\epsilon)R_P$.
\end{prevproof}

We conclude with a lemma with a similar flavor as Lemma~\ref{lem:betweenminmax}.

\begin{lemma}\label{lem:trunclow}
Let $\alpha >0$ be arbitrary, let $P$ be any price vector, and define $P'$ as follows: set $p'_i = p_i$, if $p_i\geq \alpha$, and $p'_i = \alpha$ otherwise. Then the expected revenues $\mathcal{R}_{P}$ and $\mathcal{R}_{P'}$ from these price vectors in an instance of the item pricing problem satisfy $ \mathcal{R}_{P'} \geq \mathcal{R}_{P}-\alpha$.
 \end{lemma}

\begin{prevproof}{Lemma}{lem:trunclow}
Let $S_{low}=\{i\ |\ p_{i}< \alpha\ \}$ and fix the buyer's values for the items. The only case where the buyer's behavior is different under $P$ and $P'$ is when the buyer is buying some item $S_{low}$ under $P$, as these are the only items whose value-minus-price gap changed/decreased from $P$ to $P'$. So the difference in revenue is bounded by the contribution to $R_{P}$ of items in $S_{low}$, which is no greater than $\alpha$.
%
%
%
%
%
\end{prevproof}

\subsection{Discretization of Values} \label{sec:discretizing values}

In this appendix, we establish polynomial-time reductions, discretizing the support of the value distributions in the input to the item pricing problem. Our reductions are specialized depending on whether we want to achieve multiplicative (Lemma~\ref{lem:horizontal-discretization}) or additive (Lemma~\ref{lem:horizontal-discretization additive}) approximations to the optimal revenue. Both reductions are enabled by an extension of Nisan's lemma to value distributions, summarized in Lemma~\ref{lem:nisan-value}. 

\begin{lemma}\label{lem:nisan-value}
    Let $\{v_i\}_{i\in[n]}$ and $\{\hat{v}_i\}_{i\in[n]}$ be two collections of mutually independent random variables, {where all $v_i$'s are supported on a common set $[u_{min},u_{max}] \subset \mathbb{R}_+$, and let $r=u_{max}/u_{min}$}. Let also $\delta\in \left(0,\frac{1}{(4 \lceil\log_2 r\rceil)^{1/(2a-1)}}\right]$, where $a\in (1/2,1)$, and
%
    suppose that we can couple the two collections of random variables so that, for all $i\in[n]$, $\hat{v}_i\in[1+\delta-\delta^2, 1+\delta]\cdot v_i$ with probability $1$. Finally, let $\mathcal{R}_{OPT}$ be the optimal expected revenue from any price vector when the buyer's values are  $\{v_i\}_{i\in[n]}$. Then, for any price vector {$P\in[u_{min},u_{max}]^{n}$}, such that $\mathcal{R}_P(\{v_i\}_i)\geq \mathcal{R}_{OPT}/2$, it holds that
   $${\mathcal{R}}_P(\{\hat{v}_i\}_i) \geq (1-3\delta^{1-a})\mathcal{R}_P(\{v_i\}_i),$$ 
    where $\mathcal{R}_{P}(\{v_i\}_i)$ is the expected revenue under price vector $P$ when the values are $\{v_i\}_{i\in[n]}$, while ${\mathcal{R}}_P(\{\hat{v}_i\}_i)$ is the revenue under $P$ when the  values are $\{\hat{v}_i\}_{i\in[n]}$. 
\end{lemma}

\begin{prevproof}{Lemma}{lem:nisan-value}
    For notational convenience, throughout this proof we use ${{\mathcal{R}}}_P:={\mathcal{R}}_P(\{{v}_i\}_i)$ and $\hat{{\mathcal{R}}}_P:={\mathcal{R}}_P(\{\hat{v}_i\}_i)$. 
    
    Consider now the joint distribution of $\{v_i\}_{i\in[n]}$ and $\{\hat{v}_i\}_{i\in[n]}$ satisfying $\hat{v}_i\in[1+\delta-\delta^2, 1+\delta]\cdot v_i$, for all $i$, with probability $1$. For every point in the support of the joint distribution, we show that the revenue of the seller under price vector $P$ is {not much larger} \notshow{approximately equal} in ``Scenario A'', where the values of the buyer are  $\{v_i\}_{i\in[n]}$, than in ``Scenario B'', where the values are $\{\hat{v}_i\}_{i\in[n]}$. In particular, we argue first that the price paid in ``Scenario A'' is at most\notshow{the two scenarios are within} $\delta\cdot u_{max}$ {larger than the price paid in ``Scenario B,''} with probability $1$.  Indeed, for every point in the support of the joint distribution, we distinguish two cases: 
    \begin{enumerate}
    \item The items sold are the same in the two scenarios. In this case, the winning prices are also the same.
    \item The items sold are different in the two scenarios. In this case, we show that the winning prices are close. 
    Since $\hat{v_i}$ is greater than $v_i$ for all $i$, if there is a winner (item) in Scenario A, there is a winner in Scenario B. Let $i$ be the winner in Scenario A, and $j$ be the winner in Scenario B. We have the following two inequalities:
        \begin{align*}
            v_i-p_i \geq&\ v_j-p_j\\
            \hat{v}_j-p_j \geq&\ \hat{v}_i-p_i
        \end{align*}
    
    The two inequalities imply that 
    $$\hat{v}_j-v_j\geq \hat{v}_i-v_i.$$
    Since $\hat{v}_j\in [1+\delta-\delta^2, 1+\delta]\cdot v_j$, it follows that $\hat{v}_j-v_j\leq \delta\cdot v_j$. Using the same starting condition for $i$, we can show that $\hat{v}_i-v_i\geq (\delta-\delta^2)\cdot v_i$.
    
    Hence,
    $$\delta\cdot v_j\geq (\delta-\delta^2)\cdot v_i.$$
    
    Also we know that $$p_j\geq p_i + v_j-v_i.$$
    
    Therefore, \begin{equation}\label{eq:additive error}p_j\geq p_i + v_j-v_i\geq p_i + (1-\delta)\cdot v_i - v_i = p_i -\delta\cdot v_i.\end{equation}  
\end{enumerate}

The above establishes that with probability $1$ the price paid in ``Scenario A'' is larger than that paid in ``Scenario B'' by \notshow{the two scenarios are within} at most an {\em additive} $\delta u_{max}$. We proceed to convert this additive approximation guarantee into a multiplicative approximation guarantee. Observe that whenever $p_i\geq \delta^{a} v_i$, 
$ p_i -\delta\cdot v_i\geq (1-{\delta}^{1-a})p_i$. Hence, if we can show that most of the revenue ${{\mathcal{R}}}_P$ is contributed by value-price pairs $(v_i,p_i)$ satisfying $p_i\geq \delta^{a} v_i$, we can convert our additive approximation to a $(1-\delta^{1-a})$ multiplicative approximation.
Indeed, we argue next that when a price vector $P$ satisfies $\mathcal{R}_P\geq \mathcal{R}_{OPT}/2$, the contribution to the revenue from the event
    $$S=\{\text{the sold item }k\text{ satisfies } p_k<\delta^{a}v_k\}$$
    is small. More precisely, 
\begin{proposition}
    If $\mathcal{R}_P\geq \mathcal{R}_{OPT}/2$, then the contribution to $\mathcal{R}_P$  from the event $S$ is no greater than $2\delta^{1-a}\mathcal{R}_P$.
\end{proposition}
\begin{proof}
    The proof is by contradiction. For all $i\in[\lceil \log_2 r \rceil]$, define the event  
    $$S_i=\{(\text{the sold  item }k \text{ has price } p_k<{\delta}^{a}v_k) \land (p_k\in [2^{i-1}u_{min}, 2^{i}u_{min})) \}.$$ Note that $S_i$ and $S_j$ are disjoint for all $i\neq j$. Let $n_p = \lceil \log_2 r \rceil$ and note that $S=\cup_{i=1}^{n_p} S_i$.\footnote{To be more accurate, replace the set $[2^{i-1}u_{min}, 2^{i}u_{min})$ by $[2^{i-1}u_{min}, 2^{i}u_{min}]$ for the definition of the event $S_{n_p}$.} Assuming that the contribution to $\mathcal{R}_P$ from the event $S$ is larger than $2\delta^{1-a}\mathcal{R}_P$, there must exist some $i$ such that the contribution to $\mathcal{R}_P$ from $S_i$ is at least $2\delta^{1-a}\mathcal{R}_P/n_p\geq \delta^{1-a}\mathcal{R}_{OPT}/n_p$. For this $i$, let us modify the price vector $P$ to $P'$ in the following fashion: $$p'_k=\begin{cases} +\infty & p_k\notin[2^{i-1}u_{min}, 2^{i}u_{min})\\ \frac{2^{i-1}u_{min}}{{\delta}^{a}} & otherwise\end{cases}$$
    
    We claim that for all outcomes $(v_1,v_2,\ldots,v_n)\in S_i$, there always exists an item sold under $P'$. Indeed, let $k$ be the winner under $P$. Then $p_k<{\delta}^{a} v_k$. By the definition of $p'_k$, we know that $$p'_k=\frac{2^{i-1}u_{min}}{{\delta}^{a}}\le p_k/{\delta}^{a}<v_k.$$ Thus, an item has to be sold. Moreover, the sold item has price $\frac{2^{i-1}u_{min}}{{\delta}^{a}}$, as all the other prices are set  to $+\infty$. Hence, we can lower bound $\mathcal{R}_{P'}$ as follows 
    $$\mathcal{R}_{P'} \ge \Pr[S_i]\cdot \frac{2^{i-1}u_{min}}{{\delta}^{a}} \geq \frac{\text{Contribution of } S_i \ \text{to} \  \mathcal{R}_{P}}{{2\delta}^{a}}\geq \frac{\delta^{1-a}\mathcal{R}_{OPT}}{2n_p{\delta}^{a}}.$$
    Given that $\delta \leq (\frac{1}{4n_p})^{1/(2a-1)}$, the above implies $\mathcal{R}_{P'}\geq 2\mathcal{R}_{OPT}$, which is impossible, i.e. we get a contradiction. This concludes the proof of the proposition.
\end{proof}
        Given the proposition, at least $(1-2\delta^{1-a})$ fraction of ${\cal R}_P$ is contributed by value-price pairs $(v_i,p_i)$ satisfying $p_i\geq \delta^{a} v_i$.  Recalling our earlier discussion, this implies that $\hat{\mathcal{R}}_{P}\geq(1-2\delta^{1-a})(1-{\delta}^{1-a})\mathcal{R}_P\geq(1-3\delta^{1-a})\mathcal{R}_{P}$.
%
\end{prevproof}

%
Lemma~\ref{lem:nisan-value} enables polynomial-time reductions from value distributions supported on some bounded range $[u_{min},u_{max}]$ to value distributions supported on some discrete set of cardinality $O(\log r)$, where $r=u_{max}/u_{min}$. We provide two reductions (Lemmas~\ref{lem:horizontal-discretization} and~\ref{lem:horizontal-discretization additive}) depending on whether the  approximation to the optimal revenue is intended to be additive or multiplicative. We note that a straightforward extension of Nisan's lemma to value distributions would have resulted in supports of cardinality $O(r^2 \log r)$. The exponential improvement in the size of the support comes from our more intricate extension obtained by Lemma~\ref{lem:nisan-value}.

\begin{lemma}[Value Discretization for Multiplicative Approximations]\label{lem:horizontal-discretization}
    Let $\{v_i\}_{i\in[n]}$ be a collection of mutually independent random variables supported on a bounded range $[u_{min},u_{max}] \subset \mathbb{R}_+$, and $r={u_{max}\over u_{min} }$. 
    For any $\delta \in \left(0,\frac{1}{(4 \lceil \log_2 r\rceil) ^{4/3}}\right)$, there exists a collection of mutually independent random variables $\{\hat{v}_i\}_{i\in[n]}$, which are supported on a discrete set of cardinality $O\left(\frac{\log r}{\delta^2}\right)$ and satisfy the following properties.
    \begin{enumerate}
    \item The optimal revenue when the buyer's values are $\{\hat{v}_i\}_{i\in[n]}$ is at least a $(1-3\delta^{1/8})$-fraction of the optimal revenue when the values are $\{v_i\}_{i\in[n]}$. 
   I.e. $\hat{\mathcal{R}}_{OPT}\geq (1-3\delta^{1/8})\mathcal{R}_{OPT}$, where $\mathcal{R}_{OPT}= \max_{P}\mathcal{R}_{P}(\{v_i\}_i)$ and $\hat{\mathcal{R}}_{OPT}= \max_{P}{\mathcal{R}}_{P}(\{\hat{v}_i\}_i)$.  
    
    \item Moreover, for any constant {$\rho\in (0,1/2)$} and any price vector $P$ such that ${\mathcal{R}}_{P}(\{\hat{v}_i\}_i)\geq (1-\rho)\hat{\mathcal{R}}_{OPT}$, we can construct in time polynomial in the description of $P$ and $1/\delta$ another price vector $\tilde{P}$ such that $\mathcal{R}_{\tilde{P}}(\{v_i\}_i)\geq (1-7\delta^{1/8}-\rho)\mathcal{R}_{OPT}$. 
\end{enumerate}
If $u_{min}$ and $u_{max}$ are provided explicitly as input to the reduction,\footnote{This requirement is only relevant if we have oracle access to the distributions of the $v_i$'s, as if we are given the distributions explicitly we immediately also know $u_{min}$ and $u_{max}$.} we can compute the distributions of the $\hat{v}_i$'s\footnote{The $\hat{v}_i$'s will inherit the same type of access that we have to the distributions of the $v_i$'s, according to Appendix~\ref{sec:model}. In particular, if the $v_i$'s are specified explicitly in the input to the reduction then the $\hat{v}_i$'s will also be specified explicitly in the output of the reduction; if the $v_i$'s are given as oracles then the $\hat{v}_i$'s will be given as oracles; etc. } and their support in time polynomial in the description of $\{v_i\}_{i\in[n]}$, $\angler{u_{min}}$, $\angler{u_{max}}$ and $1/\delta$. \notshow{\yangnote{Remove: We can also compute the distributions of the variables $\{\hat{v}_i\}_{i\in[n]}$ in time polynomial in the size of the input and $1/\delta$, if we have the distributions of the variables $\{v_i\}_{i\in[n]}$ explicitly. If we have oracle access to the distributions of the variables $\{v_i\}_{i\in[n]}$, we can construct an oracle for the distributions of the variables $\{\hat{v}_i\}_{i\in[n]}$, running in time polynomial in  $\log u_{min}$, $\log u_{\max}$, $1/\delta$, the input to the oracle and the desired precision.}}
\end{lemma}



\begin{prevproof}{Lemma}{lem:horizontal-discretization}
Informally, our reduction establishes the following properties: (1) Suppose that we transform a buyer with arbitrary valuations (call this buyer ``Buyer A'') to a buyer with discrete valuations (called ``Buyer B'') by first multiplying each of Buyer A's values by $(1+\delta)$ and then rounding them down to the closest real of the form $(1+\delta)(1+\xi)^{j}u_{min}$, for some integer $j$, where $\delta$ is fixed and $\xi=\frac{\delta^{2}}{1+\delta-\delta^{2}}$. {We show that the optimal revenue from Buyer B is very close to the optimal revenue from Buyer A by exploiting that Buyer B's values have been boosted, using Lemma~\ref{lem:nisan-value}.} (2) For the reduction to be computationally useful, we also show that given an approximately optimal price vector for Buyer B, if we divide all prices by $(1+\delta)(1+\delta-\delta^{2})$, the new price vector will be an approximately optimal price vector for Buyer~A. {Intuitively, scaling down the prices undoes the effect of boosting the values.}

\medskip   We proceed to make the above plan precise, beginning with the description of the random variables $\{\hat{v}_i\}_{i\in[n]}$. We will use $\{F_i\}_{i\in[n]}$ and $\{\hat{F}_i\}_{i\in[n]}$ to denote respectively the cumulative distribution functions of the variables $\{v_i\}_{i\in[n]}$ and $\{\hat{v}_i\}_{i\in[n]}$. 
    Our variables $\{\hat{v}_i\}_{i\in[n]}$ will only be supported on the set 
    $$\left\{a_{j}=(1+\delta)(1+\xi)^{j}u_{min}\ \vline \ j\in\Big{\{}0,\ldots,\big{\lfloor}\log_{1+\xi}\frac{u_{max}}{u_{min}}\big{\rfloor}\Big{\}}\right\}.$$
Moreover, for all $i$, $\hat{F}_i$ is defined in terms of $F_i$ as follows: 
$$\hat{F}_i(a_{j}) = F_i(a_{j}/(1+\delta-\delta^{2}))-F_i(a_{j}/(1+\delta))+\hat{F}_i(a_{j-1}) \ind_{j>0},~~~\forall j.$$
    
\noindent Now, for all $i$, we couple $v_i$ with $\hat{v}_i$ as follows: If $v_i\in[a_{j}/(1+\delta),a_{j}/(1+\delta-\delta^{2}))$, we set $\hat{v}_i = a_{j}$. Given our definition of the $\hat{F}_i$'s, this defines a valid coupling of the collections ${\cal V}=\{v_i\}_i$ and $\hat{{\cal V}}=\{\hat{v}_i\}_i$. Moreover, by definition, our coupling satisfies 
\begin{align}
\hat{v}_i\in[1+\delta-\delta^2, 1+\delta]\cdot v_i, \forall i, \label{eq: coupling costas 1}
\end{align} with probability $1$, {and all the $\hat{v}_i$'s are supported on $[(1+\delta)u_{min}, (1+\delta)u_{max}]$}.


We are now ready to establish the first part of the lemma. Using Lemma~\ref{lem:nisan-value} and the property of our coupling it follows immediately that $${\mathcal{R}}_P(\hat{{\cal V}})\geq (1-3\delta^{1/8})\mathcal{R}_P({\cal V}),$$ for any price vector {$P\in[u_{min},u_{max}]^{n}$} s.t. $\mathcal{R}_P({\cal V}) \ge {1 \over 2} \mathcal{R}_{OPT}$. {Lemma~\ref{lem:betweenminmax} implies that the optimal revenue for $\cal V$ is achieved by some price vector in $[u_{min},u_{max}]^{n}$.}  Hence, we get from the above that $\hat{\mathcal{R}}_{OPT}\geq (1-3\delta^{1/8})\mathcal{R}_{OPT}$. 

We proceed to show the second part of the lemma. We do this by defining another collection of random variables $\tilde{{\cal V}}=\{\tilde{v}_i\}_{i\in[n]}$. These are defined  implicitly via the following coupling between $\{\tilde{v}_i\}_{i\in[n]}$ and  $\{\hat{v}_i\}_{i\in[n]}$: for all $i$, we set $$\tilde{v}_i = \frac{\hat{v}_i}{(1+\delta-\delta^2)(1+\delta)}.$$
{It follows that the $\tilde{v}_i$'s are supported on $[u_{min}/(1+\delta-\delta^{2}), u_{max}/(1+\delta-\delta^{2})]$.}

Moreover, for any price vector $P$, let us construct another price vector $\tilde{P}$ as follows:
\begin{align}
\tilde{p}_i = \frac{{p}_i}{(1+\delta-\delta^2)(1+\delta)}.\label{eq: price modificationnnn}
\end{align}

Under our coupling between $\{\tilde{v}_i\}_{i\in[n]}$ and  $\{\hat{v}_i\}_{i\in[n]}$, it is not hard to see that if we use price vector $P$ when the buyer's values are $\{\hat{v}_i\}_{i\in[n]}$ and price vector $\tilde{P}$ when the buyer's values are $\{\tilde{v}_i\}_{i\in[n]}$, then the index of the item that the buyer buys is the same in the two cases, with probability $1$. Hence:

\begin{align}
{\mathcal{R}}_{\tilde{P}}(\tilde{\cal V}) = \frac{{\mathcal{R}}_P(\hat{\cal V})}{(1+\delta-\delta^2)(1+\delta)}.\label{eq:revenue tilde and hat}
\end{align}
This follows from the fact that both $\tilde{P}$ and  $\{\tilde{v}_i\}_{i\in[n]}$  are the same linear transformations of $P$ and $\{\hat{v}_i\}_{i \in [n]}$ respectively.

Composing the coupling between $v_i$ and $\hat{v}_i$ and the coupling between $\hat{v}_i$ with $\tilde{v}_i$, we obtain a coupling between $v_i$ and $\tilde{v}_i$. We show that this coupling satisfies $v_i\in[1+\delta-\delta^2,1+\delta]\cdot \tilde{v_i}$, with probability $1$. Since $(1+\delta-\delta^{2})v_{i}\leq\hat{v}_{i}\leq(1+\delta)v_{i}$, it follows that $$v_{i}/(1+\delta)\leq \hat{v}_{i}/(1+\delta-\delta^{2})(1+\delta)=\tilde{v}_{i}\leq v_{i}/(1+\delta-\delta^{2}).$$
Hence 
\begin{align}
v_i\in[1+\delta-\delta^2,1+\delta]\cdot \tilde{v_i}, \forall i, \label{eq: coupling costas 2}
\end{align}
with probability $1$. Now an application of Lemma~\ref{lem:nisan-value} implies that, for any price vector {$\tilde{P}\in[u_{min}/(1+\delta-\delta^{2}), u_{max}/(1+\delta-\delta^{2})]^{n}$} satisfying ${\cal R}_{\tilde{P}}(\tilde{V}) \ge {1 \over 2}{\cal R}_{OPT}(\tilde{V})$:
\begin{align}\mathcal{R}_{\tilde{P}} ({\cal V})\geq (1-3\delta^{1/8}){\mathcal{R}}_{\tilde{P}}(\tilde{\cal V}).\label{eq:revnue tilde and original}
\end{align}

Now let $P$ be a price vector satisfying  ${\mathcal{R}}_P(\hat{\cal V})\geq (1-\rho)\hat{\mathcal{R}}_{OPT}$. {Lemma~\ref{lem:betweenminmax} implies that WLOG we can assume that $P\in[(1+\delta)u_{min}, (1+\delta)u_{max}]^{n}$ (as if the given price vector is not in this set, we can efficiently convert it into one that is in this set without losing any revenue).} Then the vector $\tilde{P}$ obtained from $P$ via Eq.~\eqref{eq: price modificationnnn} {is in $[u_{min}/(1+\delta-\delta^{2}), u_{max}/(1+\delta-\delta^{2})]^{n}$}, and clearly satisfies 
${\mathcal{R}}_{\tilde{P}}(\tilde{\cal V})\geq (1-\rho){\mathcal{R}}_{OPT}(\tilde{\cal V})$, 
as $\tilde{P}$ and  $\tilde{\cal V}$  are the same linear transformations of $P$ and $\hat{\cal V}$ respectively. Hence, Equations~\eqref{eq:revenue tilde and hat} and \eqref{eq:revnue tilde and original} give 
\begin{align*}
{\mathcal{R}}_{\tilde{P}}({\cal V}) &\geq \Big{(}(1-3\delta^{1/8})\Big{/}(1+\delta)(1+\delta-\delta^{2})\Big{)}{\mathcal{R}}_P(\hat{\cal V}) \notag\\
&\geq(1-3\delta^{1/8})(1-2\delta){\mathcal{R}}_{P} (\hat{\cal V}) \notag\\
&\geq(1-4\delta^{1/8}){\mathcal{R}}_{P}(\hat{\cal V})\\ 
&\geq(1-4\delta^{1/8})(1-\rho)\hat{\mathcal{R}}_{OPT} \notag\\
&\geq(1-4\delta^{1/8})(1-\rho)(1-3\delta^{1/8})\mathcal{R}_{OPT}~~~~~~\text{(using the first part of the theorem)} \\
&\geq (1-7\delta^{1/8}-\rho)\mathcal{R}_{OPT}. \notag
\end{align*}
\end{prevproof}

\begin{lemma}[Value Discretization for Additive Approximations]\label{lem:horizontal-discretization additive}
    Let $\{v_i\}_{i\in[n]}$ be a collection of mutually independent random variables supported on a bounded range $[u_{min},u_{max}] \subset \mathbb{R}_+$, and $r={u_{max}\over u_{min} }$. 
    For any $\delta >0$, there exists another collection of mutually independent random variables $\{\hat{v}_i\}_{i\in[n]}$, which are supported on a discrete set of cardinality $O\left(\frac{\log r}{\delta^2}\right)$ and satisfy the following properties.
    \begin{enumerate}
    \item The optimal revenue when the buyer's values are $\{\hat{v}_i\}_{i\in[n]}$ is at most $\delta u_{\max}$ smaller than the optimal revenue when the values are $\{v_i\}_{i\in[n]}$. 
   I.e. $\hat{\mathcal{R}}_{OPT}\geq \mathcal{R}_{OPT}-\delta u_{\max}$, where $\mathcal{R}_{OPT}= \max_{P}\mathcal{R}_{P}(\{v_i\}_i)$ and $\hat{\mathcal{R}}_{OPT}= \max_{P}{\mathcal{R}}_{P}(\{\hat{v}_i\}_i)$.  
    
    \item Moreover, for any constant {$\rho>0$} and any price vector $P$ such that ${\mathcal{R}}_{P}(\{\hat{v}_i\}_i)\geq \hat{\mathcal{R}}_{OPT}-\rho$, we can construct in time polynomial in the description of $P$ and $1/\delta$ another price vector $\tilde{P}$ such that $\mathcal{R}_{\tilde{P}}(\{v_i\}_i)\geq \mathcal{R}_{OPT}-4\delta u_{\max}-\rho$. 
\end{enumerate}
If $u_{min}$ and $u_{max}$ are provided explicitly as input to the reduction,\footnote{This requirement is only relevant if we have oracle access to the distributions of the $v_i$'s, as if we are given the distributions explicitly we immediately also know $u_{min}$ and $u_{max}$.} we can compute the distributions of the $\hat{v}_i$'s\footnote{The $\hat{v}_i$'s will inherit the same type of access that we have to the distributions of the $v_i$'s, according to Appendix~\ref{sec:model}. In particular, if the $v_i$'s are specified explicitly in the input to the reduction then the $\hat{v}_i$'s will also be specified explicitly in the output of the reduction; if the $v_i$'s are given as oracles then the $\hat{v}_i$'s will be given as oracles; etc. } and their support in time polynomial in the description of $\{v_i\}_{i\in[n]}$, $\angler{u_{min}}$, $\angler{u_{max}}$, and $1/\delta$. \notshow{\yangnote{Remove: We can also compute the distributions of the variables $\{\hat{v}_i\}_{i\in[n]}$ in time polynomial in the size of the input and $1/\delta$, if we have the distributions of the variables $\{v_i\}_{i\in[n]}$ explicitly. If we have oracle access to the distributions of the variables $\{v_i\}_{i\in[n]}$, we can construct an oracle for the distributions of the variables $\{\hat{v}_i\}_{i\in[n]}$, running in time polynomial in  $\log u_{min}$, $\log u_{\max}$, $1/\delta$, the input to the oracle and the desired precision.}}
\end{lemma}

\begin{proof} 
The proof is very similar to the proof of Lemma~\ref{lem:horizontal-discretization}. In particular, let $\mathcal{\hat{V}}=\{\hat{v}_{i}\}_{i\in[n]}$ and $\mathcal{\tilde{V}}=\{\tilde{v}_{i}\}_{i\in[n]}$ be defined in the same way as in that lemma. So, with probability 1, \eqref{eq: coupling costas 1} and~\eqref{eq: coupling costas 2} are satisfied. So Eq.~\eqref{eq:additive error} of Lemma~\ref{lem:nisan-value} implies that $\hat{\mathcal{R}}_{OPT}\geq \mathcal{R}_{OPT}-\delta u_{\max}$.

Now let $P$ be a price vector satisfying  ${\mathcal{R}}_P(\hat{\cal V})\geq \hat{\mathcal{R}}_{OPT}-\rho$. Lemma~\ref{lem:betweenminmax} implies that WLOG we can assume that $P\in[(1+\delta)u_{min}, (1+\delta)u_{max}]^{n}$. Then the vector $\tilde{P}$ obtained from $P$ via Eq.~\eqref{eq: price modificationnnn} {is in $[u_{min}/(1+\delta-\delta^{2}), u_{max}/(1+\delta-\delta^{2})]^{n}$}, and satisfies 
${\mathcal{R}}_{\tilde{P}}(\tilde{\cal V})=\frac{{\mathcal{R}}_{{P}}(\hat{\cal V})}{(1+\delta)(1+\delta-\delta^{2})}$, 
as $\tilde{P}$ and  $\tilde{\cal V}$  are the same linear transformations of $P$ and $\hat{\cal V}$ respectively. Hence, 
\begin{align*}
{\mathcal{R}}_{\tilde{P}}({\cal V}) &\geq {\mathcal{R}}_{\tilde{P}}(\mathcal{\tilde{V}})-\delta u_{max} ~~~~~~~~~~\text{(by Eq.~\eqref{eq:additive error} of Lemma~\ref{lem:nisan-value} given~\eqref{eq: coupling costas 2})}\notag\\
&= \frac{{\mathcal{R}}_{{P}}(\hat{\cal V})}{(1+\delta)(1+\delta-\delta^{2})}-\delta u_{max} \notag\\
&\geq(1-2\delta){\mathcal{R}}_{P} (\hat{\cal V}) -\delta u_{max}\notag\\
&\geq\hat{\mathcal{R}}_{OPT}-3\delta u_{max}-\rho\\ 
&\geq \mathcal{R}_{OPT}-4\delta u_{max}-\rho. \notag
\end{align*}

\end{proof}

\subsection{Omitted Details from Section~\ref{sec:additive}}\label{appendix:additive}

\medskip \begin{prevproof}{Lemma}{lem:additive to balanced}
Let $OPT$ be the optimal revenue under $\mathcal{V}$. Also, let $\epsilon'=\epsilon/3$.

By Lemma~\ref{lem:betweenminmax}, we only need to consider price vectors in $[0,1]^{n}$ for an optimal one. Moreover, it follows from Lemma~\ref{lem:trunclow} that, if we restrict all prices to be higher than $\epsilon'$, we lose at most an additive $\epsilon'$ in revenue. So there exists a price vector $\bar{P}\in[\epsilon',1]^{n}$ such that $\mathcal{R}_{\bar{P}}(\mathcal{V})\geq OPT -\epsilon'$.

Now, let us define a new collection of random variables $\mathcal{\tilde{V}}=\{\tilde{v}_{i}\}_{i\in[n]}$ via the following coupling: for all $i\in [n]$, set  $\tilde{v}_i={\epsilon'\over 2}$ if $v_i < {\epsilon'}$, and $\tilde{v}_i=v_i$ otherwise. We claim the following:

\begin{claim}\label{claim:little claim}
For any price vector $P$ in $[\epsilon',1]^{n}$, $\mathcal{R}_{P}(\tilde{\mathcal{V}})=\mathcal{R}_{P}(\mathcal{V})$.
\end{claim}
\begin{prevproof}{Claim}{claim:little claim}
 Recall that the variables $\{\tilde{v}_i\}_i$ are defined via a coupling with the $v_i$'s. Under the same coupling, sample values from the $v_i$'s and the $\tilde{v}_i$'s. For all items $i$ such that ${v}_{i}\neq \tilde{v}_{i}$, the price of item $i$ is higher than the value of item $i$ in both cases, so item $i$ will not be purchased in both cases. That means the buyer will make the same decision in both cases, as she will only consider items whose values are the same. So the revenues under ${\cal V}$ and $\tilde{\cal V}$ are pointwise equal.
\end{prevproof}

We proceed to show that an approximately optimal solution for $\tilde{\cal V}$ provides an approximately optimal solution for ${\cal V}$. Suppose that a price vector $\tilde{P}$ satisfies $\mathcal{R}_{\tilde{P}}(\tilde{\mathcal{V}})\geq \mathcal{R}_{OPT}(\tilde{V})-\epsilon'$. By Lemmas~\ref{lem:betweenminmax} and~\ref{lem:trunclow}, we can efficiently convert $\tilde{P}$ to $P'\in[\epsilon',1]^{n}$, such that $\mathcal{R}_{P'}(\tilde{\mathcal{V}})\geq \mathcal{R}_{\tilde{P}}(\tilde{\mathcal{V}})-\epsilon'$.

Combining the inequalities above, we have $$\mathcal{R}_{OPT}(\tilde{\mathcal{V}})\geq \mathcal{R}_{\bar{P}}(\tilde{\mathcal{V}})= \mathcal{R}_{\bar{P}}(\mathcal{V})\geq OPT -\epsilon',$$
and
$$\mathcal{R}_{P'}({\mathcal{V}})=\mathcal{R}_{P'}(\tilde{\mathcal{V}})\geq \mathcal{R}_{\tilde{P}}(\tilde{\mathcal{V}})-\epsilon' \geq \mathcal{R}_{OPT}(\tilde{\mathcal{V}})-2\epsilon'.$$

Thus, $$\mathcal{R}_{P'}({\mathcal{V}})\geq OPT-3\epsilon'.$$
\end{prevproof}

\begin{prevproof}{Theorem}{thm:additive discretization}
Lemma~\ref{lem:horizontal-discretization additive} implies that we can reduce the problem {\sc AdditivePrice}$({\cal V},\epsilon)$ to the problem {\sc AdditivePrice}$(\hat{{\cal V}},{\frac{\epsilon}{3})}$, where $\hat{\cal V}=\{\hat{v}_i\}_i$ is a collection of mutually independent random variables supported on a common discrete set ${\cal S}=\{s^{(1)},\ldots,s^{(k_1)}\} \subset [(1+\frac{\epsilon}{6u_{max}})u_{min},(1+\frac{\epsilon}{6u_{max}})u_{max}]$ of cardinality $k_1=O({u_{max}^{2}\log r \over \epsilon^{2}})$. Now, Lemmas~\ref{lem:betweenminmax} and~\ref{cor:discreteprice} imply that we can reduce the problem {{\sc AdditivePrice}$(\hat{{\cal V}},{\epsilon \over 3})$} to the problem of approximating {\sc ResrtictedPrice}$(\hat{{\cal V}}, {\cal P})$ to within an additive $\epsilon\over 6$, where ${\cal P}$ is a discrete set of prices of cardinality $O(\frac{{u_{max}^{2}}\log r}{\epsilon^{2}})$, satisfying $\max_{x \in {\cal P}}{x} \le {7 \over 6} u_{max}$.
\end{prevproof}

\subsection{Proof of Theorem~\ref{thm:discretization}} \label{app:discretizationreduction for multiplicative}
\begin{prevproof}{Theorem}{thm:discretization}
Lemma~\ref{lem:horizontal-discretization} implies that we can reduce the problem {\sc Price}$({\cal V},\epsilon)$ to the problem {{\sc Price}$(\hat{{\cal V}},{{\epsilon \over 8}})$}, where $\hat{\cal V}=\{\hat{v}_i\}_i$ is a collection of mutually independent random variables supported on a common discrete set ${\cal S}=\{s^{(1)},\ldots,s^{(k_1)}\} \subset [(1+({\epsilon / 8})^8)u_{min},(1+({\epsilon / 8})^8)u_{max}]$ of cardinality $k_1=O({\log r \over \epsilon^{16}})$. Now, Lemmas~\ref{lem:betweenminmax} and~\ref{cor:discreteprice} imply that we can reduce the problem {{\sc Price}$(\hat{{\cal V}}, {{\epsilon \over 8}})$} to the problem of approximating {\sc ResrtictedPrice}$(\hat{{\cal V}}, {\cal P})$ to within a multiplicative factor of {$1-{\epsilon \over 16}$}, where ${\cal P}$ is a discrete set of prices of cardinality {$O(\frac{\log r}{\epsilon^{2}})$}. 
\end{prevproof}

\section{Proof of Theorem~\ref{thm:general algorithm}} \label{app:proof of general multiplicative PTAS}

We restate and prove Theorem~\ref{thm:general algorithm}. We note that we have not tried to carefully optimize the constants in the running time. There may be room for improvement with a more careful analysis.

\begin{varthm}{{\bf \ref{thm:general algorithm} [Restated]}}
    Let $\{F_i\}_{i \in [n]}$ be a collection of distributions that are supported on a bounded set $[u_{min},u_{max}] \subset \mathbb{R}_+$, where $u_{min}$ and $u_{max}$ are specified as part of the input,\footnote{This requirement is only relevant if we have oracle access to the $F_i$'s, as if we are given the distributions explicitly we immediately also know $u_{min}$ and $u_{max}$.}  and let $r:=u_{max}/u_{min}$. Then, for any constant $\epsilon >0$, there is an algorithm that runs in time polynomial in the size of the input and  {$\max\left\{n^{\log^{11} r \cdot  \log \log r }, n^{\log^{3} r \cdot\log {1 \over \epsilon}\over \epsilon^{8}}\right\}$} 
      and computes a price vector $P$ such that $$\mathcal{R}_{P}\geq(1-\epsilon) {OPT},$$ where $\mathcal{R}_{P}$ is the expected revenue under price vector $P$ when the buyer's values for the items are independent draws from the distributions $\{F_i\}_i$ and ${OPT}$ is the optimal revenue.
\end{varthm}

\begin{prevproof}{Theorem}{thm:general algorithm}
    First set $\hat{\epsilon} = \min\left\{\epsilon, {1 \over (4 \lceil \log_2 r \rceil)^{1/6}} \right\}$. Clearly, it suffices to find a price vector with expected revenue $(1-\hat{\epsilon})OPT$.  Now, let us invoke the reduction of Theorem~\ref{thm:discretization}, reducing this task  to approximating {\sc RestrictedPrice}$(\{ \hat{F_i}\}_i,{\cal P})$ to within a factor of $(1-{\hat{\epsilon} \over 16})$, where the distributions $\{\hat{F_i}\}_i$ are supported on a discrete set ${\cal S}=\{s^{(1)},\ldots,s^{(k_1)}\}$ of cardinality $k_1=O(\log r/ \hat{\epsilon}^{16})$ and the prices are also restricted to a discrete set ${\cal P}=\{p^{(1)},\ldots,p^{(k_2)}\}$ of cardinality $k_2=O(\log r/ \hat{\epsilon}^{2})$. It is important to note that ${\cal S}\subset [(1+({\hat{\epsilon} / 8})^8)u_{min},(1+({\hat{\epsilon} / 8})^8)u_{max}]$ and $\min_i\{p^{(i)}\} \le \min_i \{{s^{(i)}}\}$ (this can be checked by a careful study of the proof of Theorem~\ref{thm:discretization}). Hence, if $\widehat{OPT}$ is the optimal revenue of the discrete instance resulting from the reduction, we have  $\widehat{OPT} \ge \min_i\{p^{(i)}\}$. It is our goal to  achieve revenue at least $(1-{\hat{\epsilon} \over 16}) \widehat{OPT}$.

    \notshow{\costasfootnote{There was a case analysis here that didn't make sense to me, at least after the new formulation of Theorem~\ref{lem:dynamic programming revenue}. It is commented out.}}
To do this, we invoke the algorithm of Theorem~\ref{lem:dynamic programming revenue} with a choice of { $m=\Theta({ n r k_1 k_2 \over \hat{\epsilon}})=O({n r \log^2 r \over \hat{\epsilon}^{19}})$}, obtaining a price vector with revenue at least:
\begin{align}
   \widehat{OPT} - O\left( {n k_{1}k_{2} \over m} \max_i \{p^{(i)}\} \right) =    \widehat{OPT} - O\left( \hat{\epsilon} \right) \cdot \min_i \{p^{(i)}\} \ge \widehat{OPT}  \left( 1 - O(\hat{\epsilon})\right),
   \end{align}
    as we wanted. The running time of the algorithm in this case is polynomial in the input and $m^{{ {\log^2 r}\over\hat{\epsilon}^{18}}}$, that is polynomial in the input and {$\max\left\{ n^{\log^6 r \cdot  \log \log r},n^{ {\log^3 r \cdot  \log {1\over\epsilon} \over {\epsilon}^{18}}} \right\}$}. \notshow{\footnote{\yangnote{It was \costasnote{$n^{ {\log^3 r \cdot  \log {1/\epsilon} \over {\epsilon}^{18}}}$}, but I think we should replace $\hat{\epsilon}$ with ${\epsilon \over \log r^{1/6}}$}}. }

Being a bit more careful in the application of our discretization lemmas we can obtain running time polynomial in the input and {$\max\left\{n^{\log^{11} r \cdot  \log \log r}, n^{\log^{3} r \cdot\log {1 \over \epsilon}\over \epsilon^{8}}\right\}$}. Recall that to establish our reduction in Theorem~\ref{thm:discretization} we employed Lemma~\ref{lem:horizontal-discretization}, which in turn made use of Lemma~\ref{lem:nisan-value}, setting $a={7 \over 8}$. Setting instead $a={2 \over 3}$ would result in a stronger Lemma~\ref{lem:horizontal-discretization} and Theorem~\ref{thm:discretization}, improving our running time here. 
    \end{prevproof}

\section{Details of Section~\ref{sec:truncate}: MHR to Bounded Distributions}\label{sec:MHR to BD}

\subsection{Basic Properties of MHR Distributions} \label{appendix:MHR}

\begin{definition} \label{def:alpha}
    For a random variable $X$, we define $\alpha_{1}=u_{min}$, and for every real number $p\in(1,+\infty)$, we define $\alpha_p = \inf\left\{x|F(x)\geq 1-\frac{1}{p}\right\}$.
\end{definition}

The following lemma establishes an interesting property of MHR distributions. Intuitively, the lemma provides a lower bound on the speed of the decay of the tail of a MHR distribution. We prove the lemma by showing that the function $\log_e \big(1-F(x)\big)$ is concave if $F$ is MHR, and exploiting this concavity (see Appendix~\ref{app:proofs for properties MHR}).

\begin{lemma}\label{lem:concentrate}
    If the distribution of a random variable $X$ satisfies MHR, $m\geq 1$ and $d\geq 1$, $d\cdot\alpha_m\geq \alpha_{m^d}$.
\end{lemma}
  
  Next we study the expectation of a random variable that satisfies MHR. We show that the contribution to the expectation from  values $\ge m$, is $O(m\cdot \Pr[X\geq m])$. We start with a definition.
  
\begin{definition}\label{def:contribution}
For a random variable $X$, let $Con[X\geq x]=\mathbb{E}[X|X\geq x]\cdot \Pr\{X\geq x\}$ be the contribution to expectation of $X$ from values which are no smaller than $x$, i.e. 
$$Con[X\geq x] = \int_{x}^{+\infty} t\cdot f(t) dt.$$
\end{definition}

It is an obvious fact that for any random variable X and any two points $x_{1}\leq x_{2}$, $Con[X\geq x_{1}]\geq Con[X\geq x_{2}]$. Using the bound on the tail of a MHR distribution obtained in Lemma~\ref{lem:concentrate}, we bound the contribution to the expectation of $X$ by the values at the tail of the distribution. The proof is given in Appendix~\ref{appendix:MHR}.

\begin{lemma}\label{lem:contribution}
    Let $X$ be a random variable whose distribution satisfies MHR. For all $m\geq 2$, $Con[X\geq \alpha_{m}]\leq 6\alpha_{m}/m$.
\end{lemma}

\subsubsection{Proofs Omitted from Appendix~\ref{appendix:MHR}} \label{app:proofs for properties MHR}

\begin{prevproof}{Lemma}{lem:concentrate}
It is not hard to see that $f(x)>0$, for all $x \in (u_{min},u_{max})$. For a contradiction, assume this is not true, that is, for some $x'\in(u_{min},u_{max})$, $f(x')=0$. We know $1-F(x')>0$. Thus $\frac{f(x')}{1-F(x')}=0$. Since the distribution satisfies MHR and $1-F(x)$ is positive for all $x\in(u_{min},x')$, $f(x) = 0$ in this interval. Hence, it must also be that $F(x)=0$ in $[u_{min},x')$. Since $x'>u_{min}$, it follows that $u_{min}\neq\sup\{x|F(x)=0\}$, a contradiction.
    
    Since $f(x)>0$ in $(u_{min},u_{max})$, $F(x)$ is monotone in $(u_{min},u_{max})$. So we can define the inverse $F^{-1}(x)$ in $(u_{min},u_{max})$. It is not hard to see that for any $p\in[1,+\infty)$, $F(\alpha_{p})=1-1/p$ and $\alpha_p = F^{-1}(1-1/p)$. 
    
    Now let $G(x) = \log_e (1-F(x))$. We will show that $G(x)$ is a concave function.
    
    Let us consider the derivative of $G(x)$. By the definition of MHR, $G'(x) = \frac{-f(x)}{1-F(x)}$ is monotonically non-increasing. Therefore, $G(x)$ is concave. It follows that, for every $m$, by the concavity of $G(x)$, the following inequality holds:
    $$G\left(\frac{d-1}{d}\cdot\alpha_1+\frac{1}{d}\cdot\alpha_{m^d}\right)\geq \frac{d-1}{d} G(\alpha_1)+\frac{1}{d}G(\alpha_{m^d}).$$
    
    Let us rewrite the RHS  as follows
    \begin{align*}
     \frac{d-1}{d} G(\alpha_1)+\frac{1}{d}G(\alpha_{m^d})
       =\frac{d-1}{d}\log_e 1+\frac{1}{d}\log_e (1-F(\alpha_{m^d}))
       =\frac{1}{d} \log_e \left(\frac{1}{m^d}\right)
       =\log_e \left(\frac{1}{m}\right)
    \end{align*}
    
  Hence, we have the following:
    \begin{align*}
        & G\left(\frac{d-1}{d}\cdot\alpha_1+\frac{1}{d}\cdot\alpha_{m^d}\right) \geq \log_e \left(\frac{1}{m}\right)\\
        \Longrightarrow & \log_e \left(1-F\left(\frac{d-1}{d}\cdot\alpha_1+\frac{1}{d}\cdot\alpha_{m^d}\right)\right) \geq \log_e \left(\frac{1}{m}\right)\\
        \Longrightarrow & 1-F\left(\frac{d-1}{d}\cdot\alpha_1+\frac{1}{d}\cdot\alpha_{m^d}\right) \geq \frac{1}{m}\\
        \Longrightarrow & 1-F\left(\frac{d-1}{d}\cdot\alpha_1+\frac{1}{d}\cdot\alpha_{m^d}\right) \geq 1-F(\alpha_m)\\
        \Longrightarrow & F(\alpha_m) \geq F\left(\frac{d-1}{d}\cdot\alpha_1+\frac{1}{d}\cdot\alpha_{m^d}\right)\\
        \Longrightarrow & \alpha_m \geq \frac{d-1}{d}\cdot\alpha_1+\frac{1}{d}\cdot\alpha_{m^d} \quad(F\ is\ monotone\ increasing)\\
        \Longrightarrow & \alpha_m \geq \frac{1}{d}\cdot\alpha_{m^d} \quad(u_{min}\geq 0)\\
        \Longrightarrow & d\cdot \alpha_m \geq \alpha_{m^d}.
    \end{align*}    
\end{prevproof}

\begin{prevproof}{Lemma}{lem:contribution}
Let $S = Con[X\geq \alpha_{m}]$, and consider the sequence $\{\beta_i := \alpha_{m^{(2^i)}}\}$, defined for all non-negative integers $i$. It can easily be seen that $\lim_{i\rightarrow+\infty}\alpha_{m^{(2^{i})}}=u_{max}$; hence, $\lim_{i\rightarrow+\infty} \beta_i = u_{max}$ and  by continuity $\lim_{i\rightarrow+\infty} F(\beta_i) =F(u_{max})=1$.
    
    Also, $$\int_{\beta_i}^{\beta_{i+1}} x\cdot f(x) dx\leq \beta_{i+1}(1-F(\beta_i)) = \beta_{i+1}/m^{(2^i)}.$$ Moreover, Lemma~\ref{lem:concentrate} implies that $\beta_i\leq 2\beta_{i-1}$; thus, $\beta_i\leq 2^{i}\beta_0\leq 2^{i}\alpha_m$. Hence, we have the following:
    \begin{align*}
        S =& \int_{\alpha_{m}}^{u_{max}} x\cdot f(x) dx\leq \sum_{i=0}^{+\infty}\frac{\beta_{i+1}}{m^{(2^i)}}\leq \sum_{i=0}^{+\infty}\frac{2^{i+1}\alpha_m}{m^{(2^i)}}\\
        \leq& \frac{2\alpha_m}{m}+\sum_{i=1}^{+\infty}\frac{2^{(i+1)}\alpha_m}{m^{(2i)}} = \frac{2\alpha_m}{m}+\frac{4\alpha_m}{m^2}\sum_{i=0}^{+\infty}\left(\frac{2}{m^2}\right)^i\\
         =& \frac{2\alpha_m}{m}+\frac{4\alpha_m}{m^2}\cdot \frac{1}{1-2/m^2}\\
         \leq & \frac{2\alpha_m}{m}+\frac{4\alpha_m}{m}\\
         \leq&\frac{6\alpha_m}{m}.
    \end{align*}
    
\end{prevproof}

\subsection{Proof of Theorem~\ref{thm:extreme MHR}: Extreme Value Theorem for MHR Distributions} \label{sec:proof of extreme MHR} \label{sec:beta}
We start with some useful notation. For all $i=1,\ldots,n$, we denote by $F_{i}$  the distribution of variable $X_{i}$.  We also let $\alpha_{m}^{(i)}:=\inf\left\{x|F_i(x)\geq 1-\frac{1}{m}\right\}$, for all $m \ge 1$. Moreover, we assume that $n$ is a power of $2$. If not, we can always include at most $n$ additional random variables that are detreministically $0$, making the total number of variables a power of $2$.

We proceed with the proof of Theorem~\ref{thm:extreme MHR}. The threshold $\beta$ is computed by an algorithm. At a high level, the algorithm proceeds in $O(\log n)$ rounds, indexed by $t \in \{0,\ldots,\log_2 n\}$,   eliminating half of the variables at each round. The way the elimination works is as follows. In round $t$, we compute for each of the variables that have survived so far the threshold $\alpha_{n/2^t}$ beyond which the size of the tail of their distribution becomes smaller than $2^t\over n$. We then sort these thresholds and eliminate the bottom half of the variables, recording the threshold of the last variable that survived this round. The maximum of these records among the $\log_2 n$ rounds of the algorithm is our $\beta$. The pseudocode of the algorithm is given below. Given that we may  only be given oracle access to the distributions $\{F_i\}_{i \in [n]}$, we allow some slack $\eta \le {1 \over 2}$ in the computation of our thresholds so that the computation is efficient. If we know the distributions explicitly, the description of the algorithm simplifies to the case $\eta=0$.
\begin{algorithm}[h!] 
        \caption{Algorithm for finding $\beta$}
    \begin{algorithmic}[1]
        \STATE  Define the permutation of the variables $\pi_{0}(i) = i$, $\forall$ $i\in[n]$, and the set of remaining variables $Q_{0}=[n]$.\\
        \FOR{$t:= 0$ to $\log_2 n-1$}
            \STATE For all $j\in[n/2^t]$, compute some $x_{n/2^{t}}^{(\pi_t(j))} \in[1-\eta,1+\eta]\cdot\alpha_{n/2^{t}}^{(\pi_t(j))}$, for a small constant $\eta \in [0,1/2)$\\
            \STATE Sort these $n/2^t$ numbers in decreasing order $\pi_{t+1}$ such that\\ $x_{n/2^{t}}^{(\pi_{t+1}(1))}\geq x_{n/2^{t}}^{(\pi_{t+1}(2))}\geq\ldots\geq x_{n/2^{t}}^{(\pi_{t+1}(n/2^t))}$\\
            \STATE $Q_{t+1}:=\{\ \pi_{t+1}(i)\ |\ i\leq n/2^{t+1}\ \}$\\
            \STATE $\beta_t:= x_{n/2^{t}}^{(\pi_{t+1}(n/2^{t+1}))}$\\
        \ENDFOR
        \STATE Compute  $x_{2}^{(\pi_{\log_2 n}(1))}\in[1-\eta,1+\eta]\cdot\alpha_{2}^{(\pi_{\log_2 n}(1))}$
        \STATE Set $\beta_{\log_2 n}:=x_{2}^{(\pi_{\log_2 n}(1))}$
        \STATE Output $\beta:=\max_{t} \beta_t$
    \end{algorithmic}
\end{algorithm}

Crucial in the proof of the theorem is the following lemma.
\begin{lemma}\label{lem:tinycontribution}
    For all $i \in [n]$ and $\epsilon \in (0,{1/4})$, let $S_i = Con[X_{i}\geq 2\log_2(\frac{1}{\epsilon})\cdot\beta]$, where $Con[\cdot]$ is defined as in Definition~\ref{def:contribution}. Then $$\sum_{i=1}^n S_i \leq 36\log_2 ({1}/{\epsilon})\epsilon\cdot \beta,~~\text{for all {$\epsilon \in (0,1/4)$}.}$$
\end{lemma}

\begin{proof}
    Let $d = \log_2 (\frac{1}{\epsilon})$ and {notice that $d\geq 2$}. It is not hard to see that we can divide $[n]$ into $(\log_2 n)+1$ different groups $\{G_t\}_{t\in\{0,\ldots,\log_2 n\}}$ based on the sets $Q_t$ maintained by the algorithm, as follows. For $t \in \{0,\ldots,\log_2 n\}$, set $$G_t = \begin{cases} Q_{t}\setminus Q_{t+1} & t<\log_2 n\\ Q_{\log_2 n} & t=\log_2 n\end{cases}$$

    Now, it is not hard to see that, for all $t<\log_2 n$ and all $i\in G_t$, $S_i\leq Con[X_{i}\geq 2d\cdot \beta_{t}]$, since $\beta_{t}\leq \beta$. Also for any $i\in G_{t}$, there must exist some $k\in(n/2^{t+1},n/2^{t}]$, such that $i=\pi_{t+1}(k)$. Then by the definition of the algorithm, we know that 
    $$(1-\eta)\alpha_{n/2^{t}}^{(i)}\leq x_{n/2^{t}}^{(i)}\leq x_{n/2^{t}}^{(\pi_{t+1}(n/2^{t+1}))} = \beta_t.$$
    Recall that $\eta$ is chosen to satisfy $2\geq 1/(1-\eta)$.  Then
    $d\cdot\alpha_{n/2^{t}}^{(i)}\leq 2d\cdot \beta_{t}$. But Lemma~\ref{lem:concentrate} gives $d\cdot\alpha_{n/2^{t}}^{(i)}\geq \alpha_{(n/2^{t})^d}^{(i)}$. Hence, 
    $$2d\cdot \beta_{t}\geq d\cdot\alpha_{n/2^{t}}^{(i)}\geq \alpha_{(n/2^{t})^d}^{(i)},$$
    which implies that 
 $$Con[v_{i}\geq 2d\cdot \beta_{t}]\leq Con[v_{i}\geq\alpha_{(n/2^{t})^{d}}^{(i)}].$$
    
    Using Lemma~\ref{lem:contribution}, we know that $$Con[v_{i}\geq\alpha_{(n/2^{t})^{d}}^{(i)}]\leq 6\alpha_{(n/2^{t})^d}^{(i)}(2^{t}/n)^d\leq 12d\beta_t (2^{t}/n)^d.$$
    
    Now, since $|G_t| = n/2^{t+1}$, 
    $$\sum_{i\in G_t} S_i \leq  12d\beta_{t}(2^{t}/n)^d\times n/2^{t+1} = 6d\cdot\beta_{t} (2^{t}/n)^{d-1} = \frac{6d\cdot\beta_{t}}{n^{d-1}}(2^{d-1})^{t}.$$
    
    Thus, 
    \begin{align*}
    \sum_{i\in [n] \setminus G_{\log_2 n}} S_i \leq& \sum_{t= 0}^{(\log_2 n) -1} \frac{6d\cdot\beta_{t}}{n^{d-1}}(2^{d-1})^{t}\\
     \leq &\frac{6d\cdot \beta}{n^{d-1}}\cdot\frac{(2^{d-1})^{\log_2 n}-1}{2^{d-1}-1}\\
     = &\frac{6d\cdot \beta}{n^{d-1}}\cdot\frac{n^{d-1}-1}{2^{d-1}-1}\\
     \leq & \frac{12d\cdot \beta}{2^{d}-2}\\
      \leq &  \frac{24d\cdot \beta}{2^{d}}\\
      = &24\log_2 ({1}/{\epsilon})\epsilon\cdot\beta
  \end{align*}
  
  Let $i$ be the unique element in $G_{\log_2 n}$. Then $\beta_{\log_2 n} = x_{2}^{(i)}$. Using Lemma~\ref{lem:concentrate} and the definition of $x_{2}^{(i)}$, we obtain 
  $$2d\cdot \beta\geq 2d\cdot \beta_{\log_2 n}\geq 2d\cdot x_{2}^{(i)}\geq 2(1-\eta)d\cdot\alpha_{2}^{(i)}\geq d\cdot\alpha_{2}^{(i)}\geq \alpha_{2^{d}}^{(i)}= \alpha_{1/\epsilon}^{(i)}.$$
   Using the above and Lemma~\ref{lem:contribution} we get
  $$S_i\leq Con[v_{i}\geq\alpha_{1/\epsilon}^{(i)}]\leq 6\epsilon\cdot\alpha_{1/\epsilon}^{(i)}\leq 12\epsilon d\cdot\beta.$$
  
  Putting everything together, $$\sum_{i=1}^n S_i \leq 36\log_2 ({1}/{\epsilon})\epsilon\cdot \beta.$$
\end{proof}
 
Using Lemma~\ref{lem:tinycontribution}, we obtain
$$\int_{2\beta \log_2{1/\epsilon}}^{+\infty} t \cdot f_{\max_i\{X_i\}}(t) dt \le \sum_{i=1}^n S_i \le 36\log_2 ({1}/{\epsilon})\epsilon\cdot \beta. $$

{It remains to show that 
\begin{align}
\Pr[ \max_i\{X_i\} \ge \beta/2] \ge {1-\frac{1}{e^{1/2}}}. \label{eq:probabbiliiiti is large}
\end{align}
We show that, for all $t$, $\Pr\left[ \max_i\{X_i\} \ge {\beta_t \over 1+\eta}\right] \ge {1-\frac{1}{e^{1/2}}}$, where $\eta$ is the parameter used in Algorithm~1. This is sufficient to imply~\eqref{eq:probabbiliiiti is large}, as $\eta \le 1/2$. Observe that for all $i\in[n/2^{t+1}]$, 
    $$(1+\eta)\cdot\alpha_{n/2^{t}}^{(\pi_{t+1}(i))}\geq x_{n/2^{t}}^{(\pi_{t+1}(i))}\geq \beta_{t},$$
    where $\pi_{t+1}$ is the permutation constructed in the $t$-th round of the algorithm. This implies
    $$\alpha_{n/2^{t}}^{(\pi_{t+1}(i))}\geq \frac{\beta_t}{1+\eta}.$$
    Hence, for all $i\in[n/2^{t+1}]$, $\Pr[X_{\pi_{t+1}(i)} \le \frac{\beta_t}{1+\eta}]\leq 1-2^{t}/n$. Thus,
    \begin{align*}
    \Pr\left[ \max_i\{X_i\} \ge {\beta_t \over 1+\eta}\right] &\ge \Pr \left[\exists i\in[n/2^{t+1}],\ X_{\pi_{t+1}(i)}\geq \frac{\beta_t}{1+\eta} \right]\\ 
    &\geq 1-(1-2^{t}/n)^{n/2^{t+1}}\\&\geq 1-\frac{1}{e^{1/2}}.
    \end{align*}
Eq. \eqref{eq:probabbiliiiti is large} now follows.}

\subsection{Proof of Theorem~\ref{thm:reduction MHR to balanced}: Reduction from MHR to Bounded Distributions} \label{app:reduction from MHR to balanced}

Recall that we represent by $\{v_{i}\}_{i\in[n]}$ the values of the buyer for the items. We will denote their distributions by $\{F_i\}_{i\in [n]}$ throughout this appendix. 

 \subsubsection{Relating $OPT$ to $\beta$} \label{sec:OPT vs Beta}

We demonstrate that the anchoring point $\beta$ of Theorem~\ref{thm:extreme MHR} provides a lower bound to the optimal revenue. In particular, we show that the optimal revenue satisfies $OPT=\Omega(\beta)$.  This lemma justifies the relevance of $\beta$.
\begin{lemma}\label{lem:constantopt}
If $\beta$ is the anchoring point of Theorem~\ref{thm:extreme MHR}, then $OPT\geq \left({1-\frac{1}{\sqrt{e}}}\right){\beta \over 2}$.
\end{lemma}
 
\begin{prevproof}{Lemma}{lem:constantopt}
Suppose  we priced all items at $\beta \over 2$. The revenue we would get from such price vector would be at least
$${\beta \over 2} \Pr \left[\max\{v_i\} \ge {\beta \over 2}\right] \ge {\beta \over 2}\left({1-\frac{1}{\sqrt{e}}}\right),$$
where we used Theorem~\ref{thm:extreme MHR}. Hence,
$OPT\geq \left({1-\frac{1}{\sqrt{e}}}\right){\beta \over 2}$.
%
%
%
\end{prevproof}
   \noindent For simplicity, we set $c_{1}:={1 \over 2}\left(1-\frac{1}{\sqrt{e}}\right)$ for the next appendices,  keeping in mind that $c_1$ is an absolute constant.


    \subsubsection{Restricting the Prices} \label{sec:bounding the prices}
 
This appendix culminates in Lemma~\ref{lem:boundingprice} (given below), which states that we can constrain our prices to the set $[\epsilon\cdot\beta, 2\log_2(\frac{1}{\epsilon})\cdot\beta]$ without hurting the revenue by more than a fraction of $\frac{\epsilon+c_{2}(\epsilon)}{c_{1}},$ where $c_{2}(\epsilon):=36\log_2 (\frac{1}{\epsilon})\epsilon$ and $c_1$ is the constant defined in Appendix~\ref{sec:OPT vs Beta}.  We prove this in two steps. First, exploiting our extreme value theorem for MHR distributions (Theorem~\ref{thm:extreme MHR}), we show  that for a given price vector, if we lower the prices that are above $2\log_2(\frac{1}{\epsilon})\cdot\beta$ to $2\log_2(\frac{1}{\epsilon})\cdot\beta$, the loss in revenue is bounded by $c_{2}(\epsilon)\cdot\beta$, namely
  \begin{lemma}\label{lem:trunchighprice}
    Fix an arbitrary $\epsilon \in (0,1/4)$. Given a price vector $P$, define $P'$ as follows: set $p'_i = p_i$, if $p_i\leq 2\log_2(\frac{1}{\epsilon})\cdot\beta$, and  $p'_i = 2\log_2(\frac{1}{\epsilon})\cdot\beta$ otherwise. Then the expected revenues $\mathcal{R}_{P}$ and $\mathcal{R}_{P'}$ achieved by price vectors $P$ and $P'$ respectively satisfy: $ \mathcal{R}_{P'} \geq \mathcal{R}_{P}-c_{2}(\epsilon)\cdot\beta$.
\end{lemma}

Using Lemma~\ref{lem:trunchighprice}, we obtain our main result for this appendix. Observe that we can make the loss in revenue arbitrarily small be taking $\epsilon$ sufficiently small.
  

\begin{lemma}\label{lem:boundingprice}
    For all $\epsilon \in (0,1/4)$, there exists a price vector $P^{*}\in [\epsilon\cdot\beta, 2\log_2(\frac{1}{\epsilon})\cdot\beta]^{n}$, such that the revenue from this price vector satisfies $\mathcal{R}_{P^{*}}\geq\left(1-\frac{\epsilon+c_{2}(\epsilon)}{c_{1}} \right)OPT,$ where $OPT$ is the optimal revenue under any price vector.
\end{lemma}

\noindent All proofs of this appendix can be found in Appendix~\ref{sec:boundingprice}. 
\subsubsection{Truncating the Value Distributions} \label{sec:truncating}

Exploiting Lemma~\ref{lem:boundingprice}, i.e. that we can constrain the prices to $[\epsilon\cdot\beta, 2\log_2(\frac{1}{\epsilon})\cdot\beta]$ without hurting the revenue, we show Theorem~\ref{thm:reduction MHR to balanced}, i.e. that we can also constrain the support of the value distributions into a bounded range. In particular, we show that we can ``truncate'' the value distributions to the range $[{\epsilon \over 2}\cdot\beta, 2\log_2(\frac{1}{\epsilon})\cdot\beta]$, where for our purposes ``truncating'' means this: for every distribution $F_{i}$, we shift all probability mass from $(2\log_2 (\frac{1}{\epsilon})\cdot\beta,+\infty)$ to the point $2\log_2 (\frac{1}{\epsilon})\cdot\beta$, and all probability mass from $(-\infty, \epsilon\cdot\beta)$ to  $\frac{\epsilon}{2}\cdot\beta$.\notshow{\yangnote{Remove: Clearly, this modification can be computed in polynomial time if $F_i$ is known explicitly; if we have oracle access to $F_i$, we can produce in polynomial time another oracle whose output behaves according to the modified distribution.}} We show that our modification does not hurt the revenue. 
That is, we establish a polynomial-time reduction from the problem of computing a near-optimal price vector when the buyer's value distributions are arbitrary MHR distributions to the case where the buyer's value distributions are supported on a bounded interval $[u_{min}, c \cdot u_{min}]$, where $c=c(\epsilon)=4 {1 \over \epsilon} \log_2({1 \over \epsilon})$ is a constant that only depends on the desired approximation $\epsilon$. The proof of Theorem~\ref{lem:reduction} is given in Appendix~\ref{sec:boundingdistribution}.


\begin{theorem}[Reduction from MHR to Bounded Distributions]\label{lem:reduction}
Given $\epsilon \in (0,1/4)$ and a collection of mutually independent random variables $\{v_i\}_i$ that are MHR, let us define a new collection of random variables $\{\tilde{v}_i\}_i$ via the following coupling: for all $i\in [n]$, set  $\tilde{v}_i={\epsilon \over 2}\cdot \beta$ if $v_i < \epsilon \cdot \beta$, set $\tilde{v}_i=2\log_2 (\frac{1}{\epsilon})\cdot\beta$ if $v_i\ge2\log_2 (\frac{1}{\epsilon})\cdot\beta$, and set $\tilde{v}_i=v_i$ otherwise, where $\beta = \beta(\{v_i\}_i)$ is the anchoring point of Theorem~\ref{thm:extreme MHR}
computed from the distributions of the variables $\{v_i\}_i$. Let also $\widetilde{OPT}$ be the optimal revenue of the seller when the buyer's values are distributed as $\{\tilde{v}_i\}_{i\in[n]}$ and $OPT$  the optimal revenue when the buyer's values are distributed as $\{v_i\}_{i\in[n]}$. Then given a price vector that achieves revenue $(1-\delta)\cdot \widetilde{OPT}$ when the buyer's values are distributed as $\{\tilde{v}_i\}_{i\in[n]}$, we can efficientlly compute a price vector with revenue 
$$\left(1-\delta-\frac{2\epsilon+3c_2(\epsilon)}{c_1}\right)OPT$$
when the buyer's values are distributed as $\{{v}_i\}_{i\in[n]}$.
\end{theorem}
Theorem~\ref{thm:reduction MHR to balanced} follows from Theorem~\ref{lem:reduction}.



%

\subsection{Proofs Omitted from Appendix~\ref{app:reduction from MHR to balanced}}\label{appendix:bounding}
\subsubsection{Restricting the Price Range for MHR Distributions: the Proofs}\label{sec:boundingprice}

\begin{prevproof}{Lemma}{lem:trunchighprice}
We will show a slightly more general result. Given a price vector, if we make all prices that are above $\alpha$ equal to $\alpha$, then the loss in revenue can bounded by the sum, over all items whose price was above $\alpha$, of the contribution to this item's expected value by points above $\alpha$.  Formally,


\begin{lemma}\label{lem:trunchigh}
Let $\alpha>0$ and $S(\alpha) = Con[\max_i{v_{i}}\geq \alpha]$, where $Con[\cdot]$ is defined as in Definition~\ref{def:contribution}.  Moreover, for a given price vector $P$, define $P'$ as follows: set $p'_i = p_{i}$, if $p_i<\alpha$, and $p'_i = \alpha$, otherwise. Then the expected revenues $\mathcal{R}_{P}$ and $\mathcal{R}_{P'}$ from $P$ and $P'$ respectively satisfy $$\mathcal{R}_{P'} \ge \mathcal{R}_{P}-S(\alpha).$$

\end{lemma}

\begin{proof}
Let $S_{exp} = \{\ i\ |\ p_i > \alpha\ \}$ be the set of expensive items under $P$, and let $P''$ be a new price vector obtained from $P$ by changing the price of all items $i \in S_{exp}$ from $p_{i}$ to $+\infty$. If we switch from $P$ to $P''$, the only case where the buyer makes a different decision is when she used to buy some item from $S_{exp}$ under $P$. So the decrease in revenue can be bounded by the contribution to $\mathcal{R}_{P}$ from the items in $S_{exp}$. Clearly, this contribution is at most $ Con\left[\max_{i\in S_{exp}} v_{i}\geq \alpha \right] \le S(\alpha)$.

We proceed to argue that $\mathcal{R}_{P'} \ge \mathcal{R}_{P''}$. If we switch from $P''$ to $P'$, the buyer will pointwise either make the same decision or switch to buy some item in $S_{exp}$ paying $\alpha$. Since $\alpha$ is larger than any finite price in $P''$, the revenue does not decrease. 

Combining the two inequalities, we have $\mathcal{R}_{P'}\geq \mathcal{R}_{P''}\geq \mathcal{R}_{P}-S(\alpha).$
\end{proof}

\notshow{We showed that setting the price of expensive items to $+\infty$ is not detrimental to the revenue. We show that the same is true of a less aggressive strategy.

     
\begin{lemma}\label{lem:decreaseprice}
     Let $P$ be a price vector, define $S_{\infty}= \{i: p_{i}=+\infty \}$, and let $p_{max} \geq \max_{i\in[n] \setminus S_{\infty}}p_i$. Let $P'$ be the following price vector: set $p'_i = p_{max}$, for all $i\in S_{\infty}$, and $p'_i=p_i$ otherwise. Then the expected revenues $\mathcal{R}_{P}$ and $\mathcal{R}_{P'}$ from $P$ and $P'$ respectively satisfy $$\mathcal{R}_{P'}\geq\mathcal{R}_{P}.$$
 \end{lemma}
 
 \begin{proof}
     Take any valuation vector $(v_1,v_2,\ldots,v_n)$ and suppose that $i$ is the winner under price vector $P$, i.e. $v_i-p_i\geq v_j-p_j$, for all $j$, and $v_i-p_i\geq 0$. Under $P'$, there are two possibilities: (1) $i$ is still the winner; then the contribution to the revenue is the same under $P$ and $P'$. (2) $i$ is not the winner; in this case, the winner should be among those items whose price was lowered from $P$ to $P'$. So the new winner must some $j \in S_{\infty}$. But notice that $p'_j = p_{max}\geq p_i$. 
     Hence, the contribution of item $j$ to revenue under price vector $P'$ is not smaller than the contribution of item $i$ to the revenue under price vector $P$. Hence, $\mathcal{R}_{P'}\geq\mathcal{R}_{P}$. 
\end{proof}}

Combining Theorem~\ref{thm:extreme MHR} with Lemma~\ref{lem:trunchigh}, it is easy to argue that if we truncate a price vector $P$ at value $2\log_2(\frac{1}{\epsilon})\cdot\beta$ to obtain a new price vector $P'$ the change in revenue can be bounded as follows: $ \mathcal{R}_{P'} \geq \mathcal{R}_{P}-c_{2}(\epsilon)\cdot\beta.$
\end{prevproof}


\begin{prevproof}{Lemma}{lem:boundingprice}
	Lemma~\ref{lem:trunchighprice} implies that, if we start from any price vector $P$, we can modify it into another price vector $P'$ that does not use any price above $2\log_2(\frac{1}{\epsilon})\cdot\beta$, and satisfies $\mathcal{R}_{P'}\geq \mathcal{R}_{P}-c_{2}(\epsilon)\cdot \beta$. Then Lemma~\ref{lem:trunclow} implies that we can change $P'$ into another vector $P''\in[\epsilon\cdot\beta, 2\log_2(\frac{1}{\epsilon})\cdot\beta]^{n}$, such that $\mathcal{R}_{P''}\geq \mathcal{R}_{P'}-\epsilon\cdot\beta$.
	
	By Lemma~\ref{lem:constantopt}, we know that $OPT\geq c_{1}\cdot\beta$. Hence, if we start with the optimal price vector $P$ and apply the above transformations, we will obtain a price vector $P^{*} \in[\epsilon\cdot\beta, 2\log_2(\frac{1}{\epsilon})\cdot\beta]^{n}$ such that
	$$\mathcal{R}_{P^{*}}\geq OPT-\big(\epsilon+c_{2}(\epsilon)\big)\cdot\beta\geq \left(1-\frac{\epsilon+c_{2}(\epsilon)}{c_{1}} \right)OPT.$$
\end{prevproof}

\subsubsection{Bounding the Support of the Distributions: the Proofs}\label{sec:boundingdistribution}

To establish Theorem~\ref{lem:reduction} we show that we can transform $\{v_i\}_{i\in[n]}$ into $\{\tilde{v}_i\}_{i\in[n]}$ such that, for all $i$, $\tilde{v}_i$ only takes values in $[\frac{\epsilon}{2}\cdot\beta, 2\log_2 (\frac{1}{\epsilon})\cdot\beta]$, and for any price vector $P\in[\epsilon\cdot\beta, 2\log_2(\frac{1}{\epsilon})\cdot\beta]^{n}$, $|\tilde{\mathcal{R}}_{P}-\mathcal{R}_{P}|\leq c_{2}(\epsilon)\cdot\beta$, where  $\mathcal{R}_{P}$ and $\tilde{\mathcal{R}}_{P}$ are respectively the revenues of the seller when the buyer's values are distributed as  $\{v_i\}_{i\in[n]}$ and $\{\tilde{v}_i\}_{i\in[n]}$. We first  show that one side of our truncation works.

\begin{lemma}\label{lem:boundinghigh}
   Given $\epsilon\in (0,1/4)$ and a collection of random variables $\{v_i\}_i$ that are MHR, let us define a new collection of random variables $\{\hat{v}_i\}_i$ via the following coupling: for all $i\in [n]$, if $v_i\leq 2\log_2 (\frac{1}{\epsilon})\cdot\beta$, set $\hat{v}_i=v_i$, otherwise set $\hat{v}_i = 2\log_2 (\frac{1}{\epsilon})\cdot\beta$, where $\beta = \beta(\{v_i\}_i)$ is the anchoring point of Theorem~\ref{thm:extreme MHR}
computed from the distributions of the variables $\{v_i\}_i$. Then, for any price vector $P\in [\epsilon\cdot\beta, 2\log_2(\frac{1}{\epsilon})\cdot\beta]^{n}$, $|\mathcal{R}_{P}-\hat{\mathcal{R}}_{P}|\leq c_{2}(\epsilon)\cdot\beta$, where $\mathcal{R}_P$ and $\hat{\mathcal{R}}_P$ are respectively the revenues of the seller when the buyer's values are distributed as $\{v_i\}_{i\in[n]}$ and as $\{\hat{v}_i\}_{i\in[n]}$.
\end{lemma}

\begin{proof}  For convenience let $d=\log_2 (\frac{1}{\epsilon})$, and let  $R_P$ and $\hat{R}_P$ be random variables representing the revenue when the buyer's values are $\{v_i\}_{i\in[n]}$ and $\{\hat{v}_i\}_{i\in[n]}$ respectively. Recall that $\{{v}_i\}_i$ and $\{\hat{v}_i\}_i$ are defined via a coupling, so $R_P-\hat{R}_P \neq 0$ only in the event $v_i\neq \hat{v}_i$, for some $i$. Notice that the probability of this event is $\Pr[\exists i, v_i> 2d\cdot\beta]$, and that $0\le R_P, \hat{R}_P \le 2d\cdot\beta$ since the maximum price of any item is $2d\cdot\beta$. Hence, we can bound $|\mathbb{E}[R_P]-\mathbb{E}[\hat{R}_P]|$ using Theorem~\ref{thm:extreme MHR} as follows:
%
\begin{align*}
|\mathbb{E}[R_P]-\mathbb{E}[\hat{R}_P]| &\le 2d\cdot\beta\cdot\Pr[\exists i, v_i> 2d\cdot\beta]\\
&=2d\cdot\beta\cdot\Pr[\max_{i}{v_{i}}\geq 2d\cdot\beta]\leq Con[\max_{i}{v_{i}}\geq 2d\cdot\beta]\leq c_{2}(\epsilon)\cdot\beta .
\end{align*}
    \notshow{For convenience let $d=\log_2 (\frac{1}{\epsilon})$. Recall that $\{{v}_i\}_i$ and $\{\hat{v}_i\}_i$ are defined via a coupling. We distinguish two events: (1) $v:=(v_1,v_2,\ldots,v_n) \equiv (\hat{v}_1,\hat{v}_2,\ldots,\hat{v}_n)=:\hat{v}$; (2) $v_i\neq \hat{v}_i$, for some $i$.
    
    Under Event (1), the item sold is the same under $v$ and $\hat{v}$. Hence, the revenue is the same under Event (1). So we only need to worry about Event (2). We show that the total contribution to the revenue from this event is very small. Indeed, the probability of this event is $\Pr[\exists i, v_i> 2d\cdot\beta]$, and the maximum price is $2d\cdot\beta$. So using our extreme value theorem (Theorem~\ref{thm:extreme MHR}), we can bound the contribution to the revenue from this event as follows.
    $$2d\cdot\beta\cdot\Pr[\exists i, v_i> 2d\cdot\beta]=2d\cdot\beta\cdot\Pr[\max_{i}{v_{i}}\geq 2d\cdot\beta]\leq Con[\max_{i}{v_{i}}\geq 2d\cdot\beta]\leq c_{2}(\epsilon)\cdot\beta .$$
    We obtain $|\mathcal{R}_{P}-\hat{\mathcal{R}}_{P}|\leq c_{2}(\epsilon)\cdot\beta$.}
\end{proof}

Next we show that the other side of the truncation works.

\begin{lemma}\label{lem:boundinglow}
       Given $\epsilon, \beta >0$ and a collection of random variables $\{\hat{v}_i\}_i$, let us define a new collection of random variables $\{\tilde{v}_i\}_i$ via the following coupling: for all $i\in [n]$, if $\hat{v}_i\geq \epsilon\cdot\beta$, set $\tilde{v}_i=\hat{v}_i$, otherwise set $\tilde{v}_{i}=\frac{\epsilon}{2}\cdot\beta$. Then, for any price vector $P\in [\epsilon\cdot\beta,+\infty)^{n}$, $\tilde{\mathcal{R}}_{P}=\hat{\mathcal{R}}_{P},$  where  $\hat{\mathcal{R}}_P$ and $\tilde{\mathcal{R}}_P$ are respectively the revenues of the seller when the buyer's values are distributed as $\{\hat{v}_i\}_{i\in[n]}$ and as $\{\tilde{v}_i\}_{i\in[n]}$.
\end{lemma}

\begin{proof}

Recall that the variables $\{\tilde{v}_i\}_i$ are defined via a coupling with the $\hat{v}_i$'s. Under the same coupling, sample values from the $\hat{v}_i$'s and the $\tilde{v}_i$'s. For all items $i$ such that $\hat{v}_{i}\neq \tilde{v}_{i}$, the price of item $i$ is higher than the value of item $i$ in both cases, so item $i$ will not be purchased in both cases. That means the buyer will make the same decision in both cases, as she will only consider items whose values are the same. So the revenues are pointwise equal.
\end{proof}

\noindent Putting these lemmas together we obtain our reduction.\\

\begin{prevproof}{Theorem}{lem:reduction}
Let $P$ be a near-optimal price vector when the values of the buyer are distributed as $\{\tilde{v}_i\}_{i\in[n]}$, i.e. one that satisfies
$$ \tilde{\mathcal{R}}_{P} \ge (1-\delta)\cdot \widetilde{OPT},$$
where $\tilde{\mathcal{R}}_{P}$ denotes the expected revenue of the seller under price vector $P$ when the buyer's values are $\{\tilde{v}_i\}_{i\in[n]}$.
Given that each $\tilde{v}_i$ lies in $[\frac{\epsilon}{2}\cdot\beta, 2\log_2(\frac{1}{\epsilon})\cdot\beta]$, it follows from Lemma~\ref{lem:betweenminmax} that we can (efficiently) transform $P$ into another vector $P'\in[\frac{\epsilon}{2}\cdot\beta, 2\log_2(\frac{1}{\epsilon})\cdot\beta]^n$, such that $\tilde{\mathcal{R}}_P\leq \tilde{\mathcal{R}}_{P'}$.

We can then apply the following efficient transformation from $P'$ to $P''$: For any $i$, if $p'_i<\epsilon\cdot\beta$, set $p''_i=\epsilon\cdot\beta$, and set $p''_i=p'_i$ otherwise. By Lemma~\ref{lem:trunclow}, we know that, $\tilde{\mathcal{R}}_{P''}\geq \tilde{\mathcal{R}}_{P'}-
\epsilon\cdot\beta.$

Now, since $P''$ is a price vector in $[\epsilon\cdot\beta, 2\log_2(\frac{1}{\epsilon})\cdot\beta]^n$, by Lemmas~\ref{lem:boundinghigh} and \ref{lem:boundinglow}, we get $\mathcal{R}_{P''}\geq \tilde{\mathcal{R}}_{P''} - c_2(\epsilon)\cdot\beta,$
where $\mathcal{R}_{P''}$ is the expected revenue of the seller under price vector $P''$ when the values of the buyers are $\{v_i\}_i$.

On the other hand, suppose that $P^*$ is the optimal price vector in $[\epsilon\cdot\beta, 2\log_2(\frac{1}{\epsilon})\cdot\beta]^n$ for values $\{v_i\}_{i\in[n]}$. By Lemma~\ref{lem:boundingprice}, we know that $\mathcal{R}_{P^*}\geq \left(1-\frac{\epsilon+c_2(\epsilon)}{c_1}\right)OPT$. Now Lemmas~\ref{lem:boundinghigh} and \ref{lem:boundinglow} give $$\tilde{\mathcal{R}}_{P^*}\geq \mathcal{R}_{P^*}-c_2(\epsilon)\cdot\beta \geq \left(1-\frac{\epsilon+c_2(\epsilon)}{c_1}\right)OPT-c_2(\epsilon)\cdot\beta \geq \left(1-\frac{\epsilon+2c_2(\epsilon)}{c_1}\right)OPT,$$
where we used that $OPT \ge c_1 \cdot \beta $, by Lemma~\ref{lem:constantopt}.

Since $\widetilde{OPT}\geq \tilde{\mathcal{R}}_{P^*}$, $$\tilde{\mathcal{R}}_{P'}\geq \tilde{\mathcal{R}}_P \ge (1-\delta) \widetilde{OPT} \ge(1-\delta)\left(1-\frac{\epsilon+2c_2(\epsilon)}{c_1}\right)OPT\geq \left(1-\delta-\frac{\epsilon+2c_2(\epsilon)}{c_1}\right)OPT.$$

Recall that $\mathcal{R}_{P''}\geq \tilde{\mathcal{R}}_{P''} - c_2(\epsilon)\cdot\beta\geq \tilde{\mathcal{R}}_{P'}-\epsilon\cdot\beta-c_2(\epsilon)\cdot\beta$. Therefore, $$\mathcal{R}_{P''}\geq \left(1-\delta-\frac{\epsilon+2c_2(\epsilon)}{c_1}\right)OPT-\epsilon\cdot\beta-c_2(\epsilon)\cdot\beta\geq \left(1-\delta-\frac{2\epsilon+3c_2(\epsilon)}{c_1}\right)OPT.$$

So given a near-optimal price vector $P$ for $\{\tilde{v}_i\}_{i\in[n]}$, we can construct a near-optimal price vector $P''$ for $\{v_i\}_{i\in[n]}$ in polynomial time. 
\end{prevproof}

\section{Details of Section~\ref{sec:regular}}\label{appendix:regularbounding}

We start by establishing some useful properties of regular distributions, and proceed to show our extreme value theorem (Theorem~\ref{thm:regextremevalue}), and our reduction from item pricing problems with regular distributions to item pricing problems with bounded distributions (Theorem~\ref{thm:regreduction}).

\subsection{Basic Properties of Regular Distributions}\label{sec:regconcavity}
If $F$ is a differentiable continuous regular distribution, it is not hard to see the following: if $f(x)=0$ for some {$x\in(u^{X}_{min},u_{max}^{X})$, then $f(x')=0$ for all $x'< x$} (as otherwise the definition of regularity would be violated.) Hence, if $X$ is a random variable distributed according to $F$, it must be that $f(x)>0$ for $x\in[u_{min}^{X},u_{max}^{X}]$. So we can define $F^{-1}$ on $[u_{min}^{X},u_{max}^{X}]$, and it will be differentiable, since $F$ is differentiable and $f$ is non-zero. Now we can make the following definition, capturing the revenue of a seller who prices an item with value distribution $F$, so that the item is bought with probability exactly $q$.

\begin{definition}[Revenue Curve]
    For a differentiable continuous regular distribution $F$, define $R_F: [0,1] \rightarrow \mathbb{R}$ as follows
    $$R_F(q) =q\cdot F^{-1}(1-q).$$
\end{definition}
The following is well-known. We include its short proof for completeness.
\begin{lemma}\label{lem:regconcavity}
    If $F$ is regular, $R_F(q)$ is a concave function on $(0,1]$.
    \end{lemma}
\begin{proof}
    The derivative of $R_F(q)$ is
    $$R_F'(q) = F^{-1}(1-q)-{q\over f\big{(}F^{-1}(1-q)\big{)}}.$$
    Notice that $F^{-1}(1-q)$ is monotonically non-increasing in $q$. This observation and the regularity of $F$ imply that  $R_F'(q)$ is monotonically non-increasing in $q$. (To see this try the change of variable $x(q)= F^{-1}(1-q)$.) This implies that $R_F(q)$ is concave.
\end{proof}

\begin{lemma}\label{lem:reglittlegain}
    For any regular distribution $F$, if $0<\tilde{q}\le q\leq {p}<1$, then 
    $$R_F(\tilde{q})\leq {1\over 1-p}R_{F}(q).$$ 
\end{lemma}
\begin{proof}
    Since $q\in[\tilde{q},1)$, there exists a $\lambda\in(0,1]$, such that 
    $$\lambda\cdot\tilde{q}+(1-\lambda)\cdot 1=q.$$
    Hence: $\lambda = {1-q\over 1-\tilde{q}}\geq {1-{p}\over 1} = 1-{p}.$
    Now, from Lemma~\ref{lem:regconcavity}, we have that $R_F(x)$ is concave. Thus
    $$R_F(q)=R_F\big{(}\lambda\cdot\tilde{q}+(1-\lambda)\cdot1\big{)}\geq\lambda\cdot R_F(\tilde{q})+(1-\lambda)\cdot R_F(1).$$ 
    Since $R_F(1)\geq 0$, $R_F(q)\geq \lambda\cdot R_F(\tilde{q})\geq (1-p)R_F(\tilde{q})$. Thus, $R_F(\tilde{q})\leq {1\over 1-p}R_{F}(q)$.
\end{proof}    

\begin{corollary}\label{cor:reglittlegain}
        For any regular distribution $F$, if $\tilde{q}\leq q\leq {1\over n^{3}}$,  then 
$$R_F(\tilde{q})\leq {n^{3}\over n^{3}-1}R_{F}(q).$$ 
\end{corollary}

\subsection{Proof of Theorem~\ref{thm:regextremevalue}: Extreme Value Theorem for Regular Distributions} \label{sec:proof of extreme Regular} \label{sec:regextremevalue}

We define $\alpha$ explicitly from the distributions $\{F_i\}_i$ of the  variables $\{X_i\}_i$. We first need a definition.
\begin{definition}\label{def:regalpha}
A point $x$ is a \textsl{$(c_1,c_2)$-anchoring point} of a  distribution $F$, if  $F(x) \in [c_1,c_2]$.
\end{definition}

Now fix two arbitrary constants $0<c_1<c_2\le {7 \over 8}$, and let, for all $i$,  $\alpha_{i}$ be a $(c_1,c_2)$-anchoring point of the distribution $F_{i}$. Then define 
$$\alpha= {n^{3}\over c_{1}}\cdot\max_{i} \Big{[}\alpha_{i}\cdot \big{(}1-F_{i}(\alpha_{i})\big{)}\Big{]}.$$
Clearly, a collection $\alpha_1,\ldots,\alpha_n$ of $(c_1,c_2)$-anchoring points can be computed efficiently from the $F_i$'s.  \notshow{\yangnote{Remove: (whether we have these distributions explicitly or have oracle access to them.)}} Hence, an $\alpha$ as above can be computed efficiently. We proceed to establish anchoring properties satisfied by $\alpha$.
\begin{proposition}\label{prop:regalpha}
$\alpha\geq \max_i \alpha_{n^3}^{(i)},$ where $\alpha^{(i)}_p = \inf\left\{x|F_i(x)\geq 1-\frac{1}{p}\right\}$as in Definition~\ref{def:alpha}.
\end{proposition}
\begin{proof}
Because $1/n^{3}\leq 1-c_{2}\leq 1-F(\alpha_{i})\leq 1-c_{1}$, it follows from Lemma~\ref{lem:reglittlegain} that $${1\over c_{1}}\cdot \alpha_{i}\cdot\big{(}1-F_{i}(\alpha_{i})\big{)}\geq \alpha_{n^{3}}^{(i)}/n^{3}.$$ Hence: $\alpha\geq {n^{3}\over c_{1}}\cdot \Big{[}\alpha_{i}\cdot \big{(}1-F_{i}(\alpha_{i})\big{)}\Big{]}\geq \alpha_{n^{3}}^{(i)}$. This is true for all $i$, hence the theorem.
\end{proof}

\begin{prevproof}{Theorem}{thm:regextremevalue}
We first show that $\Pr[X_{i}\geq \ell\alpha]\leq 2/(\ell n^{3})$, for any $\ell \ge 1$. By Proposition~\ref{prop:regalpha} and Corollary~\ref{cor:reglittlegain}, we have that $(\ell \alpha)\Pr[X_{i}\geq \ell \alpha]\leq {n^{3}\over n^{3}-1}\alpha\Pr[X_{i}\geq \alpha]$. Thus 
\begin{align}
\Pr[X_{i}\geq \ell \alpha]\leq {n^{3}\over n^{3}-1} \cdot {1 \over \ell}\cdot\Pr[X_{i}\geq \alpha]\leq 2/(\ell n^{3}), \label{eq:usefffuelel}
\end{align}
which establishes the first anchoring property satisfied by $\alpha$.

Moreover, we have that
$$\alpha/n^{3}={1\over c_{1}}\cdot\max_{i} \Big{[}a_{i}\cdot \big{(}1-F_{i}(a_{i})\big{)}\Big{]}\leq{1\over c_{1}} \max_{z}(z\cdot \Pr[\max_{i}\{X_{i}\}\geq z]),$$
which establishes the second anchoring property of $\alpha$.

Finally, we demonstrate the homogenization property of $\alpha$. We want to show that, for any integer $m\leq n$, thresholds $t_{1},\ldots, t_{m}\geq t \geq {2n^{2}\alpha\over \epsilon^{2}}$,  index set  $S=\{a_1,\ldots,a_m\}\subseteq [n]$, and $\epsilon \in (0,1)$:
\begin{align}\label{eq:regular extreme}
\sum_{i=1}^{m} t_{i} \Pr[X_{a_{i}}\geq t_{i}]\leq \left(t-{2\alpha\over \epsilon}\right)\cdot \Pr\left[\max_{i}\{X_{a_{i}}\} \geq t\right]+{7\epsilon\cdot(2\alpha/\epsilon\cdot\Pr[\max_{i}\{X_{a_{i}}\}\geq 2\alpha/\epsilon])\over n}.\end{align}

\noindent For notational simplicity, we define $f_{i}(z_{i})=z_{i}\cdot \Pr[X_{a_{i}}\geq z_{i}]$ and $f_{max}^{(S)}(z)=z\cdot\Pr[\max_{i}\{X_{a_{i}}\}\geq z]$. Notice that for any $t_{i}\geq t\geq 2\alpha/\epsilon$, a double application of Proposition~\ref{prop:regalpha}, Lemma~\ref{lem:reglittlegain} and Equation~\eqref{eq:usefffuelel} gives \begin{align}f_{i}(t_{i})\leq {(n^3/\epsilon)\over (n^3/\epsilon)-1}f_{i}(t) \leq {2(n^3/\epsilon)\over (n^3/\epsilon)-1}f_{i}\left({2\alpha\over\epsilon}\right).\label{eq:ususuuseful}\end{align} Thus,
\begin{align*}
LHS\ of\ (\ref{eq:regular extreme}) &\leq \sum_{i=1}^{m}f_{i}(t)+{1\over (n^3/\epsilon)-1}\sum_{i=1}^{m}f_{i}(t)\\
      &\leq  \sum_{i=1}^{m}f_{i}(t) +{2\over (n^3/\epsilon)-1}\sum_{i=1}^{m}f_{i}\left({2\alpha\over \epsilon}\right)\\
	&\leq \sum_{i=1}^{m}f_{i}(t)+{2n\over (n^3/\epsilon)-1} f_{max}^{(S)}\left({2\alpha\over \epsilon}\right)\\
         &\leq \sum_{i=1}^{m}f_{i}(t)+{2\epsilon\over n} f_{max}^{(S)}\left({2\alpha\over \epsilon}\right).
\end{align*}
On the other hand, for any $t\ge 2\alpha/\epsilon$: $\Pr[X_{a_i}\geq t]\leq \Pr[X_{a_i}\geq2\alpha/\epsilon]\leq \epsilon/n^{3}$ (using~\eqref{eq:usefffuelel}). Thus: 
\begin{align}
\sum_{i} \Pr[X_{a_i}\geq t]\geq\Pr[\max_{i}\{X_{a_{i}}\}\geq t]\geq (1-\epsilon/n^{2})\sum_{i} \Pr[X_{a_i}\geq t],\label{eq:ususususususul}
\end{align}
where the last inequality follows from the fact that, for all $i$, the probability that $X_{a_{i}}\geq t$, while $X_{a_j}<t$ for all $j \in S \setminus\{i\}$ is at least $\Pr[X_{a_i}\geq t](1-\epsilon/n^{3})^{m-1}\ge \Pr[X_{a_i}\geq t](1-\epsilon/n^{2})$. Therefore, continuing our upper-bounding from above:
\begin{align*}
LHS\ of\ (\ref{eq:regular extreme})\le &\quad\sum_{i=1}^{m}f_{i}(t)+{2\epsilon\over n} f_{max}^{(S)}\left({2\alpha\over \epsilon}\right)\\
&\leq (t-2\alpha/\epsilon)\Pr[\max_{i}\{X_{a_{i}}\}\geq t] +(2\alpha/\epsilon)\Pr[\max_{i}\{X_{a_{i}}\}\geq t]+(\epsilon/n^{2}) \sum_{i=1}^{m} f_{i}(t) +{2\epsilon\over n} f_{max}^{(S)}\left({2\alpha\over \epsilon}\right)\\
&\leq (t-2\alpha/\epsilon)\Pr[\max_{i}\{X_{a_{i}}\}\geq t]+(2\alpha/\epsilon t)\sum_{i=1}^{m} f_{i}(t)+(2\epsilon/n^{2}) \sum_{i=1}^{m} f_{i}\left({2\alpha\over\epsilon}\right)+{2\epsilon\over n} f_{max}^{(S)}\left({2\alpha\over \epsilon}\right)\\
&\leq  (t-2\alpha/\epsilon)\Pr[\max_{i}\{X_{a_{i}}\}\geq t]+(\epsilon/n^{2})\sum_{i=1}^{m} f_{i}(t)+{4\epsilon\over n} f_{max}^{(S)}\left({2\alpha\over \epsilon}\right)\\
&\leq (t-2\alpha/\epsilon)\Pr[\max_{i}\{X_{a_{i}}\}\geq t]+{6\epsilon\over n} f_{max}^{(S)}\left({2\alpha\over \epsilon}\right),
\end{align*}
where we got the third inequality by invoking~\eqref{eq:ususuuseful} and~\eqref{eq:ususususususul}, the fourth inequality by invoking~\eqref{eq:ususususususul} with $t=2\alpha/\epsilon$, and the fifth inequality by invoking~\eqref{eq:ususuuseful} and then~\eqref{eq:ususususususul} with $t=2\alpha/\epsilon$. This concludes the proof of Theorem~\ref{thm:regextremevalue}.
\end{prevproof}

\subsection{Proof of Theorem~\ref{thm:regreduction}: Reduction from Regular to Bounded Distributions} \label{app: reduction for regular}
\subsubsection{Restricting the Prices for the Input Regular Distributions}\label{sec:regrestrictprice}
\begin{lemma}\label{lem:regrestrictprice}
    Let $\mathcal{V}=\{v_{i}\}_{i\in[n]}$ be a collection of independent regular value distributions, $\epsilon \in (0,1)$, and $c$  the absolute constant  in the statement of Theorem~\ref{thm:regextremevalue}. For any price vector $P$, we can construct a new price vector $\hat{P}\in [\epsilon\alpha/n^{4},2n^{2}\alpha/\epsilon^{2}]^{n}$, such that $\mathcal{R}_{\hat{P}}\geq \mathcal{R}_{P}-{(c+{9})\epsilon\mathcal{R}_{OPT}\over n}$, where $\mathcal{R}_{P}$ and $\mathcal{R}_{\hat{P}}$ are respectively the expected revenues under price vectors $P$ and $\hat{P}$, and ${\cal R}_{OPT}$ is the optimal expected revenue for $\cal V$.
\end{lemma}
\begin{proof}
	{\bf First step:} We first construct a price vector $P'\in[0,2n^{2}\alpha/\epsilon^{2}]^{n}$ based on $P$, such that the revenue under ${P'}$ is at most an additive $O({\epsilon\cdot\mathcal{R}_{OPT}\over n})$ smaller than the revenue under ${P}$.
	
	We define ${ P}'$ as follows. Let $S=\{i\ \vline\ p_{i}>2n^{2}\alpha/\epsilon^{2}\}$. For any $i\in S$ set $p'_{i}={2(n^{2}/\epsilon-1)\alpha\over \epsilon}$, while if $i\notin S$ set $p'_{i}=p_{i}$. Now assume $|S|=m$. For notational convenience we assume that $S=\{a_{i}\ \vline\ i\in[m]\}$, and set $X_{a_{i}}=v_{a_{i}}$. 
	Moreover, let $t={2n^{2}\alpha\over \epsilon^{2}}$ and $t_{i} = p_{a_{i}}$.  

Clearly, the contribution to ${\cal R}_P$ from items in $S$ is upper bounded by $\sum_{i=1}^{m} t_{i} \Pr[X_{a_i}\geq t_{i}]$. We proceed to analyze the contribution to revenue ${\cal R}_P'$ from items in $S$. Notice that, when $\max_{i\in S}\{v_{i}\}=\max_{i} \{X_{a_{i}}\}\geq t$, the largest value-minus-price gap for items in $S$ is at least $2\alpha/\epsilon$ (given our subtle choice of prices for items in $S$ above). Hence, for the item of $S$ achieving this gap not to be the winner, it must be that some item in $[n]\setminus S$ has a larger value-minus-price gap. For this to happen, the value for this item has to be higher than $2\alpha/\epsilon$. However, the probability that there exists an item in $[n]\setminus S$ with value greater than $2\alpha/\epsilon$ is smaller than $n\cdot \epsilon/n^{3}=\epsilon/n^{2}$ (by Theorem~\ref{thm:regextremevalue}). Thus, when $\max_{i}\{X_{a_{i}}\}\geq t$, then with probability at least $1-\epsilon/n^{2}$, the item in $S$ achieving the largest value-minus-price gap is the item bought by the buyer. So when the price vector is $P'$, the revenue from the items in $S$ is lower bounded by $(t-2\alpha/\epsilon)\Pr[\max_{i}\{X_{a_i}\} \geq t](1-\epsilon/n^{2})$ (where we used independence and the fact that $p_i'=t-2\alpha/\epsilon$ for all $i\in S$.) 

Clearly, $(t-2\alpha/\epsilon)\Pr[\max\{X_{a_{i}}\}\geq t]\leq t\Pr[\max_{i}\{X_{a_{i}}\}\geq t] \leq \mathcal{R}_{OPT}$. To see this, notice that the first inequality is obvious and the second  follows from the observation that we could set the prices of all items in $S$ to $t$ and of all other items to $+\infty$ to achieve revenue $t\Pr[\max_{i}\{X_{a_{i}}\}\geq t]$. So ${\cal R}_{OPT}$ should be larger than this revenue. Similarly, we  see that $2\alpha/\epsilon\cdot \Pr[\max_{i}X_{a_{i}}\geq 2\alpha/\epsilon]\leq \mathcal{R}_{OPT}$. Using these observations and Theorem~\ref{thm:regextremevalue} we get
\begin{align*}
	&(t-2\alpha/\epsilon)\Pr[\max_{i}\{X_{a_{i}}\} \geq t](1-\epsilon/n^{2})+{{8}\epsilon\cdot\mathcal{R}_{OPT}\over n}\\\geq &(t-2\alpha/\epsilon)\Pr[\max_{i}\{X_{a_{i}}\} \geq t]+{7\epsilon\cdot(2\alpha/\epsilon\cdot\Pr[\max_{i}\{X_{a_{i}}\}\geq 2\alpha/\epsilon])\over n}\\
	\geq & \sum_{i=1}^{m} t_{i} \Pr[X_{a_i}\geq t_{i}].
\end{align*} 
	
	The above imply that the contribution to ${\cal R}_{P'}$ from the items in $S$ is at most an additive ${8}\epsilon\cdot\mathcal{R}_{OPT}\over n$ smaller than the contribution to ${\cal R}_{P}$ from the items in $S$.
	
	We proceed to compare the contributions from the items in $[n]\setminus S$ to ${\cal R}_{P}$ and ${\cal R}_{P'}$. We start with ${\cal R}_{P}$. The contribution from the items in $[n]\setminus S$ is no greater than the total revenue when we ignore the existence of the items in $S$ (e.g. by setting the prices of these items to $+\infty$), since this only boosts the winning probabilities of each item in $[n]\setminus S$.
	
      Under price vector $P'$, $\forall i\in S$, $\Pr[v_i\geq p'_{i}]\leq {\epsilon\over n^3}$ (Theorem~\ref{thm:regextremevalue}). So with probability at least $1-{\epsilon\over n^2}$, no item in $S$ has a positive value-minus-price gap and the item that has the largest positive gap among the items in $[n]-S$ is the item that is bought by the buyer. Hence, by independence the contribution to ${\cal R}_{P'}$ from the items in $[n]-S$ is at least a $1-{\epsilon\over n^{2}}$ fraction of the revenue when the items of $S$ are ignored.

By the above discussion, the contribution to ${\cal R}_{P'}$ from the items in $[n]-S$ is at most an additive $\epsilon \mathcal{R}_{OPT}\over n^2$ smaller than the contribution to ${\cal R}_P$ from the items in $[n]-S$.      

Putting everything together, we get that ${\cal R}_{P'} \ge {\cal R}_P - {{9}\epsilon\mathcal{R}_{OPT}\over n}$.
%
        
		{\bf Second step:} To truncate the lower prices, we invoke Lemma~\ref{lem:trunclow}. This implies that we can set all the prices below $\epsilon\alpha/n^4$ to $\epsilon\alpha/n^4$, only hurting our revenue by an additive $\epsilon\alpha/n^4\leq {c\epsilon\over n}\cdot \max_{z} (z\cdot \Pr[\max_{i}\{X_{{i}}\}\geq z])\leq c\epsilon\mathcal{R}_{OPT}/n$ (where we used Theorem~\ref{thm:regextremevalue} for the first inequality).
	 
	 Hence, we can define $\hat{P}$ as follows: if $p'_{i}\leq \epsilon\alpha/n^{4}$, set $\hat{p}_{i} = \epsilon\alpha/n^{4}$, otherwise set $\hat{p}_{i}=p'_{i}$. It follows from the above that ${\cal R}_{\hat{P}} \ge {\cal R}_{{P}} - {(c+{9})\epsilon\mathcal{R}_{OPT}\over n}$.
	 
\end{proof}

Thus, we have reduced the problem of finding a near-optimal price vector in $[0,+\infty]^n$ to the problem of finding a near-optimal price vector in the set $[\epsilon\alpha/n^4, 2n^2\alpha/\epsilon^{2}]^{n}$.

\subsubsection{Truncating the Support of the Input Regular Distributions}\label{sec:regboundingdistribution}

We show that we can truncate the support of the distributions if the price vectors are restricted. Namely

\begin{lemma}\label{lem:regbounding}
   Given a collection of independent regular random variables $\mathcal{V}=\{v_i\}_{i\in[n]}$ and any $\epsilon\in(0,1)$, let us define a new collection of random variables $\tilde{\mathcal{V}}=\{\tilde{v}_i\}_{i\in[n]}$ via the following coupling:  for all $i\in [n]$, set  $\tilde{v}_i={\epsilon\alpha\over 4n^4}$ if $v_i < {\epsilon\alpha\over 2n^4}$, set $\tilde{v}_i=4n^{4}\alpha/\epsilon^{3}$, if $v_i\ge 4n^{4}\alpha/\epsilon^{3}$, and $\tilde{v}_i=v_i$ otherwise. Also, let $c$ be the absolute constant defined in Theorem~\ref{thm:regextremevalue}. For any price vector $P\in[\epsilon\alpha/n^{4},2n^{2}\alpha/\epsilon^{2}]^{n}$, $|\mathcal{R}_{P}(\mathcal{V})-{\mathcal{R}}_{P}(\tilde{\mathcal{V}})|\leq {c\epsilon\mathcal{R}_{OPT}(\mathcal{V})\over n}$, where $\mathcal{R}_P({\cal V})$ and $\mathcal{R}_P(\tilde{\cal V})$ are respectively the revenues of the seller under price vector $P$ when the values of the buyer are  $\cal V$ and $\tilde{\cal V}$.
\end{lemma}

\begin{proof}
First, let us define another collection of mutually independent random variables $\hat{\mathcal{V}}=\{\hat{v}_{i}\}_{i\in[n]}$ via the following coupling: for all $i\in[n]$  set $\hat{v}_i=4n^{4}\alpha/\epsilon^{3}$ if $v_i\ge 4n^{4}\alpha/\epsilon^{3}$, and set $\hat{v}_i=v_i$ otherwise.

    By Theorem~\ref{thm:regextremevalue}, we know that for every $i$, $\Pr[v_{i}\geq 4n^4\alpha/\epsilon^{3}]\leq {\epsilon^{3}\over 2n^{7}}$.
    Hence, the probability of the event that there exists an $i$ such that $v_i\geq 4n^4\alpha/\epsilon^{3}$ is no greater than $n\times \epsilon^{3}/2n^7=\epsilon^{3}/2n^6$. Thus the difference between the contributions of this event to the revenues  $\mathcal{R}_P({\cal V})$ and $\mathcal{R}_P(\hat{\cal V})$ is no greater than $2n^{2}\alpha/\epsilon^{2}\cdot (\epsilon^{3}/2n^{6}) = {\epsilon \alpha\over n^4}\leq  {c\epsilon \over n}\cdot\max_{z}(z\cdot \Pr[\max_{i}\{v_{i}\}\geq z])\leq {c\epsilon\mathcal{R}_{OPT}\over n}$, given that the largest price is at most $2n^{2}\alpha/\epsilon^{2}$.

    Now let us consider the event: {$v_{i}\leq 4n^4\alpha/\epsilon^{3}$}, for all $i$. In this case $\hat{v}_{i}=v_{i}$ for all $i$. So the contribution of this event  to the revenues  $\mathcal{R}_P({\cal V})$ and $\mathcal{R}_P(\hat{\cal V})$ is the same.
    
    Thus, $|\mathcal{R}_P({\cal V})-\mathcal{R}_P(\hat{\cal V})|\leq {c\epsilon\mathcal{R}_{OPT}\over n}$.

Now it follows from Lemma~\ref{lem:boundinglow} that the seller's revenue under any price vector in $[\epsilon\alpha/n^4, 2n^2\alpha/\epsilon^2]^{n}$ is the same when the buyer's value distributions are $\hat{\mathcal{V}}$ and $\tilde{\mathcal{V}}$. 
\end{proof}
The above lemma shows that we can reduce the problem of finding a near-optimal price vector in  $[\epsilon\alpha/n^4, 2n^2\alpha/\epsilon^{2}]^{n}$ for the original  value distributions $\cal V$ to the problem of finding a near-optimal price vector in the set $[\epsilon\alpha/n^4, 2n^2\alpha/\epsilon^{2}]^{n}$ for a collection of value distributions $\tilde{\cal V}$ supported on the set $[{\epsilon\alpha\over 4n^4},4n^4\alpha/\epsilon^{3}]$. Next, we establish that the latter problem can be reduced to finding any (i.e. not necessarily restricted) near-optimal price vector for the distributions $\tilde{\cal V}$.


\begin{lemma}\label{lem:regboundedrestrictprice}
Given a collection of independent regular random variables $\mathcal{V}=\{v_i\}_{i\in[n]}$ and any $\epsilon\in(0,1)$, let us define a new collection of random variables $\tilde{\mathcal{V}}=\{\tilde{v}_i\}_{i\in[n]}$ via the following coupling:  for all $i\in [n]$, set  $\tilde{v}_i={\epsilon\alpha\over 4n^4}$ if $v_i < {\epsilon\alpha\over 2n^4}$, set $\tilde{v}_i=4n^{4}\alpha/\epsilon^{3}$ if $v_i\ge 4n^{4}\alpha/\epsilon^{3}$, and set $\tilde{v}_i=v_i$ otherwise. Let also $c$ be the absolute constant defined in Theorem~\ref{thm:regextremevalue}. For any price vector $P$, we can efficiently construct a new price vector $\hat{P}\in [\epsilon\alpha/n^{4},2n^{2}\alpha/\epsilon^{2}]^{n}$, such that $\mathcal{R}_{\hat{P}}(\tilde{\mathcal{V}})\geq \mathcal{R}_{P}(\tilde{\mathcal{V}})-{(c+{9})\epsilon\cdot\mathcal{R}_{OPT}(\tilde{\mathcal{V}})\over n}$.
\end{lemma}
The proof is essentially the same as the proof of Lemma~\ref{lem:regrestrictprice} and we skip it. Combining Lemmas~\ref{lem:regrestrictprice},~\ref{lem:regbounding} and~\ref{lem:regboundedrestrictprice} we obtain Theorem~\ref{thm:regreduction}. The proof is given in the next appendix.


\subsubsection{Finishing the Reduction} \label{app: finish regular reduction}
\begin{prevproof}{Theorem}{thm:regreduction}
We start with computing $\alpha$. This can be done efficiently as specified in the statement of Theorem~\ref{thm:regextremevalue}. Now let us define $\tilde{\mathcal{V}}$ via the following coupling: for all $i\in [n]$, set    $\tilde{v}_i={\epsilon\alpha\over 4n^4}$ if $v_i < {\epsilon\alpha\over 2n^4}$, set $\tilde{v}_i=4n^{4}\alpha/\epsilon^{3}$ if $v_i\ge 4n^{4}\alpha/\epsilon^{3}$, and set $\tilde{v}_i=v_i$ otherwise. \notshow{\yangnote{Remove: It is not hard to see that the distributions of the $\tilde{v}_i$'s can be computed in time polynomial in $n$, $1/\epsilon$ and the description complexity of the distributions of the $v_i$'s, if these are given to us explicitly. If we have oracle access to these distributions,  we can construct oracles for the distributions of the $\tilde{v}_i$'s that run in time polynomial in $n$, $1/\epsilon$ and the desired oracle accuracy.}}

Now let $P$ be a price vector such that ${\mathcal{R}}_{P}(\tilde{\mathcal{V}})\geq (1-\epsilon+{(4c+{19})\epsilon\over n})\cdot {\mathcal{R}}_{OPT}(\tilde{\mathcal{V}})$. It follows from Lemma~\ref{lem:regboundedrestrictprice} that we can efficiently construct 
a  price vector $P'\in[\epsilon\alpha/n^4, 2n^2\alpha/\epsilon^2]^{n}$, such that $${\mathcal{R}}_{P'}(\tilde{\mathcal{V}})\geq\left(1-\epsilon+{(4c+{19})\epsilon\over n}\right)\cdot {\mathcal{R}}_{OPT}(\tilde{\mathcal{V}})-{(c+{9})\epsilon\over n}{\mathcal{R}}_{OPT}(\tilde{\mathcal{V}})\geq \left(1-\epsilon+{(3c+{10})\epsilon\over n}\right)\cdot {\mathcal{R}}_{OPT}(\tilde{\mathcal{V}}).$$

Lemma~\ref{lem:regrestrictprice} implies that there exists a price vector $\hat{P}\in [\epsilon \alpha/n^4, 2n^2\alpha/\epsilon^2]^{n}$, such that $\mathcal{R}_{\hat{P}}(\mathcal{V})\geq \big{(}1-{(c+{9})\epsilon\over n}\big{)}\cdot \mathcal{R}_{OPT}(\mathcal{V})$. By Lemma~\ref{lem:regbounding}, we know that $${\mathcal{R}}_{OPT}(\tilde{\mathcal{V}})\geq {\mathcal{R}}_{\hat{P}}(\tilde{\mathcal{V}})\geq\mathcal{R}_{\hat{P}}(\mathcal{V})-{c\epsilon\over n}\mathcal{R}_{OPT}(\mathcal{V})\geq \left(1-{(2c+{9})\epsilon\over n}\right) \cdot \mathcal{R}_{OPT}(\mathcal{V}).$$

So ${\mathcal{R}}_{P'}(\tilde{\mathcal{V}})\geq(1-\epsilon+{c\epsilon\over n})\cdot \mathcal{R}_{OPT}(\mathcal{V})$. We can now apply Lemma~\ref{lem:regbounding} again, and get $$\mathcal{R}_{P'}(\mathcal{V})\geq{\mathcal{R}}_{P'}(\tilde{\mathcal{V}})-{c\epsilon\over n}\mathcal{R}_{OPT}(\mathcal{V}) \geq(1-\epsilon)\cdot \mathcal{R}_{OPT}(\mathcal{V}).$$\end{prevproof}

\section{Algorithmic Results for MHR and Regular Distributions} \label{sec:app overall ptas} \label{sec:overall ptas}
The proofs of Theorems~\ref{thm:ptas mhr} and~\ref{thm:quasi ptas regular} follow immediately from Theorem~\ref{thm:general algorithm} using our reductions to bounded distributions (Theorems~\ref{thm:reduction MHR to balanced} and~\ref{thm:regreduction} of Sections~\ref{sec:truncate} and~\ref{sec:regular} respectively). We restate the theorems and prove them.

\begin{varthm}{{\bf \ref{thm:ptas mhr} [Restated]}} 
    Suppose we are given a collection of MHR distributions $\{F_i\}_{i \in [n]}$. Then, for any constant $\epsilon >0$, there is an algorithm that runs in time polynomial in the input and $n^{{1 \over \epsilon^{7}}}$  and computes a price vector $P$ such that $$\mathcal{R}_{P}\geq(1-\epsilon)\mathcal{R}_{OPT},$$ where $\mathcal{R}_{P}$ is the expected revenue under price vector $P$ when the buyer's values for the items are independently distributed according to the distributions $\{F\}_i$ and $\mathcal{R}_{OPT}$ is the revenue achieved by the optimal price vector.
\end{varthm}
\begin{prevproof}{Theorem}{thm:ptas mhr}
We apply Theorem~\ref{thm:reduction MHR to balanced} to reduce the item pricing problem for MHR distributions to the item pricing problem for bounded distributions. Then we use our algorithm from Theorem~\ref{thm:general algorithm} for bounded distributions. The resulting running time is polynomial in the input and ${n^{\log^{4} {1 \over \epsilon} \over  {\epsilon^8}}}$. Being a bit more careful in the application of our discretization lemmas we obtain running time polynomial in the input and {$n^{1/ \epsilon^{7}}$}. Recall that in the algorithm of Theorem~\ref{thm:general algorithm} we employed the reduction of Theorem~\ref{thm:discretization} to discretize supports and prices into sets of bounded cardinalities. To establish our reduction in Theorem~\ref{thm:discretization} we employed Lemma~\ref{lem:horizontal-discretization}, which in turn made use of Lemma~\ref{lem:nisan-value}, where we set $a={2 \over 3}$. Setting instead $a \approx {1 \over 2}$ would result in a different tradeoff of parameters, improving our running time here. 
\end{prevproof}

\begin{varthm}{{\bf \ref{thm:quasi ptas regular} [Restated]}} 
    Suppose we are given a collection of regular distributions $\{F_i\}_{i \in [n]}$. Then, for any constant $\epsilon >0$, there exists an algorithm that runs in time polynomial in the input and {$\max\left\{ n^{\log^{11} {n\over \epsilon} \cdot  \log \log {n \over \epsilon}},n^{ {\log^3 {n\over\epsilon} \cdot  \log {1\over\epsilon} \over {\epsilon}^{8}}} \right\}$}\notshow{\costasnote{$n^{{\log^3 n\over \epsilon^{9}}}$}}  and computes a price vector $P$ such that $$\mathcal{R}_{P}\geq(1-\epsilon)\mathcal{R}_{OPT},$$ where $\mathcal{R}_{P}$ is the expected revenue under price vector $P$ when the buyer's values for the items are independently distributed according to the distributions $\{F\}_i$ and $\mathcal{R}_{OPT}$ is the revenue achieved by the optimal price vector.
\end{varthm}

\begin{prevproof}{Theorem}{thm:quasi ptas regular}
We apply Theorem~\ref{thm:regreduction} to reduce the item pricing problem for regular distributions to the item pricing problem for bounded distributions. Then we use our algorithm from Theorem~\ref{thm:general algorithm} for bounded distributions. The resulting running time is polynomial in the input and {$\max\left\{ n^{\log^{11} {n\over \epsilon} \cdot  \log \log{n \over\epsilon}},n^{ {\log^3 {n\over\epsilon} \cdot  \log {1\over\epsilon} \over {\epsilon}^{8}}} \right\}$}\notshow{$\costasnote{n^{\log^{3} {n} \cdot \log^4 {1 \over \epsilon} \over  {\epsilon^8}}}$}. 
\end{prevproof}

\section{Proofs of Structural Results}\label{app:structural proofs} \label{sec:proofs of structural}


{\begin{prevproof}{Theorem}{thm:single price constant factor}
Let $\beta$ be the anchoring point of Theorem~\ref{thm:extreme MHR}. It follows from the properties of the anchoring point that pricing all the items at price $\beta/2$ achieves revenue
$${\beta \over 2} \cdot \Pr[ \max_i\{X_i\} \ge \beta/2] \ge {\beta \over 2} \cdot \left(1-{1\over \sqrt{e}}\right)= \beta \cdot c_1,$$
where $c_1 = {1\over 2}\left(1-{1\over \sqrt{e}}\right).$
On the other hand, Lemma~\ref{lem:trunchighprice} shows that the optimal revenue is upper bounded by
$$\beta \cdot \min_{\epsilon \in (0,{1\over 4})}{\left(2\log_2 \frac{1}{\epsilon} + {c_2(\epsilon)} \right)},$$
where $c_2(\epsilon)=36 \epsilon \log_2 ({1\over \epsilon})$. So pricing all items at $\beta/2$ achieves a constant factor approximation to the optimal revenue.
\end{prevproof}}

\begin{prevproof}{Theorem}{thm:constant prices suffice}
Suppose that the buyer's values are $\{{v}_i\}_{i\in[n]}$ where the $v_i$'s are mutually independent, MHR random variables. We can apply Lemma~\ref{lem:boundingprice} to restrict the price-vectors to $[\epsilon\cdot\beta, 2\log_2(\frac{1}{\epsilon})\cdot\beta]^{n}$, where $\beta = \beta(\{v_i\}_i)$ is the anchoring point of Theorem~\ref{thm:extreme MHR} computed from the distributions of the variables $\{v_i\}_i$. The loss in revenue from this restriction is bounded by Lemma~\ref{lem:boundingprice}. Having this restriction in place, we may now modify the variables $\{{v}_i\}_{i\in[n]}$ into a new collection of random variables $\{\tilde{v}_i\}_i$ as follows: for all $i\in [n]$, set  $\tilde{v}_i={\epsilon \over 2}\cdot \beta$ if $v_i < \epsilon \cdot \beta$, set $\tilde{v}_i=2\log_2 (\frac{1}{\epsilon})\cdot\beta$ if $v_i\ge2\log_2 (\frac{1}{\epsilon})\cdot\beta$, and set $\tilde{v}_i=v_i$ otherwise. Lemmas~\ref{lem:boundinghigh} and~\ref{lem:boundinglow} show that the expected revenue of any price vector $P\in [\epsilon\cdot\beta, 2\log_2(\frac{1}{\epsilon})\cdot\beta]^{n}$  is approximately the same for $\{{v}_i\}_{i\in[n]}$ and for $\{\tilde{v}_i\}_i$. Now we can apply Lemma~\ref{cor:discreteprice} to discretize $[\epsilon\cdot\beta, 2\log_2(\frac{1}{\epsilon})\cdot\beta]^n$. The chain of reductions we used guarantees that a nearly-optimal among discretized prize-vectors for  $\{\tilde{v}_i\}_i$ is also nearly-optimal among all possible price-vectors for $\{{v}_i\}_i$.
\end{prevproof}

\begin{prevproof}{Theorem}{thm:logn prices suffice}
Suppose that the buyer's values are $\{{v}_i\}_{i\in[n]}$ where the $v_i$'s are mutually independent, regular random variables. We can apply Lemma~\ref{lem:regrestrictprice} to restrict the price-vectors to  $[\epsilon\alpha/n^{4},2n^{2}\alpha/\epsilon^{2}]^{n}$ where $\alpha$ is chosen as in Appendix~\ref{sec:regextremevalue}. The loss in revenue from this restriction is bounded by Lemma~\ref{lem:regrestrictprice}. Having this restriction in place, we may now modify the variables $\{{v}_i\}_{i\in[n]}$ into a new collection of random variables $\{\tilde{v}_i\}_i$ as follows: for all $i\in [n]$, set  $\tilde{v}_i={\epsilon\alpha\over 4n^4}$ if $v_i < {\epsilon\alpha\over 2n^4}$, set $\tilde{v}_i=4n^{4}\alpha/\epsilon^{3}$, if $v_i\ge 4n^{4}\alpha/\epsilon^{3}$, and $\tilde{v}_i=v_i$ otherwise. Lemma~\ref{lem:regbounding} shows that the expected revenue of any price vector $P\in [\epsilon\alpha/n^{4},2n^{2}\alpha/\epsilon^{2}]^{n}$ is approximately the same for $\{{v}_i\}_{i\in[n]}$ and for $\{\tilde{v}_i\}_i$. Now we can apply Lemma~\ref{cor:discreteprice} to discretize $[\epsilon\alpha/n^{4},2n^{2}\alpha/\epsilon^{2}]^{n}$. The chain of reductions we used guarantees that  a nearly-optimal among discretized prize-vectors for  $\{\tilde{v}_i\}_i$ is also nearly-optimal among all possible price-vectors for $\{{v}_i\}_i$.
\end{prevproof}}

\subsection{Proof of Theorem~\ref{thm:structural theorem 3}: A Single Price Suffices for I.I.D. MHR Distributions} \label{sec:MHR iid}

We improve the running time of Theorem~\ref{thm:ptas mhr} for when the buyer's values are i.i.d.~according to some MHR distribution. The main technical idea that goes into the algorithm is establishing our structural result for i.i.d.~MHR distributions described by Theorem~\ref{thm:structural theorem 3}. In particular, we show that, if the number of items is a sufficiently large function of $1/\epsilon$, then using a single price suffices to get an $(1-\epsilon)$-fraction of the optimal revenue. Theorem~\ref{thm:alg} below summarizes the improvement on the running time as well as the structural result for i.i.d. MHR distributions.

We proceed to the details of our algorithm. To simplify our notation, let us assume that all the $v_i$'s are independent copies of the random variable $v$, and denote the cumulative distribution function of $v$ by $F$. Moreover, let $\alpha_n = \inf\left\{x|F(x)\geq 1-\frac{1}{n}\right\}$ (as in Definition~\ref{def:alpha}). We start by showing an analogue of Lemma~\ref{lem:tinycontribution}.
\begin{lemma}\label{lem:iidsmallcontribution}
    If $S=Con[v\geq (1+\epsilon)\alpha_{n}]$, then $S\leq \frac{6(1+\epsilon)\alpha_n}{n^{1+\epsilon}}.$
\end{lemma}

\begin{proof}
    By Lemma~\ref{lem:concentrate}, we know that $(1+\epsilon)\alpha_{n}\geq \alpha_{n^{1+\epsilon}}$. Thus, $ S\leq Con[v\geq \alpha_{n^{1+\epsilon}}].$ But Lemma~\ref{lem:contribution} gives $Con[v\geq \alpha_{n^{1+\epsilon}}]\leq 6\alpha_{n^{1+\epsilon}}/n^{1+\epsilon}$. Hence, 
    $$S\leq \frac{6\alpha_{n^{1+\epsilon}}}{n^{1+\epsilon}}\leq \frac{6(1+\epsilon)\alpha_{n}}{n^{1+\epsilon}}.$$
\end{proof}

Using Lemma~\ref{lem:iidsmallcontribution} and Lemma~\ref{lem:trunchigh}, we deduce that if we constrain our prices to be $\le (1+\epsilon)\alpha_{n}$, we lose no more than $\frac{6(1+\epsilon)\alpha_n}{n^{\epsilon}}$ revenue. Given that the optimal revenue with the restriction that all prices be $\le (1+\epsilon)\alpha_{n}$ is at most $(1+\epsilon)\alpha_{n}$, it follows that the optimal revenue without the restriction is at most $(1+\epsilon)\alpha_{n}+\frac{6(1+\epsilon)\alpha_n}{n^{\epsilon}} = (1+\epsilon)(1+\frac{6}{n^{\epsilon}})\alpha_n$. This is very close to $\alpha_n$ if $n$ is a sufficiently large function of $\epsilon$. If that's the case, it suffices to find a price vector achieving revenue close to $\alpha_n$.

\begin{lemma}\label{lem:iidalmostopt}
    If we use the price vector $P=((1-\epsilon)\alpha_n,(1-\epsilon)\alpha_n,\ldots,(1-\epsilon)\alpha_n)$, we receive revenue at least $ \left(1-e^{(-n^{\epsilon})}-\epsilon\right)\alpha_{n}$.
\end{lemma}
\begin{proof}
Let $p = (1-\epsilon)\alpha_n$.
    By Lemma~\ref{lem:concentrate}, we know that $\frac{\alpha_{n^{1-\epsilon}}}{(1-\epsilon)}\geq \alpha_n$. Hence, for all $i$, $$\Pr[v_i < p] \leq \Pr[v_i < \alpha_{n^{1-\epsilon}}] \le1-\frac{1}{n^{1-\epsilon}}.$$
  It follows that  $$\Pr[\exists i, v_i\geq p]\geq 1-\left(1-\frac{1}{n^{1-\epsilon}}\right)^n\geq 1- e^{(-n^{\epsilon})}.$$
    
    Hence, with probability at least $1- e^{(-n^{\epsilon})}$, the buyer will purchase an item and will pay $p$. Hence, the revenue is at least $(1- e^{(-n^{\epsilon})})(1-\epsilon)\alpha_n\geq (1-e^{(-n^{\epsilon})}-\epsilon)\alpha_{n}$.
\end{proof}

Notice that, when $n\geq (1/\epsilon)^{1/\epsilon}$, $n^{\epsilon}\geq 1/\epsilon$. In this case, we have shown that $OPT\leq (1+\epsilon)(1+6\epsilon)\alpha_{n}\leq(1+8\epsilon)\alpha_{n}$. On the other hand, Lemma~\ref{lem:iidalmostopt}, says that we can achieve revenue at least $(1-\frac{1}{e^{1/\epsilon}}-\epsilon)\alpha_{n}$ using a single price. Since $e^{1/\epsilon}\geq 1/\epsilon$, this revenue is at least $(1-2\epsilon)\alpha_{n}$. Given that $(1+8\epsilon)(1-10\epsilon)\leq (1-2\epsilon)$, we have $(1-2\epsilon)\alpha_{n}\geq (1-10\epsilon)OPT$. So if we set the price for every item to be $(1-\epsilon)\alpha_{n}$, we achieve a revenue that is at least $(1-10\epsilon)OPT$.

\begin{theorem}\label{thm:alg}
If the values of the buyer are i.i.d. according to a MHR distribution, there is a PTAS for finding a price vector that achieves a $(1-\epsilon)$-fraction of the optimal revenue. The algorithm runs in time polynomial in $\log({\log n\over \epsilon})$, {$2^{{\log (1/\epsilon)\over\epsilon^{8}}}$} and the size of the input. Moreover, if $n\geq (12/\epsilon)^{12/\epsilon}$, there exists an efficiently computable price such that, if all items are priced at this price, the resulting revenue is at least { $(1-\epsilon)OPT$}.
\end{theorem}
\begin{proof}
    {Let $\epsilon'=\epsilon/12$.} Depending on the value of $n$ our algorithm proceeds in one of the following ways:
    \begin{itemize}
        {\item If $n\geq (1/\epsilon')^{1/\epsilon'}$, we do binary search starting at an anchoring point of the distribution (see Appendix~\ref{sec:model}) to find some $p\in [1-\epsilon',1+\epsilon']\alpha_n$. {This takes time polynomial in $O(\log({\log n\over \epsilon'}))$ and  the size of the input, since $\alpha_n \le \alpha_2 \cdot \log_2 n$.} We then set every item's price to $(1-2\epsilon')p$. Since $(1-2\epsilon')p\leq (1-\epsilon')\alpha_{n}$, 
    $$\Pr[\exists\ i,\ v_{i}\geq (1-2\epsilon')p]\geq \Pr[\exists\ i,\ v_{i}\geq (1-\epsilon')\alpha_{n}].$$
     On the other hand, $(1-2\epsilon')p\geq (1-2\epsilon')(1-\epsilon')\alpha_{n}$. Thus, the revenue we obtain if we price all items at $(1-2\epsilon')p$ is at least $(1-2\epsilon')$ times the revenue under price vector $P=((1-\epsilon')\alpha_n,(1-\epsilon')\alpha_n,\ldots,(1-\epsilon')\alpha_n)$. Hence, the revenue is at least $(1-12\epsilon')OPT=(1-\epsilon)OPT$. 

    \item If $n<(1/\epsilon')^{1/\epsilon'}$, we simply use the algorithm for the non-i.i.d. case (Theorem~\ref{thm:ptas mhr}).}
    
    
    
     
 \end{itemize}
\end{proof}

%

\section{An interesting example} \label{sec:counter}

A natural property than one would expect to hold  is that, when the value distributions are discrete,  there  always exists an optimal solution that uses prices from the support of the value distributions. It turns out that this is not true. Here is an example:

Suppose that the seller has two items to sell, and the buyer's values for the items are $v_{1}$, which is uniform on $\{1, 5\}$, and $v_{2}$, which is uniform on $\{3,3.5\}$. Moreover, assume that, if there is a tie between the value-minus-price gap for the two items, the buyer tie-breaks in favor of item $1$. We claim that in this case the price vector $P=(4.5,3)$ achieves higher revenue than any price vector that uses prices from the set $\{1,3,3.5,5\}$ (where the values are drawn from.) Let us do the calculation. All our calculations are written in the form $$\mathcal{R}_{P} = p_{1}\times \Pr[item\ 1\ is\ the\ winner] +p_{2}\times\Pr[item\ 2\ is\ the\ winner].$$
\begin{enumerate}
\item When $P=(4.5,3)$

$\mathcal{R}_{P} = 4.5\times(1/2\times 1)+3\times(1/2\times 1)= 30/8$
\item When $P \in \{1,3,3.5,5\}^2$:
\begin{itemize}
\item If $P=(5,3.5)$ then

$\mathcal{R}_{P} = 5\times(1/2\times 1)+ 3.5\times(1/2\times 1/2)= 27/8<30/8$

\item If $P=(5,3)$ then

$\mathcal{R}_{P} = 5\times(1/2\times 1/2)+3\times(1\times 1/2 + 1/2\times 1/2) = 28/8<30/8$

\item For any other price vector, the maximum revenue is bounded by $3.5=28/8< 30/8$.
\end{itemize}
\end{enumerate}


\bibliographystyle{alpha}
\bibliography{costasbib}
\appendix
\end{document}